\documentclass[pmlr]{jmlr}  

\RequirePackage{graphicx}
\usepackage{booktabs}
\usepackage{longtable}  
 
\usepackage{graphicx}  
\usepackage{hyperref}
\usepackage{url}
\usepackage{pifont}
\usepackage{amsmath}
\usepackage{yfonts}
\usepackage{amssymb}
\usepackage{wrapfig}
\usepackage{mathtools}
\usepackage[dvipsnames,table]{xcolor}
\usepackage{dsfont}
\usepackage{derivative}  
\usepackage{algpseudocode}
\usepackage{algorithm}
\usepackage{mathtools}
\usepackage{enumitem}
\usepackage{multirow}

\usepackage{amsfonts}
\usepackage{bm}
\usepackage{bbm} 
\usepackage{framed}
\usepackage{makecell}
\usepackage{adjustbox}
\usepackage{csquotes}

\newcommand\independent{\protect\mathpalette{\protect\independenT}{\perp}}
\def\independenT#1#2{\mathrel{\rlap{$#1#2$}\mkern2mu{#1#2}}}

\DeclareMathOperator*{\sign}{sign}
\DeclareMathOperator*{\indic}{\mathds{1}}

\DeclareMathOperator*{\E}{\mathbb{E}}
\newcommand{\X}{\mathbf{X}}
\newcommand{\x}{\mathbf{x}}
\newcommand{\U}{\mathbf{U}}
\newcommand{\unc}{\mathbf{u}}
\DeclareMathOperator*{\Prob}{\mathbb{P}}
\newcommand{\Ttreat}{T^{(1)}}
\newcommand{\Ctreat}{C^{(1)}}
\newcommand{\Tcont}{T^{(0)}}
\newcommand{\Ccont}{C^{(0)}}
\newcommand{\YI}{Y_{\mathrm{I}, t, 1}}
\newcommand{\YII}{Y_{\mathrm{II}, t, 1}}
\newcommand{\Yrmst}{Y_{\mathrm{RMST}, \tau, 1}}
\newcommand{\taumax}{\tau_\mathrm{max}}

\theorembodyfont{\upshape}
\theoremheaderfont{\scshape}
\theorempostheader{:}
\theoremsep{\newline}

\newtheorem{assumption}[theorem]{Assumption}

\jmlrvolume{340}
\jmlryear{2026}
\jmlrworkshop{Machine Learning for Healthcare}

\title[DVDS sensitivity analysis to unobserved confounding for survival outcomes]{Doubly valid and doubly sharp sensitivity analysis to unobserved confounding for survival outcomes}

\author{\Name{Jean-Baptiste Baitairian}
       \Email{jean-baptiste.baitairian@inria.fr}\\ 
       \addr Sanofi R\&D, Gentilly, France \\
       Inria, Inserm, Université Paris Cité, HeKA, F-75015 Paris, France
       \AND
       \Name{Bernard Sebastien}
       \Email{bernard.sebastien@sanofi.com}\\ 
       \addr Sanofi R\&D, Gentilly, France
       \AND
       \Name{Rana Jreich}
       \Email{rana.jreich@sanofi.com}\\ 
       \addr Sanofi R\&D, Gentilly, France
       \AND
       \Name{Sandrine Katsahian}
       \Email{sandrine.katsahian@aphp.fr}\\ 
       \addr Inria, Inserm, Université Paris Cité, HeKA, F-75015 Paris, France \\
       CIC-EC 1418 - Paris HEGP, Paris, France
       \AND
       \Name{Agathe Guilloux}
       \Email{agathe.guilloux@inria.fr}\\ 
       \addr Inria, Inserm, Université Paris Cité, HeKA, F-75015 Paris, France
       }

\begin{document}

\maketitle

\begin{abstract}
    Time-to-event outcomes are central in oncology and rare diseases, where treatment effects are often summarized by differences in survival curves or Restricted Mean Survival Time (RMST). In real-world data, estimating these causal effects relies on the absence of unobserved confounding, an assumption that is rarely satisfied. We develop a sensitivity analysis framework for causal treatment effects with survival outcomes under the Marginal Sensitivity Model (MSM). We introduce doubly valid and doubly sharp (DVDS) bounds for differences in survival functions and RMST, extending recent DVDS results to the time-to-event setting while accounting for informative censoring. In practice, our method yields tighter bounds and improved computational efficiency compared to a previous approach from the literature, on simulated and real data. For tractability, we assume independence between censoring and unobserved confounding, a limit that should be addressed in future works.
\end{abstract}

\section{Introduction}

Observational data, more generally referred to as real-world data (RWD), offer substantial opportunities to enhance clinical drug development and inform regulatory decision-making. Such data sources include Electronic Health Records (EHR), administrative claims databases, and, increasingly, measurements from wearable devices. In this context, the \textit{21st Century Cures Act} mandates the \textit{U.S. Food and Drug Administration} (FDA) to evaluate the potential use of real-world evidence (RWE) in the approval of medical products \citep{fang2025sensitivity}. 

Time-to-event outcomes are central endpoints in many clinical trials, particularly in oncology and rare diseases. Overall Survival (OS) is often regarded as the gold-standard endpoint; however, its long follow-up time motivates the use of surrogate endpoints such as Progression-Free Survival (PFS) or Disease-Free Survival (DFS) \citep{fda2018clinicaltrial, delgado2021clinical}. 

From these survival endpoints, several treatment effect measures can be defined. The hazard ratio remains a widely used summary measure \citep{barraclough2011biostatistics}, yet its causal interpretation is limited~\citep{lee2024sensitivity, voinot2025treatment}. As a result, alternative estimands have been advocated, including differences in survival functions and Restricted Mean Survival Time (RMST) \citep{royston2013restricted}, which provide more causally interpretable treatment effects over a fixed time horizon.

Estimating such treatment effects from RWD raises fundamental challenges due to confounding and selection bias, as treatment assignment is no longer randomized. These issues are typically addressed within the framework of causal inference. In particular, causal estimands such as survival curve or RMST differences can be identified under the \textit{ignorability} (or \textit{unconfoundedness}) assumption \citep{rubin1974estimating, hirano2004propensity}, which posits that all confounders of the treatment–outcome relationship are observed. In practice, this assumption is neither testable nor plausible in many clinical settings, especially when relying on routinely collected data.

Sensitivity analysis provides an approach to assess the impact of unobserved confounding by relaxing point identification to partial identification. Rather than returning a single estimate, these methods characterize a set of plausible values for the causal effect under a sensitivity model that quantifies the magnitude of hidden bias. Among such frameworks, the \textit{Marginal Sensitivity Model} (MSM) \citep{tan2006distributional} has emerged as a flexible and promising approach.

While sensitivity analysis under the MSM is now well developed for binary and continuous outcomes, extensions to causal survival estimands remain limited. In particular, accommodating both robust sensitivity bounds and informative censoring has received little attention in the literature. We address this gap by developing sensitivity bounds for survival function and RMST that inherit strong robustness properties while remaining computationally tractable.

\paragraph{Contributions.}
We build upon the DVDS framework of \citet{dorn2024doubly} to derive novel sensitivity bounds for survival outcomes, including survival functions and RMST. Our approach relies on two survival function estimators proposed by \citet{cheng2022addressing}, which admit closed-form expressions. This enables us to obtain analytical bounds that inherit the double validity and double sharpness properties.

In contrast to \citet{lee2024sensitivity}, whose bounds are defined through optimization problems, we give analytical expressions of our estimators, resulting in improved computational efficiency. Moreover, we work under the assumption of conditionally independent censoring (or informative censoring), which is weaker and more realistic than the independent censoring assumption adopted in this prior work.

Through simulations and real-data applications, we demonstrate that our bounds achieve improved robustness and validity properties, while significantly reducing computational cost, particularly in the case of survival curves.

\paragraph{Outline.}
In Section~\ref{sec:methods}, we introduce notations and briefly review causal survival analysis under the ignorability assumption. Then, we present our novel sensitivity bounds for (difference of) survival functions; extensions to the RMST are provided in the appendix. In Section~\ref{sec:experiments_chap4}, we compare our method with \citet{lee2024sensitivity} on simulated data and illustrate its performance on the Right Heart Catheterization (RHC) and German Breast Cancer Study Group (GBCSG) datasets. Section~\ref{sec:conclusion} concludes with a discussion and future directions.

\subsection*{Generalizable Insights about Machine Learning in the Context of Healthcare}

\begin{itemize}
    \item In observational healthcare data, ignorability is often unrealistic. Partial identification and sensitivity analysis provide a principled way to quantify uncertainty due to hidden confounding.
    
    \item Robustness properties (e.g., double robustness and sharpness) are essential when combining machine learning with causal inference under model misspecification.
    
    \item Time-to-event settings require to handle censoring, which is a key source of bias in survival analysis.
    
    \item Computationally efficient methods are critical for large-scale healthcare data, especially when sensitivity analyses are required for regulatory use of real-world evidence.
\end{itemize}

\section{Related Work}

\paragraph{Sensitivity Analysis under Unobserved Confounding.}
In observational studies, the ignorability assumption is generally not testable and may be violated due to unobserved confounders. Sensitivity analysis addresses this limitation by replacing point identification with partial identification. Instead of returning a single estimate, these methods characterize a \textit{partially identified set} of causal effects compatible with a sensitivity model that quantifies the degree of departure from ignorability through a sensitivity parameter.

A framework for such analyses is the \textit{Marginal Sensitivity Model} (MSM) introduced by \citet{tan2006distributional}. Within this framework, \citet{zhao2019sensitivity} proposed Inverse-Probability-Weighting-based (IPW) \citep{hirano2003efficient} sensitivity bounds for Average Treatment Effects (ATE) under unobserved confounding in the case of binary treatments. Another previous model, \textit{Rosenbaum's Sensitivity Model} (RSM) \citep{rosenbaum2002covariance}, is more appropriate when dealing with matched cohorts.

\paragraph{Sharp and Doubly Robust Sensitivity Bounds.}
Subsequent work has focused on improving the statistical properties of the bounds proposed by \citet{zhao2019sensitivity}. In particular, \citet{dorn2022sharp} introduced \textit{sharp} sensitivity bounds, ensuring that the limiting bounds coincide with the set of causal effects compatible with the sensitivity model and no others. More recently, \citet{dorn2024doubly} proposed \textit{doubly valid and doubly sharp} (DVDS) bounds, which combine propensity score and outcome regression estimators in the same way as the \textit{Augmented Inverse Probability Weighting} (AIPW) estimator \citep{robins1994estimation}.

These bounds enjoy several desirable robustness properties. \textit{Double sharpness} ensures that the bounds remain consistent and sharp when either the propensity score model or the outcome regression model is correctly specified, together with an additional nuisance parameter. \textit{Double validity} guarantees valid coverage of the true treatment effect even when this additional nuisance parameter is misspecified, at the cost of potentially conservative bounds. Related refinements have also been explored by \citet{tan2024model}.

\paragraph{Extensions to Survival Outcomes.}
Compared with binary or continuous outcomes, sensitivity analysis for time-to-event data remains relatively underdeveloped. A recent contribution by \citet{lee2024sensitivity} extended the MSM-based sensitivity analysis to survival outcomes, with a focus on differences in Restricted Mean Survival Time (RMST). While this work represents an important step toward handling survival estimands under unobserved confounding, the resulting bounds are defined through optimization procedures and rely on stronger censoring assumptions.

Moreover, estimations are biased when right censoring occurs. Right censoring can be corrected using weighting techniques such as the \textit{Inverse Probability of Censoring Weighting} (IPCW) approach \citep{koul1981regression} or the partial likelihood \citep{cox1972regression}.

\paragraph{Positioning of our Work.}
Our work builds on these developments by extending DVDS sensitivity bounds to causal survival estimands. In particular, we derive analytical bounds for survival functions and RMST that inherit strong robustness properties while accommodating conditionally independent censoring, thereby addressing limitations of existing approaches.

\section{Methods} \label{sec:methods}

\subsection{Notations and Estimands} \label{sec:notations}

We denote by $\mathbf{X} \in \mathcal{X} \subseteq \mathbb{R}^{p_\mathbf{X}}$ ($p_\mathbf{X} \geq 1$) the vector of observed confounders, and by $A \in \{0, 1\}$ the binary treatment (or exposure). Let $T \in \mathbb{R}_+$ be the event time and $C \in \mathbb{R}_+$ the censoring time. The observed time is $\tilde{T} = \min(T, C) = T \wedge C$, and we define $\Delta = \mathds{1}(T \leq C)$, indicating whether the event is observed ($\Delta = 1$) or censored ($\Delta = 0$). Since observational data typically involve hidden variables, we assume the existence of unobserved confounders $\U \in \mathcal{U} \subseteq \mathbb{R}^{p_\U}$ ($p_\U \geq 1$).

We adopt the \textit{Neyman--Rubin potential outcomes framework}~\citep{rubin1974estimating, neyman1990application}, adapted to survival data as in \citet{lee2024sensitivity} or \citet{voinot2025treatment}. We denote by $T^{(a)}$ and $C^{(a)}$ the potential event and censoring times under treatment $a \in \{0, 1\}$. The following SUTVA assumption links observed and potential outcomes.

\begin{assumption}[SUTVA for survival outcomes] \label{ass:sutva_survival}
For $a \in \{0, 1\}$, if $A = a$, then $T = T^{(a)}$ and $C = C^{(a)}$ (consistency), and there is no interference between units. Equivalently, $T = A \Ttreat + (1-A) \Tcont$ and $C = A \Ctreat + (1-A) \Ccont$.
\end{assumption}

We define the \textit{survival function} $S(t) = \Prob(T > t)$, with potential counterparts $S^{(a)}(t) = \Prob(T^{(a)} > t)$ for $a \in \{0, 1\}$. The \textit{Restricted Mean Survival Time} (RMST) at horizon $\tau$ \citep{royston2013restricted} is defined as
\begin{equation*}
    \mathrm{RMST}(\tau) = \int_0^\tau S(t) \, \mathrm{d}t = \E[T \wedge \tau].
\end{equation*}
An important measure of treatment effect in survival analysis, and therefore our first estimand, is the \textit{difference of survival functions at time $t$} \citep{mao2018propensity, cheng2022addressing}:
\begin{equation*}
    \Delta S(t) = S^{(1)}(t) - S^{(0)}(t) = \E[\indic(\Ttreat > t)] - \E[\indic(\Tcont > t)].
\end{equation*}
A positive value of $\Delta S(t)$ indicates a higher probability of surviving beyond time $t$ under treatment (assuming the event is death). Second, the \textit{difference in RMST} 
\begin{equation*}
    \Delta \mathrm{RMST}(\tau) = \int_0^\tau S^{(1)}(t) \, \mathrm{d}t - \int_0^\tau S^{(0)}(t) \, \mathrm{d}t
\end{equation*}
corresponds to the difference in areas under the survival curves up to time $\tau$, and can be interpreted as an average gain or loss of survival time up to $\tau$ \citep{han2022restricted, voinot2025treatment}. Figure~\ref{fig:survival_analysis} in Appendix~\ref{app:add_background_info} gives a representation of the different quantities at stake.

These causal estimands cannot be directly computed from the observed data, since both potential outcomes and censoring are present. However, they can be estimated from an i.i.d.\ sample $\mathcal{D} = \{ \X_i, A_i, \tilde{T}_i, \Delta_i \}_{i=1}^n$ of size $n \geq 1$, with individual index $i$, under appropriate causal assumptions.

\subsection{Survival Analysis and Causal Assumptions with Unobserved Confounders}

We consider that, had we observed $\U$, we would have captured all common causes of $A$ and $T^{(a)}$, and $A$ and $C^{(a)}$ not included in $\X$ \citep{imbens2003sensitivity, kallus2021causal}. This translates into the following ($\X, \U$)\textit{-ignorability assumption}.

\begin{assumption}[($\X, \U$)-ignorability] \label{ass:xu-ignorability_survival}
    For $a \in \{0, 1\}$, $(T^{(a)}, C^{(a)}) \independent A | \X, \U$.
\end{assumption}
We also assume \textit{conditional independent censoring}, or \textit{informative censoring} \citep{mao2018propensity}.
\begin{assumption}[Conditional independent censoring with $\U$] \label{ass:cond_indep_censoring_u}
    For $a \in \{0, 1\}$, $T^{(a)} \independent C^{(a)} | \X, \U, A$.
\end{assumption}

We define the \textit{nominal propensity score} $e(\X) = \Prob(A=1 | \X)$ and the \textit{true propensity score} $e(\X, \U) = \Prob(A=1 | \X, \U)$, and assume \textit{positivity}~\citep{rosenbaum1983central}.
\begin{assumption}[Positivity] \label{ass:positivity_surv}
    $\forall \mathbf{x} \in \mathcal{X}, \mathbf{u} \in \mathcal{U}, \, 0 < e(\mathbf{x}) < 1$ and $0 < e(\mathbf{x},\mathbf{u}) < 1$.
\end{assumption}

We next introduce $S(t|\X=\x, A=a)$ the \textit{conditional survival functions}, and the \textit{conditional censoring function} $\bar{G}(t | \X=\x,\U=\mathbf{u}, A=a) = \Prob(C > t | \X=\x, \U=\mathbf{u}, A=a)$. We also define the time $\taumax > 0$, which is the maximum time such that $S(\taumax |\X=\x, A=a) > 0$ and $\bar{G}(\taumax |\X=\x, A=a) > 0$, for all $\x$ in $\mathcal{X}$ and $a$ in $\{0, 1\}$. 

For simplicity, we assume that the censoring mechanism does not depend on $\U$.
\begin{assumption}[Independence of the censoring function with respect to $\U$] \label{ass:censoring_fun_indep_u} For all  $t \in [0, \taumax]$,
        $\bar{G}(t | \X, \U, A) = \bar{G}(t | \X, A).$
\end{assumption}
This assumption is weaker than that of \citet{lee2024sensitivity}, who impose complete independence, i.e., $C \independent (A, \X, \U)$ (non-informative censoring).

Under these assumptions, the potential survival function admits the forms
\begin{align*}
    S_\mathrm{I}^{(1)}(t) = 1 - \E \biggl[ \frac{A \Delta \indic(\tilde{T} \leq t)}{e(\X, \U) \bar{G}(\tilde{T} | \X, A=1)} \biggr] \quad \text{and} \quad S_\mathrm{II}^{(1)}(t) = \E \biggl[ \frac{A \indic(\tilde{T} > t)}{e(\X, \U) \bar{G}(t | \X, A=1)} \biggr].
\end{align*}
See Appendix~\ref{app:proofs} for the derivation of Forms~I and II of the potential survival function. Intuitively, the two forms resemble a traditional Inverse Probability of Treatment Weighting (IPTW) approach where an observed outcome is weighted by the assigned treatment and the inverse of the propensity score to correct for selection bias. Moreover, the division by the censoring function corrects for censoring bias in both forms (IPCW). Note also that Form~I uses only uncensored observations (because of $\Delta$) whereas Form~II leverages all observations. Finally, Form~I may be more stable because the censoring function is evaluated at $\tilde{T}$ whereas it is evaluated at $t$ in Form~II: when the positivity of the censoring function is violated for high values of time, small values of the censoring function are ‘‘diluted" by the expectancy in Form~I but not in Form~II. In practice, we recommend to work with Form~I, as it is more stable and as we did not identify settings where Form~II performed better. Nevertheless, in the following, we keep explanations and results for both forms.

As $\U$ is, by definition, not observed, $e(\X, \U)$ cannot be estimated from the data. Our goal is therefore not to get a point estimate but rather to derive \textit{sets of estimates} by controlling the deviation between $e(\X)$ and $e(\X, \U)$.

\subsection{Marginal Sensitivity Model and DVDS Bounds }

To this end, we adopt the \textit{Marginal Sensitivity Model} (MSM) of \citet{tan2006distributional}, which bounds the odds ratio between the true and nominal propensity scores by a parameter $\Gamma \geq 1$:
\begin{equation} \label{eqn:MSM_U_chap4}
    \Gamma^{-1} \leq \frac{\Prob(A=1 | \X=\x, \U=\unc) / \Prob(A=0 | \X=\x, \U=\unc)}{\Prob(A=1 | \X=\x) / \Prob(A=0 | \X=\x)} \leq \Gamma.
\end{equation}
When $\Gamma = 1$, this reduces to $\X$-ignorability (or simply ignorability): $\forall a \in \{0, 1\}$,\\ $(T^{(a)}, C^{(a)}) \independent A | \X$
(because the nominal and true propensity scores are equal). Larger values of $\Gamma$ represent a departure from this assumption. For ways to calibrate the sensitivity parameter $\Gamma$, we refer the reader to \citet{cinelli2020making} and \citet{baitairian2025sensitivity}.

For convenience, we reformulate the MSM in terms of the potential outcome $Y_t(a) = \indic(T^{(a)} \leq t)$ (because $S^{(a)}(t) = 1 - \E[\indic(T^{(a)} \leq t)]$). By Proposition~\ref{prop:MSM_Ya_equiv}, Equation~\eqref{eqn:MSM_U_chap4} is equivalent to
    \begin{equation} \label{eqn:MSM_Bayes_chap4}
        \Gamma^{-1} \leq \frac{f(Y_t(a)=y|\X=\x, A=0)}{f(Y_t(a)=y|\X=\x, A=1)} \leq \Gamma,
    \end{equation}
    for $a \in \{0, 1\}$, $t \in [0, \taumax]$, and $\x \in \mathcal{X}$.

To derive the sensitivity bounds on $S_\mathrm{I}^{(1)}(t)$, we need to define several quantities:
\begin{itemize}
    \item a modified outcome 
    $\YI = \Delta \indic(\tilde{T} \leq t) / \bar{G}(\tilde{T} | \X, A=1)$;
    \item the conditional quantile $Q_\nu(\X=\x, A=1)$ of order $\nu$ of the distribution of $\YI$ conditionally on $\X=\x$ and $A=1$;
    \item the regression functions $\rho_+$ and $\rho_-$:
   \begin{align*}
        \rho_\pm(t, \X, A=1, \Gamma) & = \Gamma^{-1} \E[ \YI | \X, A=1] \\
        & \quad + (1 - \Gamma^{-1}) \E \biggl[ Q_\pm(\X, 1) + \frac{\{ \YI - Q_\pm(\X, 1) \}_\pm}{1-\gamma} \bigg| \X, A=1 \biggr] 
    \end{align*} with $Q_+(\X, 1) = Q_\gamma(\X, A=1)$, $Q_-(\X, 1) = Q_{1-\gamma}(\X, A=1)$ and $\gamma = \Gamma / (1 + \Gamma)$. We write $\pm$ when we must read the formula first with a $+$ sign and then with a $-$ sign, and write $\mp$ for the opposite. 
\end{itemize}

Doubly valid and doubly sharp bounds combine two types of bounds: some based on modified outcome regressions and Equation~\eqref{eqn:MSM_Bayes_chap4}, and some based on propensity scores and Equation~\eqref{eqn:MSM_U_chap4}. Details are provided in Appendix~\ref{app:theo_bounds_estim_I}. The next proposition gives these two different upper bounds derived for $S_\mathrm{I}^{(1)}(t)$. 

\begin{proposition}
    Under the MSM of Equation~\eqref{eqn:MSM_U_chap4}, both
    \begin{equation} \label{eqn:ub_I_modified_outcome_reg}
        1 - \E[ A \YI + (1-A) \rho_-(t, \X, 1, \Gamma)]
    \end{equation}
    and 
    \begin{equation} \label{eqn:ub_I_prop_score}
         1 - \E \biggl[ \frac{A Q_{1-\gamma}(\X, 1)}{e(\X)} + A (\YI - Q_{1-\gamma}(\X, 1)) \biggl( 1 + \frac{1-e(\X)}{e(\X)} \Gamma^{-\sign(\YI - Q_{1-\gamma}(\X, 1))} \biggr) \biggr]
    \end{equation}
    are sharp upper bounds for $S_\mathrm{I}^{(1)}(t)$.
\end{proposition}
The corresponding lower bounds and details on the derivation are given in Appendix~\ref{app:theo_bounds_estim_I}.
Similar ideas can be used for $S_\mathrm{II}^{(1)}(t)$ by working with $\E[\indic(\Ttreat > t)|\X=\x]$ instead of $1-\E[\indic(\Ttreat \leq t)|\X=\x]$.

Combining these two bounds entails the following theorem. 
\begin{theorem}[DVDS bounds on $S_\mathrm{I}^{(1)}(t)$] \label{theo:dvds_estim_I}
    Define the lower $S_{\mathrm{I},-}^{(1)}(t, \Gamma)$ and upper $S_{\mathrm{I},+}^{(1)}(t, \Gamma)$ bounds for $S_\mathrm{I}^{(1)}(t)$ as
    \begin{align}
         &S_{\mathrm{I}, \pm}^{(1)}(t, \Gamma) = 1 - \E[ A \YI + (1-A)\rho_\mp(t, \X, 1, \Gamma)] \nonumber \\
         - &\E \biggl[ A \frac{1 - e(\X)}{e(\X)} \biggl( Q_\mp(\X, 1) + \Gamma^{\mp \sign(\YI - Q_\mp(\X, A=1))} ( \YI - Q_\mp(\X, 1)) - \rho_\mp(t, \X, 1, \Gamma) \biggr) \biggr], \label{eqn:PIS_estim_I}
        \end{align} or equivalently
        \begin{align}
        & S_{\mathrm{I}, \pm}^{(1)}(t, \Gamma) = 
         1 - \E \biggl[ \frac{A Q_\mp(\X, 1)}{e(\X)} + A (\YI - Q_\mp(\X, 1)) \biggl( 1 + \frac{1-e(\X)}{e(\X)} \Gamma^{\mp \sign( \YI - Q_\mp(\X, 1))} \biggr) \biggr] \nonumber \\
        & - \E \biggl[ \biggl( 1 - \frac{A}{e(\X)} \biggr) \rho_\mp(t, \X, 1, \Gamma) \biggr]. \label{eqn:PIS_estim_I_prop_score}
    \end{align}
     Under Assumptions~\ref{ass:sutva_survival} to \ref{ass:censoring_fun_indep_u}, and when the censoring function $\bar{G}$ is consistently estimated, these bounds are doubly valid and doubly sharp with respect to the propensity scores and modified outcome regressions, as defined in \citet{dorn2024doubly}.
\end{theorem}

An equivalent theorem for DVDS bounds on $S_\mathrm{II}^{(1)}(t)$ is given in Appendix~\ref{app:theo_bounds_estim_II}. On the first line of Equation~\eqref{eqn:PIS_estim_I}, we recognize the formulation of the bounds via modified outcome regression (Equation~\eqref{eqn:ub_I_modified_outcome_reg}). The second line has mean zero and introduces robustness in the sense that, if $\hat{\rho}_\pm$ was estimated consistently, the second line would be null, even if $\hat{e}(\X)$ was misspecified, and the sharpness property from the formulation of the bounds via modified outcome regression would apply. On the first line of Equation~\eqref{eqn:PIS_estim_I_prop_score}, we recognize the formulation of the bounds via propensity scores (Equation~\eqref{eqn:ub_I_prop_score}) and a similar line of reasoning applies.

Remark that the regression function $\rho_\pm$ depends on the sensitivity parameter $\Gamma$ and on the Conditional Value at Risk (CVaR), also called Expected Shortfall (the second and longer conditional expectation). See, for example, \citet{rockafellar2000optimization} for a related literature on the CVaR. In our case, the CVaR is an expected value measured in the tail of the distribution of $Y$ conditionally on $\mathbf{X}$ and $A$. When $\Gamma = 1$ (no unobserved confounding), $\rho_\pm$ is simply the expected value of $Y$ conditionally on $\mathbf{X}$ and $A$. When $\Gamma$ tends to infinity, $\rho_\pm$ gets closer to the CVaR, so to the tail of the distribution of the outcome.

Observe that the results from \citet{dorn2024doubly} are also compatible with survival outcomes and we can replace their continuous outcome ``$Y$" with
    \begin{align*}
        \YI = \Delta \indic(\tilde{T} \leq t) / \bar{G}(\tilde{T} | \X, A=1) \quad \text{and} \quad \YII = \indic(\tilde{T} > t) / \bar{G}(t | \X, A=1)
    \end{align*}
to recognize, respectively, our DVDS bounds on Form~I and Form~II. However, our Theorems~\ref{theo:dvds_estim_I} and \ref{theo:dvds_estim_II} are based on different proofs of the bounds via modified outcome regression. Note also that our bounds inherit the efficiency property of the bounds from \citet{dorn2024doubly}.

Bounds for the unexposed group, $S_{\pm}^{(0)}(t, \Gamma)$, are obtained by replacing $A$ with $1-A$ (therefore, conditioning on $A=1$ also becomes conditioning on $A=0$) and replacing $e(\X)$ with $1 - e(\X)$. Bounds for the difference of survival functions are obtained by taking $S_{-}^{(1)}(t, \Gamma) - S_{+}^{(0)}(t, \Gamma)$ (lower bound) and $S_{+}^{(1)}(t, \Gamma) - S_{-}^{(0)}(t, \Gamma)$ (upper bound).

With an i.i.d.\ sample $\mathcal{D}$ as defined earlier, empirical estimators of Equations~\eqref{eqn:PIS_estim_I} and \eqref{eqn:PIS_estim_II} can be obtained by replacing the nuisance parameters with plug-in estimators $\hat{e}$, $\hat{\bar{G}}$, $\hat{\rho}_\pm$, and $\hat{Q}_\pm$ using $K$-fold cross-fitting, and by replacing the expectations with empirical averages. In particular, the modified outcome regression $\rho_\pm$ can be achieved by regressing $\Gamma^{-1} \YI + (1 - \Gamma^{-1}) (Q_\pm(\X, 1) + (1-\gamma)^{-1} \{ \YI - Q_\pm(\X, 1) \}_\pm)$ on $\X$ and $A$ for Form~I (and similarly for Form~II). The resulting intervals are called \textit{Point Estimate Intervals} (PEI). Confidence intervals (CIs) can be obtained via the percentile bootstrap method \citep{zhao2019sensitivity} or with less computationally intensive Wald-type CIs \citep{dorn2024doubly}.

\section{Experiments} \label{sec:experiments_chap4}

In this section, we compare the performance of our DVDS bounds for Forms~I and II to the bounds of \citet{lee2024sensitivity} for the survival functions and difference of survival functions on simulated and real-world data. Results for bounds on the RMST and difference in RMST are given in Appendix~\ref{app:add_exp}. The R repository \href{https://github.com/Sanofi-Public/CSM_dvds_bounds_survival_outcomes}{\textbf{CSM\_dvds\_bounds\_survival\_outcomes}} to reproduce the results is publicly available on the Sanofi-Public GitHub under a non-commercial license.

\subsection{Implementation Details}

All experiments were performed under the R software (version 4.3.2) \citep{rstatisticalsoftware} using parallel computing on 32 CPUs (Amazon EC2 c6i.8xlarge instances). We parallelized computations between different values of time $t$ (survival function) with our DVDS bounds and parallelized the optimization step in the method of \citet{lee2024sensitivity}.

5-fold cross-fitting was used to estimate the nuisance parameters and avoid overfitting. The nominal propensity score $\hat{e}$ was estimated via logistic regression (\texttt{glm} function from the \texttt{stats} package, version 4.3.2), the conditional quantile regressions $\hat{Q}_\pm$ via linear quantile regression (\texttt{rq} function from the \texttt{quantreg} package, version 5.97) with the Frisch–Newton interior point method, the modified outcome regressions $\hat{\rho}_\pm$ via random forests (\texttt{ranger} function from the \texttt{ranger} package, version 0.16.0), and the censoring function $\hat{\bar{G}}$ via a Cox proportional hazard model (\texttt{coxph} function from the \texttt{survival} package, version 3.5-7) with the Breslow method for dealing with ties. Except for the \texttt{rq} and \texttt{coxph} functions, we used the default hyperparameter values. In the second application on real data, 95\%-confidence intervals were computed via the percentile bootstrap method with $B = 200$ bootstrap samples.

We reimplemented the direct optimization method from \citet{lee2024sensitivity} for the difference in RMST (Equations~(9) to (12) from their paper) and adapted it for the difference of survival functions by removing the integration step. Even if it is faster, we did not implement the approximated method presented in their Section~3.3 as it relies on additional assumptions on the censoring process. As they suggest in Section~3.4, we used the \texttt{optim} function from the \texttt{stats} package (more precisely, we used the \texttt{optimParallel} function, a parallelized implementation of \texttt{optim}, from the \texttt{optimParallel} package, version 1.0-2) with the L-BFGS-B optimization method, a maximum number of iterations of 100, and a numeric forward difference approximation of the gradient.

Following \citet{dorn2024doubly}, we do not normalize (or stabilize) our estimators as in Appendix~\ref{proof:estimand_I} and \ref{proof:estimand_II}. Experimental results support this choice by showing good stability properties, even without normalization.

\subsection{Results on Synthetic Experiments} \label{sec:results_synth_exp}

We applied our methodology to a simulated dataset whose construction was inspired by \citet{lee2024sensitivity} and \citet{baitairian2025sensitivity}. In our experiments, we considered $p_\mathbf{X} = 2$ observed confounders and $p_\mathbf{U} = 2$ unobserved confounders. We chose $\mathbf{X} \sim \mathcal{U}(-1, 1)^{p_\mathbf{X}}$, where $\mathcal{U}(a, b)^p$ is the uniform distribution on $(a, b)$ of dimension $p$. For $j$ in $[\![1, p_\mathbf{U}]\!]$, conditionally on $\mathbf{X}=\mathbf{x}$, we designed each unobserved confounder $U_j$ as
\begin{equation*}
    U_j|\mathbf{X}=\mathbf{x} \sim \lambda G + (1 - \lambda) \beta_j^T \mathbf{x} = \mathcal{N} \bigl( (1-\lambda) \beta_j^T \mathbf{x}, \, \lambda^2 \bigr),
\end{equation*}
with $G \sim \mathcal{N}(0, 1)$, $0 < \lambda < 1$, and $\beta_j = (\beta_{j,1}, \dots, \beta_{j,p_\mathbf{X}}) \in \mathbb{R}^{p_\mathbf{X}}_+$. We chose the distribution of $A$ conditionally on $\mathbf{X}$ and $\mathbf{U}$ to be a Bernoulli satisfying the MSM with a true sensitivity parameter $\Gamma^\star = 3$. We designed the true propensity score $e(\mathbf{X}, \mathbf{U}) = \Prob(A=1|\mathbf{X}, \mathbf{U})$ such that it marginalizes to the nominal propensity score $e(\mathbf{X}) = \Prob(A=1|\mathbf{X}) = \mathrm{logistic}(\delta^T \mathbf{X} + 0.5)$, with $\delta \in \mathbb{R}^{p_\mathbf{X}}$. We explain how this is done in Appendix~\ref{app:complete_sim_setup}.

The potential outcome $T^{(a)}$, with $a$ in $\{0, 1\}$, was randomly generated using the inverse transform sampling method via a proportional hazards model with Weibull-distributed baseline hazard \citep{bender2005generating}. In particular,
\begin{equation*}
    T^{(a)} = 10 \cdot \biggl( \frac{- \log (U^{(a)})}{0.95 \cdot \exp \bigl( a\log(5) + \beta_\X^T \X + \beta_\U^T \U \bigr)} \biggr)^{1/1.8},
\end{equation*}
where $U^{(a)} \sim \mathcal{U}(0, 1)$. The event time $T$ was then defined as $T = \indic(A=1) \cdot \Ttreat + \indic(A=0) \cdot \Tcont$ and the censoring time $C$ followed a Weibull distribution of shape 6 and scale 10. Note that, not to disadvantage the method of \citet{lee2024sensitivity}, $T$ and $C$ are independent and $C$ is independent from $\X$ and $A$, which corresponds to non-informative censoring. We provide experimental results under informative censoring in Appendix~\ref{app:add_exp}. We generated 20 Monte-Carlo samples, each of size $n = 1000$, and computed the sensitivity bounds for 10 values of time equally spaced between 0.1 and 9. Finally, administrative censoring was applied after the 0.95-quantile of the observed times. The complete simulation setup is given in Appendix~\ref{app:complete_sim_setup}. See also Appendix~\ref{app:add_exp} for additional results (higher number of Monte-Carlo samples, larger sample sizes, high dimensional dataset...).

\begin{figure}[h!]
    \centering
    \includegraphics[width=0.6\linewidth]{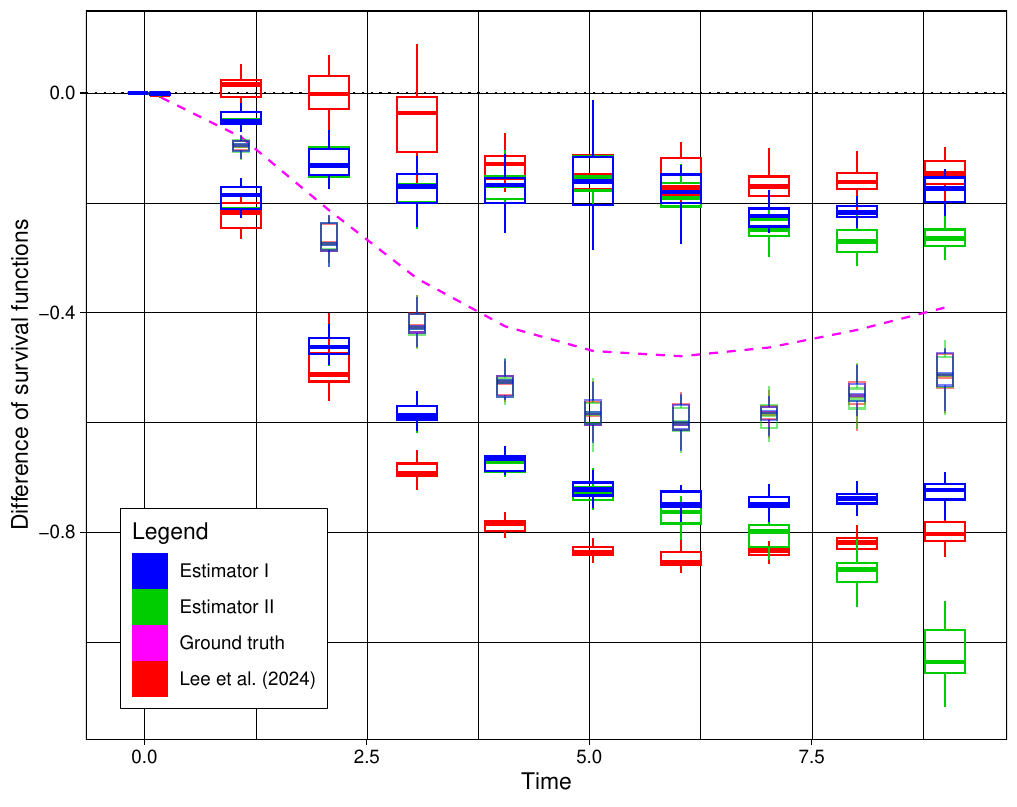}
    \caption{Sensitivity analysis for the difference of survival functions on the simulated data. The dotted lines in magenta are the true difference of survival functions. The large boxplots correspond to the estimated upper and lower sensitivity bounds (PEI) with Form~I (in blue), Form~II (in green), and the method from \citet{lee2024sensitivity} (in red) on 20 Monte-Carlo samples, for $\Gamma = 3$. The narrower and transparent boxplots correspond to the value under ignorability ($\Gamma = 1$) for each method.}
    \label{fig:simul_diff_surv_fun_results}
\end{figure}

Results of sensitivity analysis for the difference of survival functions are given in Figure~\ref{fig:simul_diff_surv_fun_results} for our DVDS bounds via Forms~I (in blue), II (in green) and the bounds from \citet{lee2024sensitivity} (in red). Under ignorability (narrow and transparent boxplots), the estimations deviate slightly from the ground truth (dotted lines in magenta) because of unobserved confounders, but are still close to the true value. Note that the three methods give approximately the same estimation.

When $\Gamma = 3$, that is, the same value $\Gamma^\star$ that generated the data, the intervals almost always contain the true value of the difference of survival functions. Remark that the DVDS bounds for Form~II are unstable for high values of time $t$. This behavior was also identified by \citet{cheng2022addressing} and can be explained by near violations of the positivity of the censoring process. On the contrary, the DVDS bounds for Form~I remain stable.

Notice that the intervals formed by the medians of our DVDS bounds for Form~I are also practically always included in the analogous intervals of the method from \citet{lee2024sensitivity}. These bounds show good coverage of the real difference of survival functions, and sharpness with respect to the bounds from \citet{lee2024sensitivity}.

Finally, our bounds also enjoy shorter computation times. To obtain Figure~\ref{fig:simul_diff_surv_fun_results}, 12 hours and 48 minutes were necessary for the method of \citet{lee2024sensitivity}. In comparison, only 1 minute and 48 seconds were required by our DVDS bounds for Form~I \textit{and} Form~II \textit{simultaneously}. This represents a reduction of computation time by a factor of 427.

\subsection{Results on Real Data} \label{sec:results_real_data}

\subsubsection{Right Heart Catheterization Dataset}

\begin{figure}[h!]
    \centering
    \subfigure[$\Gamma = 1.15$]{\includegraphics[width=0.49\textwidth]{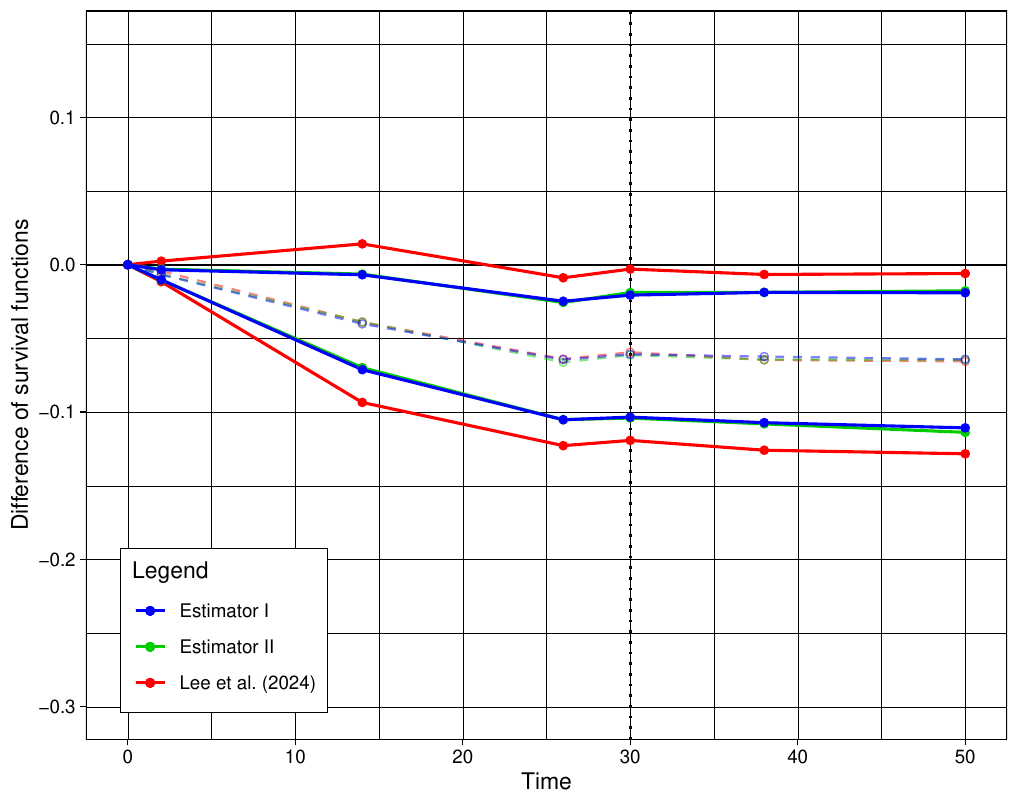}\label{fig:surv_gamma_1_15}}
    \hfill
    \subfigure[$\Gamma = 1.2$]{\includegraphics[width=0.49\textwidth]{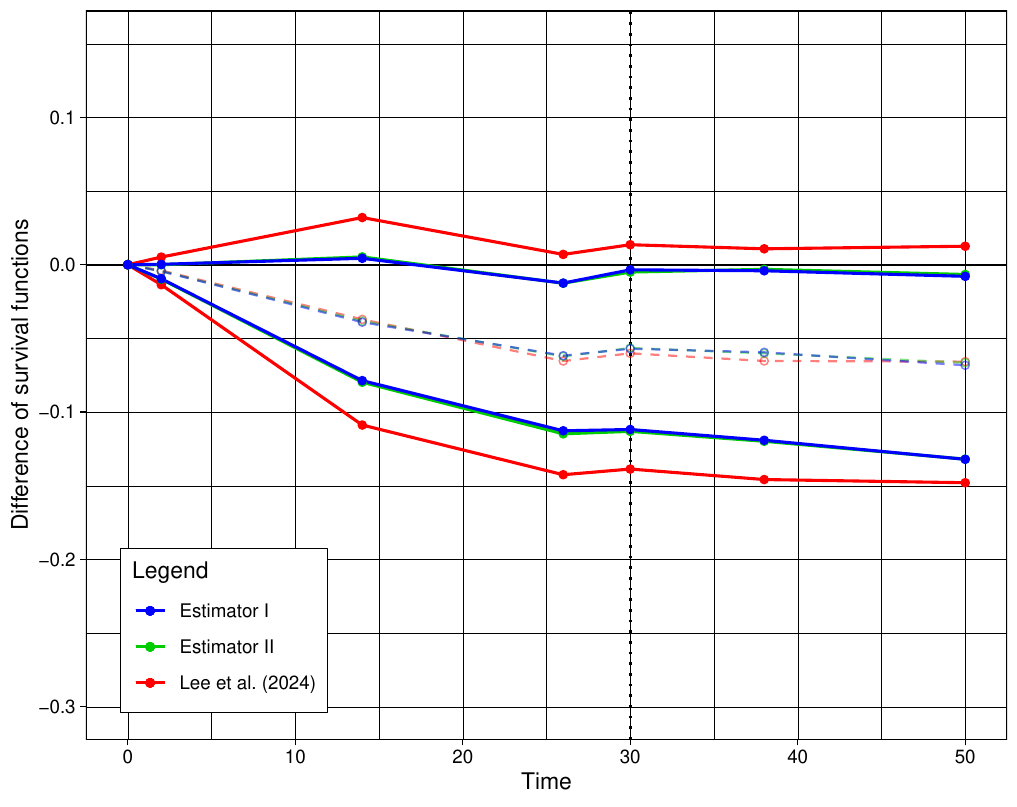}\label{fig:surv_gamma_1_2}}
    \caption{Sensitivity analysis for the difference of survival functions for $\Gamma = 1.15$ (\ref{fig:surv_gamma_1_15}) and $1.2$ (\ref{fig:surv_gamma_1_2}) on the RHC data. The curves correspond to the lower and upper sensitivity bounds (PEI) obtained with Form~I (in blue), Form~II (in green), and the method from \citet{lee2024sensitivity} (in red). The dotted and transparent lines correspond to the value under ignorability ($\Gamma = 1$) for each method. A vertical dotted black line indicates a time of 30 days. Note that the bounds with Form~II are almost coincident with the bounds obtained with Form~I.}
    \label{fig:rhc_surv_results}
\end{figure}

The \href{https://hbiostat.org/data/repo/rhc}{Right Heart Catheterization} (RHC) dataset consists in 5735 observations from 5 different US hospitals and was used by \citet{connors1996effectiveness} to explore the link between survival and the use of RHC during the first 24 hours in an intensive care unit. It was studied again later in \citet{lin1998assessing}. \citet{connors1996effectiveness} showed that RHC could lead to decreased survival rates in critically ill patients. A binary exposition signals if RHC was used ($A = 1$) or not ($A = 0$). The study admission date, hospital discharge date, date of last contact and date of death are given when available. Patients are considered uncensored ($\Delta = 1$) when a date of death is provided. In that case, the observed time $\tilde{T}$ is the period between study admission and death. When no date of death is provided, patients are considered censored ($\Delta = 0$). In that case, the observed time is the period between study admission and last contact, or hospital discharge if it corresponds to the last contact. Finally, the dataset provides 47 continuous and categorical variables at baseline (demographic and medical covariates), all included in our analysis. We applied administrative censoring on observed times that were greater than 180 days and binary encoded categorical variables.

The results of the sensitivity analysis on the difference of survival functions are given in Figure~\ref{fig:rhc_surv_results}, for a sensitivity parameter $\Gamma$ equal to 1.15 and 1.2. Additional graphs as well as computation times are provided in Appendix~\ref{app:add_exp}. The PEIs obtained via the DVDS bounds for Form~I (in blue), II (in green) and via the method of \citet{lee2024sensitivity} (in red) were computed for 5 equally spaced values of time between 2 and 50 days, to which we also added 30 days. The estimates under ignorability ($\Gamma = 1$) were represented with dotted lines.

Observe that, under ignorability, the difference of survival functions is always negative, regardless of the time point between 0 and 50 days. This is consistent with findings from \citet{connors1996effectiveness} on the negative effect of RHC on survival.

We highlight day 30 with a vertical dotted line to compare the current results of sensitivity analysis on survival outcomes to the ones obtained in \citet{dorn2024doubly} on a binary outcome (survival rate at day 30). We designate by \textit{critical value}, the lowest value of $\Gamma$ for which the sensitivity bounds contain the null effect. We recall that, with a binary outcome, the critical value of $\Gamma$ with the method from \citet{zhao2019sensitivity} is close to 1.15, whereas it is slightly larger and closer to 1.2 with the method from \citet{dorn2024doubly}. Similar critical values are obtained with survival outcomes. As the method from \citet{lee2024sensitivity} is equivalent to the bounds from \citet{zhao2019sensitivity} adapted to survival outcomes, the red PEI from Figure~\ref{fig:surv_gamma_1_15} when $\Gamma_{c, \mathrm{Lee}} = 1.15$ touches 0 at day 30. Our DVDS bounds for Form~I and Form~II give results that are closer to the ones obtained with the method from \citet{dorn2024doubly}. Indeed, when $\Gamma_{c, \mathrm{DVDS}} = 1.2$, the blue and green PEIs from Figure~\ref{fig:surv_gamma_1_2} touch 0 at day 30. Note also that our DVDS bounds are contained in the red bounds of our comparator for all time points between 0 and 50 days.

We also applied an \textit{informal benchmarking} approach in the Appendix where $\Gamma$ was estimated by taking the maximum odds ratio of the MSM with and without each observed covariate, one after another, as described by \citet{cinelli2020making} or \citet{baitairian2025sensitivity}. The minimum, median and maximum estimated values are 1.07, 1.63, and 35.6. Among the observed covariates, 8 have an associated confounding strength $\hat{\Gamma}$ lower than $\Gamma_{c, \mathrm{Lee}}$ and 9 have an associated confounding strength lower than $\Gamma_{c, \mathrm{DVDS}}$.

In the end, we can conclude that, even if RHC appears to be detrimental to survival at 30 days in intensive care units, ``small" confounders would be enough to explain away a causal association.

\subsubsection{German Breast Cancer Study Group Dataset}

\begin{figure}[h!]
    \centering
    \subfigure[$\Gamma = 1.35$]{\includegraphics[width=0.49\textwidth]{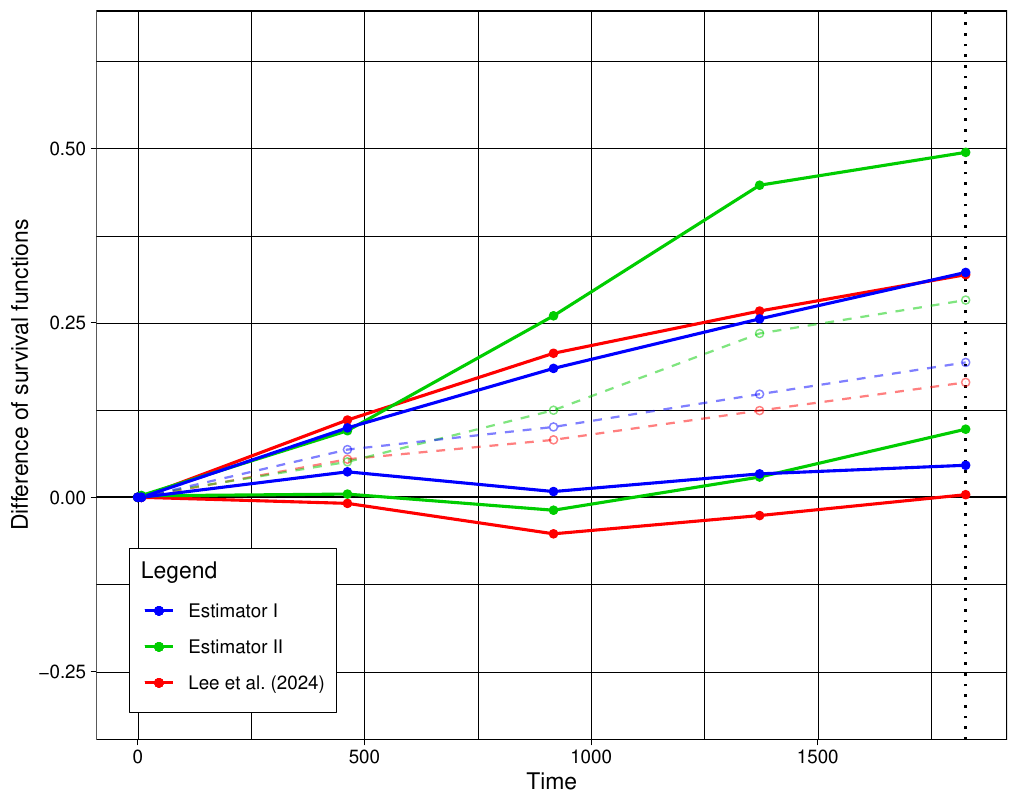}\label{fig:gbc_surv_gamma_1_35}}
    \hfill
    \subfigure[$\Gamma = 1.47$]{\includegraphics[width=0.49\textwidth]{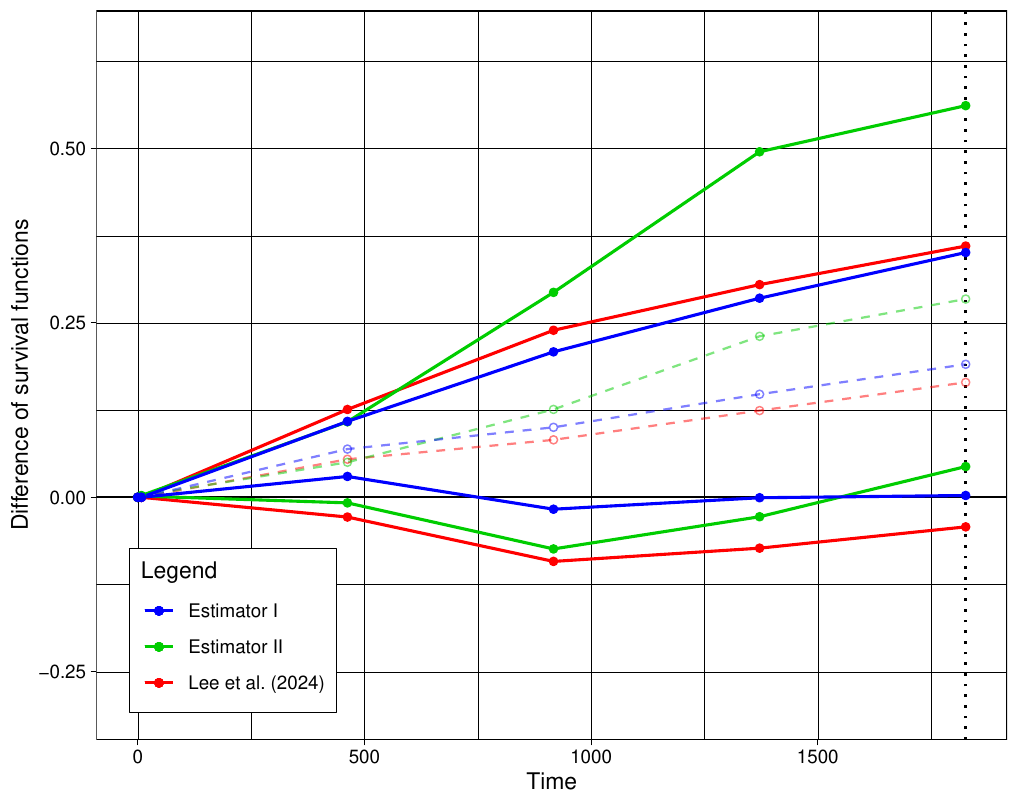}\label{fig:gbc_surv_gamma_1_47}}
    \caption{Sensitivity analysis for the difference of survival functions for $\Gamma = 1.35$ (\ref{fig:gbc_surv_gamma_1_35}) and $1.47$ (\ref{fig:gbc_surv_gamma_1_47}) on the GBCSG data. The curves correspond to the lower and upper sensitivity bounds (PEI) obtained with Form~I (in blue), Form~II (in green), and the method from \citet{lee2024sensitivity} (in red). The dotted and transparent lines correspond to the value under ignorability ($\Gamma = 1$) for each method. A vertical black dotted line indicates a time of 1826.25 days (5 years).}
    \label{fig:gbc_surv_results}
\end{figure}

The \href{https://www.kaggle.com/datasets/utkarshx27/breast-cancer-dataset-used-royston-and-altman}{German Breast Cancer Study Group} (GBCSG) dataset \citep{schumacher1994randomized} consists in 686 observations that relate the use of tamoxifen, a selective estrogen receptor modulator, to recurrence-free survival in breast cancer. We refer to \citet{lee2024sensitivity} for more details on the GBCSG dataset. A binary exposition indicates the treatment group ($A = 1$ if chemotherapy and tamoxifen were used, and $A = 0$ if only chemotherapy was used). The observed time $\tilde{T}$ is the number of days between the date of mastectomy and the date of the first occurrence of either locoregional recurrence, distant recurrence, contralateral tumor, secondary tumor, or death \citep{schumacher1994randomized, lee2024sensitivity}. Patients with no occurrence were censored at the end of follow-up ($\Delta = 0$). Finally, 7 continuous and categorical variables at baseline are provided and were included in our analysis. See Appendix~\ref{app:real_data} for more details. We used a similar setup as with the RHC data except that the PEIs were computed for 5 equally spaced values of time between 8 and 1826.25 days (5 years). Results are summarized in Figure~\ref{fig:gbc_surv_results}. Additional results as well as computation times are given in Appendix~\ref{app:add_exp}.

Observe that, under ignorability, the difference of survival functions is always positive. This is consistent with studies that showed a positive effect of the combination of tamoxifen with chemotherapy when compared with chemotherapy alone \citep{early1998tamoxifen}. The slope is slightly higher when the bounds are based on Form~II. At 5 years, the critical value for the method of \citet{lee2024sensitivity} is $\Gamma_{c, \mathrm{Lee}} = 1.35$, whereas is it equal to $\Gamma_{c, \mathrm{I}} = 1.47$ and $\Gamma_{c, \mathrm{II}} = 1.58$ with, respectively, our bounds with Form~I and Form~II (see Figure~\ref{fig:gbc_surv_all_gamma_1_58} in Appendix for the critical value with Form~II). \citet{lee2024sensitivity} indicate that ``according to an expert’s opinion [...], the degree of possible unmeasured confounding is expected to be weaker than tumor size". As the maximum value of odds ratio of propensity scores with and without tumor size is approximately 2.05 (see Appendix~\ref{app:add_exp}), our bounds do not allow to draw a conclusion about the sign of the difference of survival functions ($\Gamma_{c, \mathrm{I}} < \Gamma_{c, \mathrm{II}} \leq 2.05$) as well as the bounds of \citet{lee2024sensitivity} ($\Gamma_{c, \mathrm{Lee}} \leq 2.05$).

Finally, \citet{lee2024sensitivity} identified a critical value of 1.08 with 95\%-CIs for the RMST. In our experiments (see Figures~\ref{fig:gbc_surv_all_results} and \ref{fig:gbc_surv_ci_all_results} in Appendix~\ref{app:add_exp}), we also identified the same critical value for survival functions with their method, and a slightly larger critical value with our DVDS bounds, which would not allow us to conclude about a positive effect of tamoxifen on recurrence-free survival in breast cancer if unobserved confounders of such a strength existed.

\section{Conclusion and Discussion} \label{sec:conclusion}

In this manuscript, we have extended the sensitivity analysis to unobserved confounders proposed by \citet{dorn2024doubly} to survival outcomes and proposed two estimators of sensitivity bounds under an informative censoring assumption for the difference in survival curves and RMST. In particular, our bounds enjoy the double validity and double sharpness properties of the bounds of \citet{dorn2024doubly}. We have shown on simulated and real data that our bounds outperformed the proposition of \citet{lee2024sensitivity}---an adaptation of the bounds by \citet{zhao2019sensitivity} to survival outcomes---in terms of sharpness and computation time.

Note that we used estimators of the survival function given in \citet{cheng2022addressing} and not the traditional Kaplan Meier estimator, as in \citet{lee2024sensitivity}. This is because defining sensitivity bounds with a Kaplan Meier estimator leads to solving a fractional problem since the true propensity score appears in the numerator \textit{and} the denominator of the Kaplan Meier, whereas it only appears once in Forms~I and II of the survival function.

As compared to \citet{lee2024sensitivity}, we made the weaker assumption of informative censoring. However, this led us to suppose that the censoring function is independent from the unobserved confounders to avoid another sensitivity analysis to the violation of informative censoring. Future works should tackle this issue.

Additionally, our bounds have robustness properties with respect to the propensity score and modified outcome regression, given that the conditional quantiles are consistently estimated, but lack robustness with respect to the censoring function. \citet{voinot2025treatment} present a potential solution called double augmented correction. This should also be addressed in future works.

Finally, in the appendix, we derived DVDS bounds for the RMST by directly considering an outcome that is proportional to the minimum between the event time $T$ and the time horizon $\tau$, instead of integrating bounds for the survival function. Even though this approach suffers from instabilities for high values of $\tau$ and when censoring is important, it could pave the way for a promising and faster method to compute sensitivity bounds for the RMST.

\acks{The authors want to thank Charlotte Voinot (Sanofi, Gentilly, France, and PreMeDICaL, UMR UA11, Montpellier, France) for open discussions about causal survival analysis. The authors also want to thank the reviewers for their valuable feedback and suggestions, which helped improving the quality of this work. The research in this paper was supported by Sanofi and INRIA (Institut National de Recherche en Informatique et en Automatique) through the CIFRE program (Convention Industrielle de Formation par la Recherche).}

\conflicts{Jean-Baptiste Baitairian, Bernard Sebastien and Rana Jreich are or were Sanofi employees and they may hold shares and/or stock options in the company. Agathe Guilloux is employed by Inria. Sandrine Katsahian is employed by Université Paris-Cité and Assistance Publique - Hôpitaux de Paris (AP-HP).}

\bibliography{bibliography}

\newpage
\appendix

\section{Additional Background Information} \label{app:add_background_info}

Figure~\ref{fig:survival_analysis} gives a visual representation of different quantities at stake: the survival functions and the difference in RMST.

\begin{figure}[ht]
    \centering
    \includegraphics[width=0.8\linewidth]{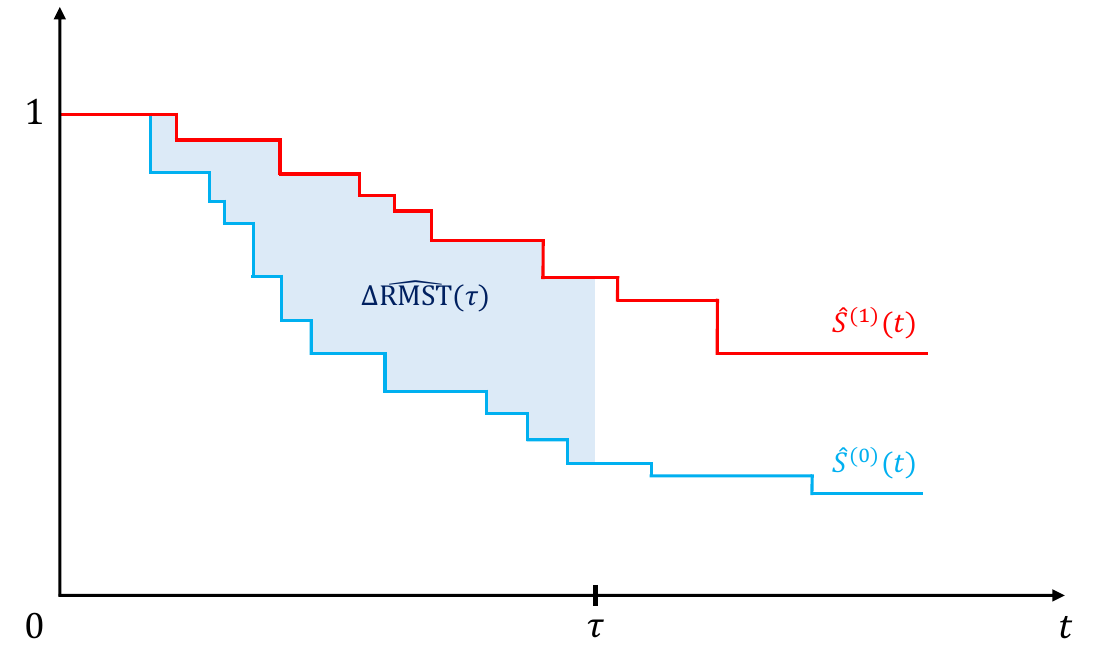}
    \caption[Illustration of the estimated difference in RMST (blue area), from time 0 up to time $\tau$, using the estimated survival functions among the exposed (in red) and the unexposed (in blue) groups.]{Illustration of the estimated difference in RMST (light blue area), from time 0 up to time $\tau$, using the estimated survival functions among the exposed group (in red) and the unexposed group (in blue).}
    \label{fig:survival_analysis}
\end{figure}

\section{Proofs} \label{app:proofs}

\subsection{Derivation of Form~I} \label{proof:estimand_I}

We refer the reader to the appendix of \citet{cheng2022addressing} for a proof of the consistency of the estimator of Form~I under ignorability and we will show that, under Assumptions~\ref{ass:sutva_survival} to \ref{ass:censoring_fun_indep_u}, $S_\mathrm{I}^{(1)}(t) = S^{(1)}(t)$.

\begin{align*}
    S^{(1)}(t) & = \mathbb{P}(\Ttreat > t)  \\
    & = \E[\indic(\Ttreat > t)] \\
    & = \E \bigl[ \E[\indic(\Ttreat > t) | \X, \U] \bigr] \quad \text{by tower property} \\
    & = \E \biggl[ \frac{e(\X, \U)}{e(\X, \U)} \E [ \indic(\Ttreat > t) | \X, \U] \biggr] \quad \text{by positivity (Assumption~\ref{ass:positivity_surv})} \\
    & = \E \biggl[ \frac{\E[A|\X, \U]}{e(\X, \U)} \E [ \indic(\Ttreat > t) | \X, \U] \biggr] \quad \text{because $\E[A|\X, \U] = \Prob(A=1|\X, \U) = e(\X, \U)$} \\
    & = \E \biggl[ \E \biggl[ \frac{A \indic(\Ttreat > t)}{e(\X, \U)} \bigg| \X, \U \biggr] \biggr] \quad \text{by ($\X, \U$)-ignorability (Assumption~\ref{ass:xu-ignorability_survival})} \\
    & = 1 - \E \biggl[ \frac{\E[A \indic(T \leq t)| \X, \U, A]}{e(\X, \U)} \biggr] \quad \text{by tower property and consistency} \\
    & = 1 - \E \biggl[ \frac{\E[A \indic(T \leq t)| \X, \U, A]}{e(\X, \U)} \cdot \frac{\E[\indic(C \geq T) | \X, \U, A]}{\E[\indic(C \geq T) | \X, \U, A]} \biggr] \\
    & = 1 - \E \biggl[ \frac{\E[A \indic(T \leq t) \indic(C \geq T) | \X, \U, A]}{e(\X, \U) \Prob(C \geq T | \X, \U, A)} \biggr] \quad \text{by Assumption~\ref{ass:cond_indep_censoring_u}} \\
    & = 1 - \E \biggl[ \frac{A \indic(T \leq t) \indic(C \geq T)}{e(\X, \U) \Prob(C \geq T | \X, \U, A=1)} \biggr] \quad \text{by Assumption~\ref{ass:sutva_survival}\textsuperscript{*}, tower prop., and $A \in \{0, 1\}$} \\
    & = 1 - \E \biggl[ \frac{A \Delta \indic(\tilde{T} \leq t)}{e(\X, \U) \Prob(C \geq \tilde{T} | \X, A=1)} \biggr] \quad \text{by definition of $\Delta$ and $\tilde{T}$, and by Assumption~\ref{ass:censoring_fun_indep_u}} \\
    & = 1 - \E \biggl[ \frac{A \Delta \indic(\tilde{T} \leq t)}{e(\X, \U) \bar{G}(\tilde{T} | \X, A=1)} \biggr] \quad \text{\textbf{Unstabilized form}} \\
    & = 1 - \E \biggl[ \frac{A \Delta \indic(\tilde{T} \leq t)}{e(\X, \U) \bar{G}(\tilde{T} | \X, A=1)} \biggr] \bigg/ \E \biggl[ \frac{A}{e(\X, \U)} \biggr] \quad \text{\textbf{Stabilized form}}
\end{align*}
\textsuperscript{*}We use SUTVA: $\indic(T>t) = A \indic(\Ttreat > t) + (1-A) \indic(\Tcont > t)$ because $A \indic(T>t) = A^2 \indic(\Ttreat > t) + A(1-A) \indic(\Tcont > t) = A \indic(\Ttreat > t)$. The same is true for $C$: $A \indic(C>t) = A \indic(\Ctreat > t)$.

The stabilized form comes from the fact that $\E[A/e(\X, \U)] = \E[ \E[A/e(\X, \U) | \X, \U] ] = \E[ \E[A|\X, \U]/e(\X, \U) ] = 1$.

\subsection{Derivation of Form~II} \label{proof:estimand_II}

We refer the reader to the appendix of \citet{cheng2022addressing} for a proof of the consistency of the estimator of Form~II and we will show that, under Assumptions~\ref{ass:sutva_survival} to \ref{ass:censoring_fun_indep_u}, $S_\mathrm{II}^{(1)}(t) = S^{(1)}(t)$.

\begin{align*}
    S^{(1)}(t) & = \mathbb{P}(\Ttreat > t)  \\
    & = \E[\indic(\Ttreat > t)] \\
    & = \E \bigl[ \E[\indic(\Ttreat > t) | \X, \U] \bigr] \quad \text{by tower property} \\
    & = \E \biggl[ \frac{e(\X, \U)}{e(\X, \U)} \E [ \indic(\Ttreat > t) | \X, \U ] \biggr] \quad \text{by positivity (Assumption~\ref{ass:positivity_surv})} \\
    & = \E \biggl[ \frac{\E[A|\X, \U]}{e(\X, \U)} \E [ \indic(\Ttreat > t) | \X, \U ] \biggr] \quad \text{because $\E[A|\X, \U] = \Prob(A=1|\X, \U) = e(\X, \U)$} \\
    & = \E \biggl[ \E \biggl[ \frac{A \indic(\Ttreat > t)}{e(\X, \U)} \bigg| \X, \U \biggr] \biggr] \quad \text{by ($\X, \U$)-ignorability (Assumption~\ref{ass:xu-ignorability_survival})} \\
    & = \E \biggl[ \E \biggl[ \frac{A \indic(T > t)}{e(\X, \U)} \bigg| \X, \U \biggr] \biggr] \quad \text{by consistency\textsuperscript{*} (Assumption~\ref{ass:sutva_survival})} \\
    & = \E \biggl[ \E \biggl[ \frac{A \indic(T > t)}{e(\X, \U)} \bigg| \X, \U, A \biggr] \biggr] \quad \text{by tower property} \\
    & = \E \biggl[ \frac{A \E[\indic(T > t) | \X, \U, A]}{e(\X, \U)} \biggr] \\
    & = \E \biggl[ \frac{A \E[\indic(T > t) | \X, \U, A]}{e(\X, \U)} \cdot \frac{\E[\indic(C > t) | \X, \U, A]}{\E[\indic(C > t) | \X, \U, A]} \biggr] \\
    & = \E \biggl[ \frac{A \E[\indic(\tilde{T} > t) | \X, \U, A]}{e(\X, \U) \Prob(C > t | \X, \U, A)} \biggr] \quad \text{by Assumption~\ref{ass:cond_indep_censoring_u} and $\indic(\tilde{T} > t) = \indic(T > t) \indic(C > t)$} \\
    & = \E \biggl[ \frac{A \indic(\tilde{T} > t)}{e(\X, \U) \Prob(C > t | \X, A=1)} \biggr] \quad \text{by tower property, $A \in \{0, 1\}$, and Assumption~\ref{ass:censoring_fun_indep_u}} \\
    & = \E \biggl[ \frac{A \indic(\tilde{T} > t)}{e(\X, \U) \bar{G}(t | \X, A=1)} \biggr] \quad \text{\textbf{Unstabilized form}} \\
    & = \E \biggl[ \frac{A \indic(\tilde{T} > t)}{e(\X, \U) \bar{G}(t | \X, A=1)} \biggr] \bigg/ \E \biggl[ \frac{A}{e(\X, \U)} \biggr] \quad \text{\textbf{Stabilized form}}
\end{align*}
\textsuperscript{*}We use SUTVA: $\indic(T>t) = A \indic(\Ttreat > t) + (1-A) \indic(\Tcont > t)$ because $A \indic(T>t) = A^2 \indic(\Ttreat > t) + A(1-A) \indic(\Tcont > t) = A \indic(\Ttreat > t)$. The same is true for $C$: $A \indic(C>t) = A \indic(\Ctreat > t)$.

The stabilized form comes from the fact that $\E[A/e(\X, \U)] = \E[ \E[A/e(\X, \U) | \X, \U] ] = \E[ \E[A|\X, \U]/e(\X, \U) ] = 1$.

\subsection{Proposition~\ref{prop:MSM_Ya_equiv} and Proof} \label{app:MSMequivalence}

\begin{proposition} \label{prop:MSM_Ya_equiv}
    For all $a \in \{0, 1\}$, $\x \in \mathcal{X}$, and $y \in \mathcal{Y}$, the ratio in the MSM from Equation~\eqref{eqn:MSM_U_chap4} can be equivalently written
    \begin{equation} \label{eqn:MSM_Bayes}
        \Gamma^{-1} \leq \frac{f(Y_t(a)=y|\X=\x, A=0)}{f(Y_t(a)=y|\X=\x, A=1)} \leq \Gamma.
    \end{equation}
\end{proposition}
\begin{proof}
By Lemma~1 from \citet{tan2024model}, the formulation of the MSM in terms of unobserved confounders \eqref{eqn:MSM_U_chap4} is equivalent to the formulation in terms of potential outcomes
\begin{equation} \label{eqn:MSM_Y}
    \Gamma^{-1} \leq \frac{\Prob(A=1 | \X=\x, Y_t(a)=y) / \Prob(A=0 | \X=\x, Y_t(a)=y)}{\Prob(A=1 | \X=\x) / \Prob(A=0 | \X=\x)} \leq \Gamma,
\end{equation}
for $a \in \{0, 1\}$, $t \in [0, \taumax]$, $(y, \x, \unc) \in \mathcal{Y} \times \mathcal{X} \times \mathcal{U}$, and $\Gamma \geq 1$. By Bayes' formula,
    \begin{align*}
        & \Prob(A=1|\X=\x, Y_t(a)=y) = \frac{f(Y_t(a)=y|\X=\x, A=1) \Prob(A=1|\X=\x)}{f(Y_t(a)=y|\X=\x)} \quad \text{and} \\
        & \Prob(A=0|\X=\x, Y_t(a)=y) = \frac{f(Y_t(a)=y|\X=\x, A=0) \Prob(A=0|\X=\x)}{f(Y_t(a)=y|\X=\x)}.
    \end{align*}
    By plugging these formula into Equation~\eqref{eqn:MSM_Y}, we obtain the desired result. The converse is also true by Bayes' formula.
\end{proof}

\subsection{Derivation of the DVDS Bounds for \texorpdfstring{$S_\mathrm{I}^{(1)}(t)$}{}} \label{app:theo_bounds_estim_I}

We denote $\{ \cdot \}_- = \min(0, \cdot)$, the negative part function, and $\{ \cdot \}_+ = \max(0, \cdot)$, the positive part function.

\subsubsection{Bounds via Modified Outcome Regression} \label{sec:bounds_modif_outc_regr}

We start by deriving bounds for $S^{(1)}(t)$ using a modified outcome regression approach. Observe that, for all $\x \in \mathcal{X}$,
\begin{equation} \label{eqn:S1_integral}
    S^{(1)}(t | \X=\x) = 1 - \E[Y_t(1) | \X=\x] = 1 - \int_\mathcal{Y} y f(Y_t(1)=y | \X=\x) \, \mathrm{d}y.
\end{equation}
However, the density $y \mapsto f(Y_t(1)=y | \X=\x)$ is not directly identifiable from the data. Therefore, we rewrite it as
\begin{align} \label{eqn:f_Y1_X}
    f(Y_t(1)=y | \X=\x) & = f(Y_t(1)=y, A=1 | \X=\x) + f(Y_t(1)=y, A=0 | \X=\x) \nonumber \\
    & = f(Y_t(1)=y | \X=\x, A=1) \Prob(A=1|\X=\x) \nonumber \\
    & \quad + f(Y_t(1)=y | \X=\x, A=0) \Prob(A=0|\X=\x).
\end{align}
The only non-identifiable term is $f(Y_t(1)=y | \X=\x, A=0)$. By introducing the density ratio
\begin{align*}
    w^\star(Y_t(1)=y, \X=\x) = \frac{f(Y_t(1)=y | \X=\x, A=0)}{f(Y_t(1)=y | \X=\x, A=1)},
\end{align*}
we finally get
\begin{align*}
    S^{(1)}(t | \X=\x) & = 1 - \biggl( e(\x) \int_\mathcal{Y} y f(Y_t(1)=y | \X=\x, A=1) \, \mathrm{d}y \\
    & \quad + (1-e(\x)) \int_\mathcal{Y} y w^\star(Y_t(1)=y, \X=\x) f(Y_t(1)=y | \X=\x, A=1) \, \mathrm{d}y \biggr) \\
    & = 1 - e(\x) \E[Y_t(1)|\X=\x, A=1] \\
    & \quad - (1-e(\x)) \E[Y_t(1) w^\star(Y_t(1), \x)|\X=\x, A=1],
\end{align*}
after plugging Equation~\eqref{eqn:f_Y1_X} into Equation~\eqref{eqn:S1_integral}.

As in the proof of Appendix~\ref{proof:estimand_I}, we further introduce the censoring function $\bar{G}$ by noting that
\begin{equation*}
    \E[Y_t(1)|\X=\x, A=1] = \E \biggl[ \frac{\Delta \indic(\tilde{T} \leq t)}{\bar{G}(\tilde{T}|\X, A=1)} \bigg| \X=\x, A=1 \biggr].
\end{equation*}
By denoting $\YI = \Delta \indic(\tilde{T} \leq t) / \bar{G}(\tilde{T}|\X, A=1)$, where the dependence in $\X$, $\tilde{T}$, and $\Delta$ is dropped for simplicity, we can write
\begin{align*}
    S^{(1)}(t | \X=\x) & = 1 - e(\x) \E[ \YI | \X=\x, A=1] \\
    & \quad - (1-e(\x)) \E[ \YI w^\star(Y_t(1), \x) | \X=\x, A=1].
\end{align*}

Notice that, by Equation~\eqref{eqn:MSM_Bayes_chap4}, $w^\star$ takes its values in $[\Gamma^{-1}, \Gamma]$ and that, for all $\x \in \mathcal{X}$,
\begin{equation} \label{eqn:w_star_condition}
    \E[w^\star(Y_t, \x)|\X=\x, A=1] = 1,
\end{equation}
with $Y_t = \indic(T > t)$. The sensitivity bounds for $S^{(1)}(t | \X=\x)$ are therefore defined as
\begin{align}
    S^{(1)}_-(t, \x, \Gamma) & = 1 - e(\x) \E[ \YI | \X=\x, A=1] \nonumber \\
    & \quad - (1-e(\x)) \underset{w \in \mathcal{W}(\Gamma)}{\max} \E[ \YI w(Y_t, \x) | \X=\x, A=1] \quad \text{and} \label{eqn:S_1_-_cond_pbm} \\
    S^{(1)}_+(t, \x, \Gamma) & = 1 - e(\x) \E[ \YI | \X=\x, A=1] \nonumber \\
    & \quad - (1-e(\x)) \underset{w \in \mathcal{W}(\Gamma)}{\min} \E[ \YI w(Y_t, \x) | \X=\x, A=1], \label{eqn:S_1_+_cond_pbm}
\end{align}
where the set $\mathcal{W}(\Gamma)$ is
\begin{align*}
    \mathcal{W}(\Gamma) = \{ & w : \{0, 1\} \times \mathcal{X} \mapsto [\Gamma^{-1}, \Gamma]; \\
    & \forall (\x, t) \in \mathcal{X} \times [0, \taumax], \, \E[w(Y_t, \x)|\X=\x, A=1] = 1 \}.
\end{align*}
As $w^\star$ is unknown, the intuition behind Equations~\eqref{eqn:S_1_-_cond_pbm} and \eqref{eqn:S_1_+_cond_pbm} is to search the minimizer and maximizer in a set $\mathcal{W}(\Gamma)$ of functions $w$ that ``look like" $w^\star$.

Without loss of generality, let us consider the upper bound, $S^{(1)}_+(t, \x, \Gamma)$. As $w \in \mathcal{W}(\Gamma)$,
\begin{equation*}
    (w(Y_t, \x) - \Gamma^{-1}) / (1-\Gamma^{-1}) \in [0, (1-\gamma)^{-1}],
\end{equation*}
with $\gamma = \Gamma / (1 + \Gamma)$.
Using this fact and the normalization constraint from Equation~\eqref{eqn:w_star_condition}, we rewrite the minimization problem as
\begin{align}
    & \underset{w \in \mathcal{W}(\Gamma)}{\min} \E[ \YI w(Y_t, \x) | \X=\x, A=1] \nonumber \\
    & = \Gamma^{-1} \E[ \YI | \X=\x, A=1] + (1-\Gamma^{-1}) \underset{w \in \mathcal{W}(\Gamma)}{\min} \E \biggl[ \YI \frac{w(Y_t, \x) - \Gamma^{-1}}{1 - \Gamma^{-1}} \bigg| \X=\x, A=1 \biggr]. \label{eqn:min_pbm_before_cvar}
\end{align}
A similar reasoning as in the proof of Theorem~3.1 from \citet{baitairian2025sharp} allows for the identification of the minimization problem as a Conditional Value at Risk (CVaR) of level $1-\gamma$, thanks to the “Fenchel--Moreau--Rockafellar” dual representation of the CVaR \citep{herdegen2023elementary}, the CVaR being defined as
\begin{align*}
    \mathrm{CVaR}_{1-\gamma}(t, \X=\x, A=1) = \E[ -\YI | \YI \leq Q_{1-\gamma}(\x, 1), \X=\x, A=1],
\end{align*}
where $Q_\nu(\X=\x, A=1)$ is the quantile of order $\nu$ of $\YI$ conditionally on $\X=\x$ and $A=1$.

As pointed out by \citet{rockafellar2000optimization}, the CVaR can be expressed directly in terms of the conditional quantiles of the outcome. By transposing this result to $Q_{1-\gamma}(\X=\x, A=1)$, we get
\begin{equation} \label{eqn:cvar}
    \mathrm{CVaR}_{1-\gamma}(t, \X=\x, A=1) = \E \biggl[ Q_{1-\gamma}(\x, 1) + \frac{\{ \YI - Q_{1-\gamma}(\x, 1) \}_-}{1-\gamma} \bigg| \X=\x, A=1 \biggr],
\end{equation}
where $\{ \cdot \}_- = \min(0, \cdot)$.

After plugging Equation~\eqref{eqn:cvar} into Equation~\eqref{eqn:min_pbm_before_cvar}, and combining the result with Equation~\eqref{eqn:S_1_-_cond_pbm}, we obtain
\begin{equation*}
    S^{(1)}_+(t, \x, \Gamma) = 1 - e(\x) \E[ \YI | \X=\x, A=1] - (1-e(\x)) \rho_-(t, \x, 1, \Gamma),
\end{equation*}
where
\begin{align*}
    \rho_-(t, \X=\x, A=1, \Gamma) & = \Gamma^{-1} \E[ \YI | \x, A=1] \\
    & \quad + (1 - \Gamma^{-1}) \E \biggl[ Q_{1-\gamma}(\x, 1) + \frac{\{ \YI - Q_{1-\gamma}(\x, 1) \}_-}{1-\gamma} \bigg| \x, A=1 \biggr].
\end{align*}
Finally, taking the expectancy over $\X$ and using the tower rule gives
\begin{equation} \label{eqn:bound_with_cvar}
    S^{(1)}_+(t, \Gamma) = 1 - \E[ A \YI + (1-A) \rho_-(t, \x, 1, \Gamma)].
\end{equation}
This result is similar to the formulation using adversarial regressions of \citet{dorn2024doubly}, where the outcome ``$Y$" is replaced with $\YI$.

By results given in the proof of Theorem~3.1 in \citet{baitairian2025sharp}, these bounds are sharp when both the conditional quantile and censoring function are consistently estimated. Similar arguments to those given in \citet{dorn2024doubly} (Proposition~4) or in \citet{baitairian2025sharp} (proof of Proposition~3.5) help us to show that the resulting bound is also singly-valid with respect to the modified outcome regression, that is, if the modified outcome regression correctly learns the conditional expectation of the \textit{estimated} modified outcome but the conditional quantile is misspecified, the bound is still valid but possibly conservative.

\subsubsection{Bounds via Propensity Scores} \label{sec:bounds_prop_score}

We now express our bounds for $S^{(1)}(t)$ in terms of propensity scores and will combine them to the formulation using modified outcome regression in the next section. Under $\X$-ignorability, Form~I writes
\begin{equation*}
    S_\mathrm{I}^{(1)}(t) = 1 - \E \biggl[ \frac{A \Delta \indic(\tilde{T} \leq t)}{e(\X) \bar{G}(\tilde{T} | \X, A=1)} \biggr] = 1 -  \E \biggl[ \frac{A \YI}{e(\X)} \biggr],
\end{equation*}
where $\YI = \Delta \indic(\tilde{T} \leq t) / \bar{G}(\tilde{T}|\X, A=1)$, as previously. 

Using the formulation of the MSM with unobserved confounders (Equation~\eqref{eqn:MSM_U_chap4}),
\citet{dorn2022sharp} (Proposition~2) introduced sharp sensitivity bounds when $e(\X)$ is replaced with worst-case propensity scores $e_\pm(\X, Y)$ satisfying the MSM, where $Y$ is the (continuous) outcome, under the constraint that $\E[A/e_\pm(\X, Y)|\X] = 1$. We reformulate their result in terms of survival outcomes.
\begin{proposition}
    There exist two worst-case propensity scores, $e_-(\x, \YI, \Gamma)$ and\\ $e_+(\x, \YI, \Gamma)$, satisfying the MSM (Equation~\eqref{eqn:MSM_U_chap4}) and
    \begin{equation} \label{eqn:qb_constraint}
        \E[A / e_\pm(\X, \YI, \Gamma) | \X=\x] = 1
    \end{equation}
    such that
    \begin{align*}
        & 1 / e_-(\x, \YI, \Gamma) = 1 + \frac{1-e(\x)}{e(\x)} \Gamma^{- \sign(\YI - Q_{1-\gamma}(\x, 1))} \quad \text{and} \\
        & 1 / e_+(\x, \YI, \Gamma) = 1 + \frac{1-e(\x)}{e(\x)} \Gamma^{\sign(\YI - Q_{\gamma}(\x, 1))},
    \end{align*}
    where $Q_\nu(\X=\x, A=1)$ is the quantile of order $\nu$ of $\YI$ conditionally on $\X=\x$ and $A=1$. The sharp lower and upper bounds on $S_\mathrm{I}^{(1)}(t)$ are then
    \begin{equation} \label{eqn:qb_bounds}
        S_{\mathrm{I}, -}^{(1)}(t, \Gamma) = 1 - \E \biggl[ \frac{A \YI}{e_+(\X, \YI, \Gamma)} \biggr] \quad \text{and} \quad S_{\mathrm{I}, +}^{(1)}(t, \Gamma) = 1 - \E \biggl[ \frac{A \YI}{e_-(\X, \YI, \Gamma)} \biggr].
    \end{equation}
\end{proposition}

As in \citet{dorn2024doubly}, we can introduce $Q_{\gamma}(\X, 1)$ or $Q_{1-\gamma}(\X, 1)$ in Equation~\eqref{eqn:qb_bounds} using the constraint from Equation~\eqref{eqn:qb_constraint}. The resulting bounds are
\begin{align*}
    & S_{\mathrm{I}, -}^{(1)}(t, \Gamma) = 1 - \E \biggl[ \frac{A Q_{\gamma}(\X, 1)}{e(\X)} + A (\YI - Q_{\gamma}(\X, 1)) \biggl( 1 + \frac{1-e(\x)}{e(\x)} \Gamma^{\sign(\YI - Q_{\gamma}(\x, 1))} \biggr) \biggr] \\
    & \text{and} \quad S_{\mathrm{I}, +}^{(1)}(t, \Gamma) = 1 - \E \biggl[ \frac{A Q_{1-\gamma}(\X, 1)}{e(\X)} \\
    & \quad \quad \quad \quad \quad \quad \quad \quad \quad \quad + A (\YI - Q_{1-\gamma}(\X, 1)) \biggl( 1 + \frac{1-e(\x)}{e(\x)} \Gamma^{-\sign(\YI - Q_{1-\gamma}(\x, 1))} \biggr) \biggr].
\end{align*}
They are sharp when the propensity score, conditional quantile and censoring function are all consistently estimated. When the censoring function is consistently estimated, they are singly-valid with respect to the propensity score, that is, if the propensity score is correctly estimated but not the conditional quantile, the bound is still valid but possibly conservative. Proof of single validity rests upon the same arguments used in \citet{dorn2024doubly} (Proposition~2).

\subsection{Theorem~\ref{theo:dvds_estim_II}} \label{app:theo_bounds_estim_II}

The next theorem establishes the DVDS bounds on the second form of the survival function $S_\mathrm{II}^{(1)}(t)$ and their properties, where we denote $\YII = \indic(\tilde{T} > t) / \bar{G}(t | \X, A=1)$. The proof is based on the same elements as those in Theorem~\ref{theo:dvds_estim_I}.
\begin{theorem}[DVDS bounds on $S_\mathrm{II}^{(1)}(t)$] \label{theo:dvds_estim_II}
    Define the lower $S_{\mathrm{II}, -}^{(1)}(t, \Gamma)$ and upper $S_{\mathrm{II}, +}^{(1)}(t, \Gamma)$ bounds for $S_\mathrm{II}^{(1)}(t)$ as
    \begin{align}
        & S_{\mathrm{II}, \pm}^{(1)}(t, \Gamma) = \E[ A \YII + (1-A)\rho_\pm(t, \X, 1, \Gamma)] \nonumber \\
        & + \E \biggl[ A \frac{1 - e(\X)}{e(\X)} \nonumber \\
        & \quad \quad \times \biggl( Q_\pm(\X, 1) + \Gamma^{\pm \sign(\YII - Q_\pm(\X, A=1))} (\YII - Q_\pm(\X, 1)) - \rho_\pm(t, \X, 1, \Gamma) \biggr) \biggr], \label{eqn:PIS_estim_II}
    \end{align}
    where $Q_+(\X, 1) = Q_\gamma(\X, A=1)$, $Q_-(\X, 1) = Q_{1-\gamma}(\X, A=1)$, with $\gamma = \Gamma / (1 + \Gamma)$, $Q_\nu(\X, A=1)$, the quantile of order $\nu$ of the distribution of $\YII$ conditionally on $\X$ and $A=1$, and
    \begin{align*}
        \rho_\pm(t, \X, A=1, \Gamma) & = \Gamma^{-1} \E[ \YII | \X, A=1] \\
        & + (1 - \Gamma^{-1}) \E \biggl[ Q_\pm(\X, 1) + \frac{\{ \YII - Q_\pm(\X, 1) \}_\pm}{1-\gamma} \bigg| \X, A=1 \biggr].
    \end{align*}
    Under Assumptions~\ref{ass:sutva_survival} to \ref{ass:censoring_fun_indep_u}, and when the censoring function $\hat{\bar{G}}$ is consistently estimated, these bounds are doubly valid and doubly sharp with respect to the propensity scores and modified outcome regressions, as defined in \citet{dorn2024doubly}.
\end{theorem}

\section{Bounds for the RMST and Difference in RMST} \label{app:rmst}

We propose two methods to derive bounds for the RMST: by integrating the bounds for the survival function, or by deriving directly the bounds for the RMST.

\subsection{Bounds via Integration by the Rectangle Method} \label{app:rmst_bounds}

Bounds for the RMST can be obtained by integrating the bounds for the survival functions with respect to $t$ between 0 and a time horizon $\tau$. The lower and upper bounds of the partially identified set for the difference in RMST are therefore, respectively,
\begin{align}
    \Delta \mathrm{RMST}(\tau)_- & = \int_0^\tau S_{-}^{(1)}(t, \Gamma) \, \mathrm{d}t - \int_0^\tau S_{+}^{(0)}(t, \Gamma) \, \mathrm{d}t \quad \text{and} \label{eqn:RMST_PIS_lb} \\
    \Delta \mathrm{RMST}(\tau)_+ & = \int_0^\tau S_{+}^{(1)}(t, \Gamma) \, \mathrm{d}t - \int_0^\tau S_{-}^{(0)}(t, \Gamma) \, \mathrm{d}t, \label{eqn:RMST_PIS_ub}
\end{align}
where the information on the form (Form~I or II) in the index is dropped.

With an i.i.d.\ sample $\mathcal{D}$ as defined earlier in Section~\ref{sec:notations}, empirical estimators of Equations~\eqref{eqn:RMST_PIS_lb} and \eqref{eqn:RMST_PIS_ub} can be obtained by integrating empirical estimators of Equations~\eqref{eqn:PIS_estim_I} and \eqref{eqn:PIS_estim_II} using the rectangle method.

\subsection{Faster Bounds via Analytical Expression} \label{app:fast_rmst_bounds}

It is possible to show that, under Assumptions~\ref{ass:sutva_survival} to \ref{ass:censoring_fun_indep_u}, the RMST among the treated can be written
\begin{equation*}
    \mathrm{RMST}^{(1)}(\tau) = \E \biggl[ \frac{A \Delta^\tau}{e(\X, \U) \bar{G}(\tilde{T} \wedge \tau | \X, A=1)} \tilde{T} \wedge \tau \biggr],
\end{equation*}
with $\Delta^\tau = \indic(T \wedge \tau \leq C) = \Delta + (1 - \Delta) \indic(\tilde{T} \geq \tau) = \max(\Delta, \indic(\tilde{T} \geq \tau))$.

\begin{proof}
    \begin{align*}
        \mathrm{RMST}^{(1)}(\tau) & = \E[\Ttreat \wedge \tau]  \\
        & = \E \bigl[ \E[\Ttreat \wedge \tau | \X, \U] \bigr] \quad \text{by tower property} \\
        & = \E \biggl[ \frac{e(\X, \U)}{e(\X, \U)} \E [ \Ttreat \wedge \tau | \X, \U ] \biggr] \quad \text{by positivity (Assumption~\ref{ass:positivity_surv})} \\
        & = \E \biggl[ \frac{\E[A|\X, \U]}{e(\X, \U)} \E [ \Ttreat \wedge \tau | \X, \U ] \biggr] \quad \text{because $\E[A|\X, \U] = \Prob(A=1|\X, \U) = e(\X, \U)$} \\
        & = \E \biggl[ \E \biggl[ \frac{A}{e(\X, \U)} \Ttreat \wedge \tau \bigg| \X, \U \biggr] \biggr] \quad \text{by ($\X, \U$)-ignorability (Assumption~\ref{ass:xu-ignorability_survival})} \\
        & = \E \biggl[ \E \biggl[ \frac{A}{e(\X, \U)} \cdot \frac{\E[\indic(\Ttreat \wedge \tau \leq C) | \X, \U, A, \Ttreat]}{\E[\indic(\Ttreat \wedge \tau \leq C) | \X, \U, A, \Ttreat]} \Ttreat \wedge \tau \bigg| \X, \U \biggr] \biggr] \\
        & = \E \biggl[ \E \biggl[ \frac{A}{e(\X, \U)} \cdot \frac{\indic(\Ttreat \wedge \tau \leq C)}{\bar{G}(\Ttreat \wedge \tau | \X, \U, A)} \Ttreat \wedge \tau \bigg| \X, \U \biggr] \biggr] \quad \text{by tower property} \\
        & = \E \biggl[ \frac{A}{e(\X, \U)} \cdot \frac{\indic(T \wedge \tau \leq C)}{\bar{G}(T \wedge \tau | \X, \U, A)} T \wedge \tau \biggr] \quad \text{by consistency and tower property} \\
        & = \E \biggl[ \frac{A}{e(\X, \U)} \cdot \frac{\Delta^\tau}{\bar{G}(T \wedge \tau | \X, A)} T \wedge \tau \biggr] \quad \text{by Assumption~\ref{ass:censoring_fun_indep_u}}
    \end{align*}
    If $A = 1$ and $\Delta^\tau = 1$, $\tilde{T} \wedge \tau = \min(T, C) \wedge \tau = T \wedge \tau$. Therefore,
    \begin{equation*}
        \mathrm{RMST}^{(1)}(\tau) = \E \biggl[ \frac{A \Delta^\tau}{e(\X, \U) \bar{G}(\tilde{T} \wedge \tau | \X, A=1)} \tilde{T} \wedge \tau \biggr].
    \end{equation*}
\end{proof}

DVDS bounds for the RMST can be obtained by replacing $\YII$ with
\begin{equation*}
    \Yrmst = \frac{\Delta^\tau}{\bar{G}(\tilde{T} \wedge \tau | \X, A=1)} \tilde{T} \wedge \tau
\end{equation*}
in Theorem~\ref{theo:dvds_estim_II}. The resulting lower and upper bounds are respectively denoted $\mathrm{RMST}^{(1)}_-(\tau, \Gamma)$ and $\mathrm{RMST}^{(1)}_+(\tau, \Gamma)$. These bounds enjoy the same robustness properties as those cited in Theorem~\ref{theo:dvds_estim_II} and can be computed the same way as the bounds for Form~II.

\section{Experiments} \label{app:experiments}

\subsection{Complete Simulation Setup} \label{app:complete_sim_setup}

In our experiments, we considered $p_\mathbf{X} = 2$ observed confounders and $p_\mathbf{U} = 2$ unobserved confounders. We chose $\mathbf{X} \sim \mathcal{U}(-1, 1)^{p_\mathbf{X}}$, where $\mathcal{U}(a, b)^p$ is the uniform distribution on $(a, b)$ of dimension $p$. For $j$ in $[\![1, p_\mathbf{U}]\!]$, conditionally on $\mathbf{X}=\mathbf{x}$, we designed each unobserved confounder $U_j$ as
\begin{equation*}
    U_j|\mathbf{X}=\mathbf{x} \sim \lambda G + (1 - \lambda) \beta_j^T \mathbf{x} = \mathcal{N} \bigl( (1-\lambda) \beta_j^T \mathbf{x}, \, \lambda^2 \bigr),
\end{equation*}
with $G \sim \mathcal{N}(0, 1)$, $0 < \lambda < 1$, and $\beta_j = (\beta_{j,1}, \dots, \beta_{j,p_\mathbf{X}}) \in \mathbb{R}^{p_\mathbf{X}}_+$. We chose the distribution of $A$ conditionally on $\mathbf{X}$ and $\mathbf{U}$ to be a Bernoulli satisfying the MSM with a true sensitivity parameter $\Gamma^\star$. We designed the true propensity score $e(\mathbf{X}, \mathbf{U}) = \Prob(A=1|\mathbf{X}, \mathbf{U})$ such that it marginalizes to the nominal propensity score $e(\mathbf{X}) = \Prob(A=1|\mathbf{X}) = \mathrm{logistic}(\delta^T \mathbf{X} + 0.5)$, with $\delta \in \mathbb{R}^{p_\mathbf{X}}$. For simplicity, we considered that $\mathbf{U} = \sum_{j=1}^{p_\mathbf{U}} U_j / p_\mathbf{U}$.

To do so, we chose a true propensity score of the form
\begin{equation} \label{eqn:adversarial_prop_score}
    e(\mathbf{X}, \mathbf{U}) = l(\mathbf{X}) \cdot \mathds{1}(\mathbf{U} > t(\mathbf{X})) + u(\mathbf{X}) \cdot \mathds{1}(\mathbf{U} \leq t(\mathbf{X})),
\end{equation}
where
\begin{align*}
    & l(\mathbf{X}) = \frac{e(\mathbf{X})}{e(\mathbf{X}) + (1-e(\mathbf{X})) \Gamma^\star} \quad \text{and} \\
    & u(\mathbf{X}) = \frac{e(\mathbf{X})}{e(\mathbf{X}) + (1-e(\mathbf{X})) /\Gamma^\star}.
\end{align*}
We want to find the threshold $t(\mathbf{X})$ such that $\E[e(\mathbf{X}, \mathbf{U}) | \mathbf{X}] = e(\mathbf{X})$. Therefore, by applying the expectancy conditionally on $\mathbf{X}$ on Equation~\eqref{eqn:adversarial_prop_score} and using the previous constraint, we get
\begin{align*}
    \mathbb{P}(\mathbf{U} \leq t(\mathbf{X}) | \mathbf{X}) = \frac{e(\mathbf{X}) - l(\mathbf{X})}{u(\mathbf{X}) - l(\mathbf{X})}.
\end{align*}
As we assumed that $\mathbf{U} = \sum_{j=1}^{p_\mathbf{U}} U_j / p_\mathbf{U}$, the previous equation becomes
\begin{align*}
    \mathbb{P} \Bigl( \sum_j U_j / p_\mathbf{U} \leq t(\mathbf{X}) \Big| \mathbf{X} \Bigr) = \frac{e(\mathbf{X}) - l(\mathbf{X})}{u(\mathbf{X}) - l(\mathbf{X})}.
\end{align*}
As the distribution of $U_j|\mathbf{X}=\mathbf{x}$ is Gaussian, if we assume that the $U_j$ are mutually independent conditionally on $\mathbf{X}=\mathbf{x}$, the same applies to $\sum_j U_j / p_\mathbf{U}$ conditionally on $\mathbf{X}=\mathbf{x}$, so
\begin{align*}
    \sum_j U_j / p_\mathbf{U} \, | \, \mathbf{X}=\mathbf{x} \sim \mathcal{N} \Bigl( \sum_j (1-\lambda) \beta_j^T \mathbf{x} / p_\mathbf{U}, \, \lambda^2 / p_\mathbf{U} \Bigr).
\end{align*}
Therefore, the threshold $t(\mathbf{X})$ is the quantile of a normal distribution, which can be found by using the \texttt{qnorm} function:
\begin{equation*}
    t(\mathbf{X}) = \texttt{qnorm} \biggl( \frac{e(\mathbf{X}) - l(\mathbf{X})}{u(\mathbf{X}) - l(\mathbf{X})}, \, \texttt{mean} = \sum_j (1-\lambda) \beta_j^T \mathbf{x} / p_\mathbf{U}, \, \texttt{sd} = \lambda / \sqrt{p_\mathbf{U}} \biggr).
\end{equation*}

The potential outcome $T^{(a)}$, with $a$ in $\{0, 1\}$, was randomly generated using the inverse transform sampling method via a proportional hazards model with Weibull-distributed baseline hazard \citep{bender2005generating}. In particular,
\begin{equation*}
    T^{(a)} = 10 \cdot \biggl( \frac{- \log (U^{(a)})}{0.95 \cdot \exp \bigl( a\log(5) + \beta_\X^T \X + \beta_\U^T \U \bigr)} \biggr)^{1/1.8},
\end{equation*}
where $U^{(a)} \sim \mathcal{U}(0, 1)$. The event time $T$ was then defined as $T = \indic(A=1) \cdot \Ttreat + \indic(A=0) \cdot \Tcont$ and the censoring time $C$ followed a Weibull distribution of shape 6 and scale 10. Finally, we took the minimum between $T$ and $C$ to get the observed time $\tilde{T}$. We generated 20 Monte-Carlo samples, each of size $n = 1000$, and computed the sensitivity bounds for 10 values of time equally spaced between 0.1 and 9. Administrative censoring was applied after the 0.95-quantile of the observed times. Other parameter values used in the simulations are summarized in Table~\ref{tab:simu_param_values_chap4}.

\begin{table*}[h]
    \centering
    \begin{tabular}{|l|l|}
        \hline
        \textbf{Parameter} & \textbf{Value} \\
        \hline
        $p_\mathbf{X}$ & 2 \\
        $p_\mathbf{U}$ & 2 \\
        $n$ & 1000 \\
        $\Gamma^\star$ & 3 \\
        $\lambda$ & 0.9 \\
        $\beta_j$ & $\mathcal{U}(0, 1)^{p_\mathbf{X}}$ \\
        $\beta_\X$ & $\log(\mathcal{U}(1.1, 1.3)^{p_\mathbf{X}})$ \\
        $\beta_\U$ & $\log(\mathcal{U}(0.5, 0.8)^{p_\mathbf{U}})$ \\
        $\delta$ & $\mathcal{U}(-0.5, 0.5)^{p_\mathbf{X}}$ \\
        \hline
    \end{tabular}
    \caption[Parameter values used in the simulations.]{Parameter values used in the simulations. $\mathcal{U}(a, b)^p$ means that the vector of size $p$ was generated randomly according to a uniform distribution on $(a, b)$. $\log(\mathcal{U}(a, b)^p)$ means that the logarithm of each of the $p$ components in the random sample was taken.}
    \label{tab:simu_param_values_chap4}
\end{table*}

\subsection{Additional Information on Real Data} \label{app:real_data}

\subsubsection{German Breast Cancer Study Group Dataset}

We included 7 covariates in our analysis: \texttt{age} (the age, in years), \texttt{size} (the tumor size, in mm), \texttt{nodes} (the number of positive lymph nodes), \texttt{pgr} (the concentration of progesterone receptors, in fmol/L), \texttt{er} (the concentration of estrogen receptors, in fmol/L), \texttt{meno} (the menopausal status, i.e.\ 0 if premenopausal, and 1 if postmenopausal), and \texttt{grade} (the tumor grade, i.e.\ 1, 2, or 3). Categorical variables were binary encoded. The treatment $A$ is the \texttt{hormon} variable, the observed time $\tilde{T}$ is the \texttt{rfstime} variable, and the censoring status variable $\Delta$ is the \texttt{status} variable. More details on the dataset can be found on \href{https://www.kaggle.com/datasets/utkarshx27/breast-cancer-dataset-used-royston-and-altman}{https://www.kaggle.com/datasets/utkarshx27/breast-cancer-dataset-used-royston-and-altman}.

\subsubsection{Right Heart Catheterization Dataset}

The list and meaning of each variable in the original data can be found on\\
\href{https://hbiostat.org/data/repo/rhc}{https://hbiostat.org/data/repo/rhc}.

\subsection{Additional Results} \label{app:add_exp}

\subsubsection{Results on Synthetic Experiments}

\paragraph{Survival Function.}

Results of sensitivity analysis for the survival functions and the difference of survival functions are given in Figure~\ref{fig:simul_surv_fun_results} for our DVDS bounds with Form~I (in blue) and II (in green), and for the bounds from \citet{lee2024sensitivity} (in red). These results were obtained with the same setup as described in Section~\ref{sec:results_synth_exp}.

When $\Gamma = 3$, that is, the same value $\Gamma^\star$ that generated the data, the intervals almost always contain the true value of the survival or difference of survival functions. As pointed out in Section~\ref{sec:experiments_chap4}, remark that the DVDS bounds for Form~II of the survival function are unstable for high values of time $t$. In particular, the upper bound obtained on the control group increases even if this should not happen, because the survival function is a decreasing function over time (Figure~\ref{fig:control_surv_fun}). This can be explained by near violations of the positivity of the censoring process. On the contrary, the DVDS bounds for Form~I remain stable.

\begin{figure}[h!]
    \centering
    \subfigure[Among the treated.]{\includegraphics[width=0.49\textwidth]{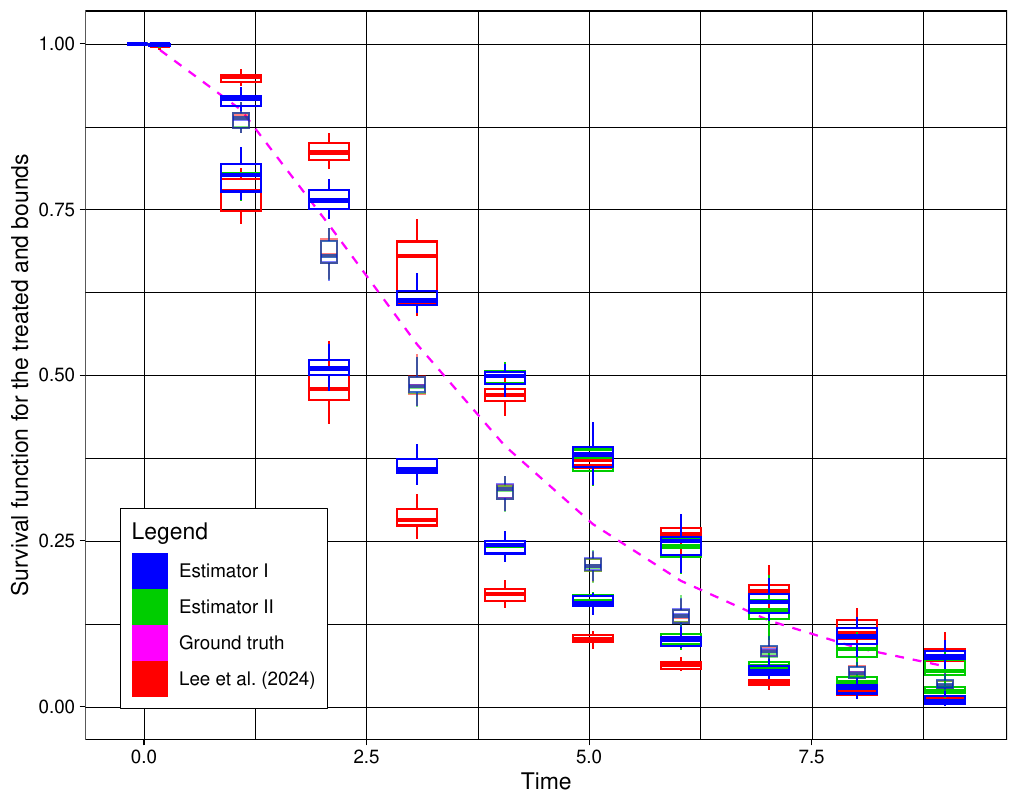}\label{fig:treated_surv_fun}}
    \hfill
    \subfigure[Among the control.]{\includegraphics[width=0.49\textwidth]{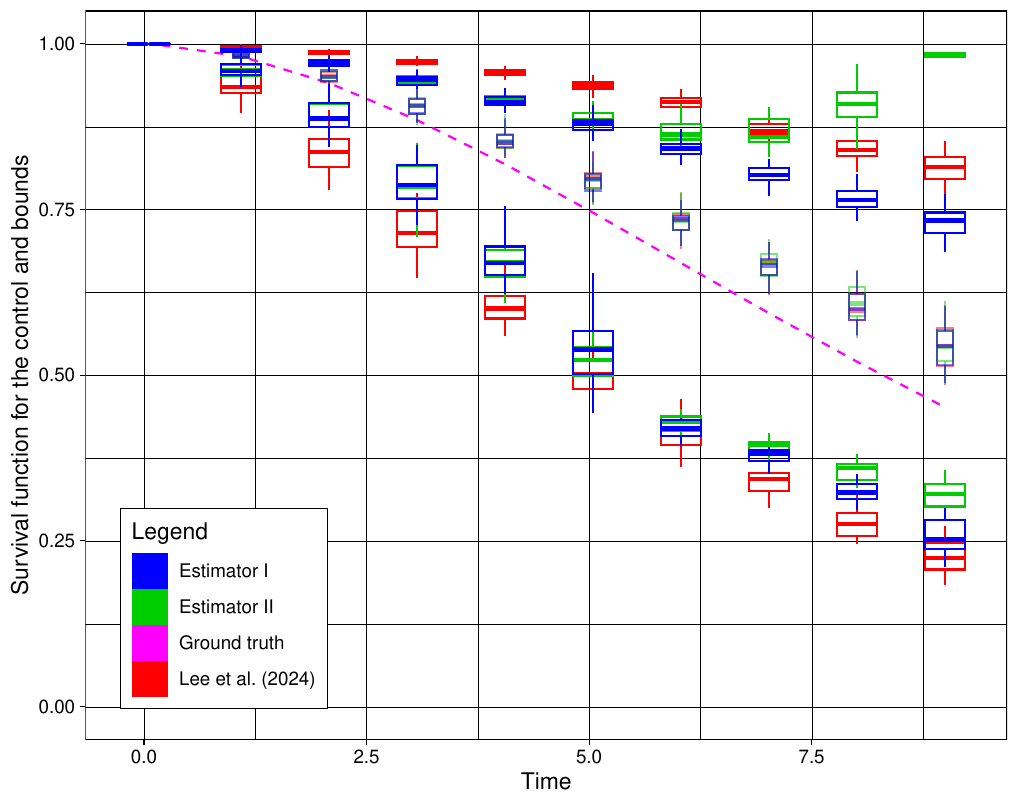}\label{fig:control_surv_fun}}
    \hfill
    \subfigure[Difference of survival functions.]{\includegraphics[width=0.49\textwidth]{simul_diff_surv_plot_lee_dvds_estim_I_and_II_v5_v13.pdf}\label{fig:diff_surv_fun}}
    \caption{Sensitivity analysis for the survival function among the treated (\ref{fig:treated_surv_fun}), the control (\ref{fig:control_surv_fun}) and the difference of survival functions (\ref{fig:diff_surv_fun}) \textbf{on the simulated data, as described in Section~\ref{sec:results_synth_exp}}. The dotted lines in magenta are the true survival functions and difference of survival functions. The large boxplots correspond to the estimated upper and lower sensitivity bounds (PEI) with Form~I (in blue), Form~II (in green), and the estimator from \citet{lee2024sensitivity} (in red) on 20 Monte-Carlo samples, for $\Gamma = 3$. The narrower and transparent boxplots correspond to the value under ignorability ($\Gamma = 1$) for each method.}
    \label{fig:simul_surv_fun_results}
\end{figure}

To check that our bounds are scalable to high dimensional datasets, we applied our approach to simulated data of sample size 5000 with 2 observed confounders and 2 unobserved confounders (Figure~\ref{fig:simul_surv_fun_hd_results}), and sample size 3000, with 40 observed confounders and 20 unobserved confounders (Figure~\ref{fig:simul_surv_fun_hd_2_results}). In both cases, we used 100 Monte Carlo samples. We parallelized the implementation on 9 CPUs and obtained an execution time of 9.74 hours for the difference of survival function in the second experiment with high dimensional variables. Our results show that the method is still valid, with conclusions that are similar to what we presented in Section~\ref{sec:results_synth_exp} and with boxplots of a similar width. Therefore, we conclude that our method is scalable to high dimensional datasets. Of course, as there are four nuisance parameters to estimate (propensity score, modified outcome regression, conditional quantile, conditional censoring function), each depending on the observed confounders $\X$, the higher the dimension of $\X$ is, the longer the execution time will be.

\begin{figure}[h!]
    \centering
    \subfigure[Among the treated.]{\includegraphics[width=0.49\textwidth]{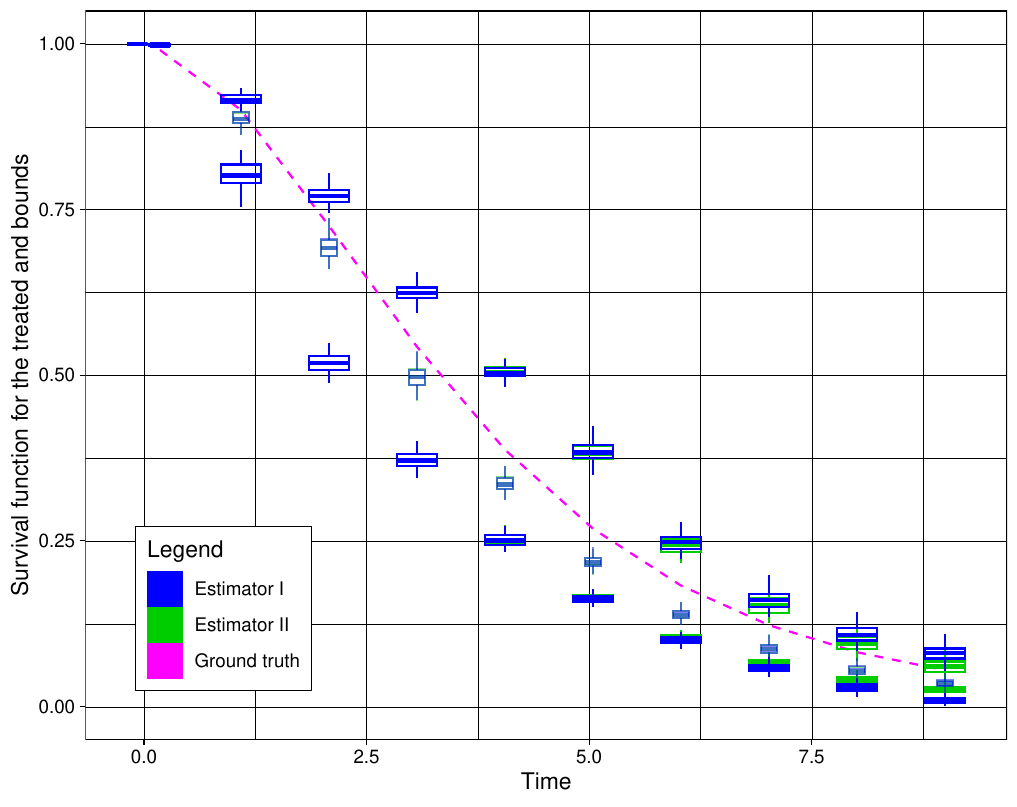}\label{fig:treated_surv_fun_hd}}
    \hfill
    \subfigure[Among the control.]{\includegraphics[width=0.49\textwidth]{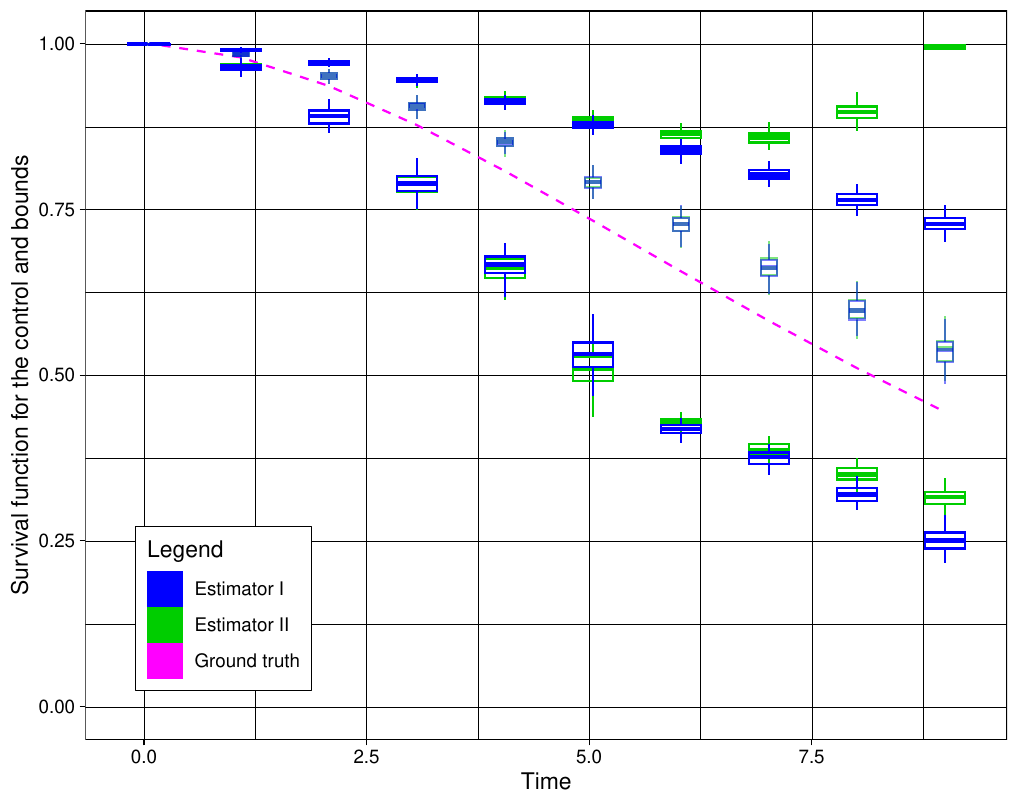}\label{fig:control_surv_fun_hd}}
    \hfill
    \subfigure[Difference of survival functions.]{\includegraphics[width=0.49\textwidth]{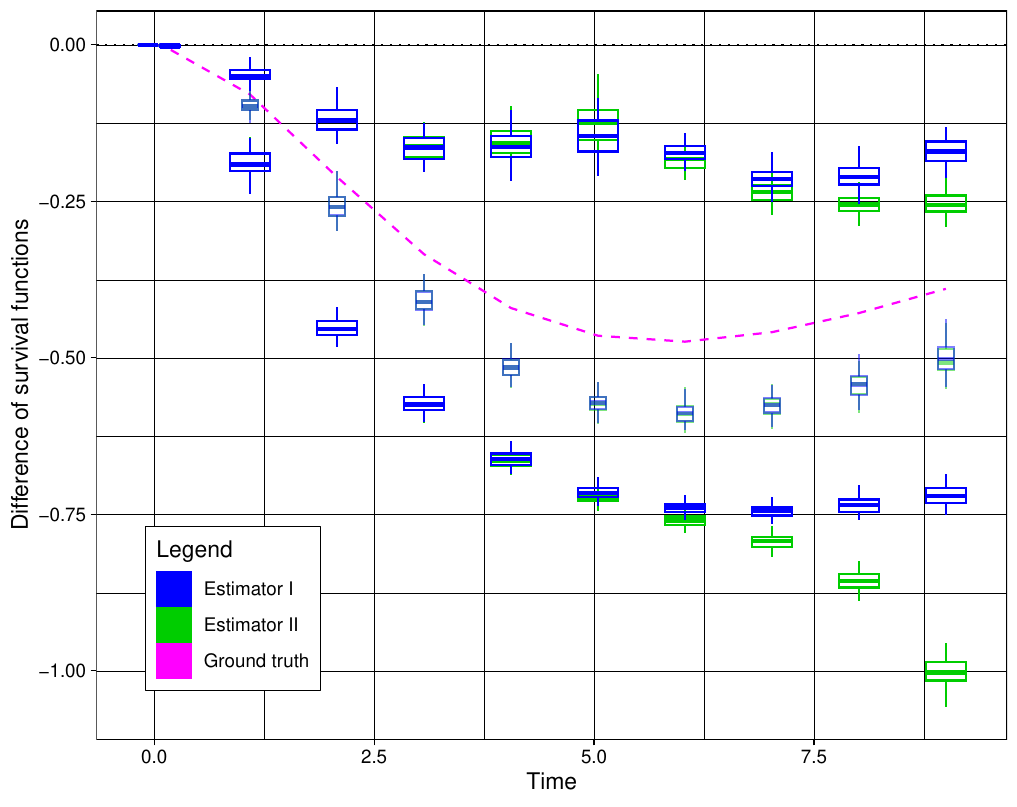}\label{fig:diff_surv_fun_hd}}
    \caption{Sensitivity analysis for the survival function among the treated (\ref{fig:treated_surv_fun_hd}), the control (\ref{fig:control_surv_fun_hd}) and the difference of survival functions (\ref{fig:diff_surv_fun_hd}) \textbf{on high dimensional simulated data ($n = 5000$, $p_\X = 2$, and $p_\U = 2$)}. The dotted lines in magenta are the true survival functions and difference of survival functions. The large boxplots correspond to the estimated upper and lower sensitivity bounds (PEI) with Form~I (in blue) and Form~II (in green) on 100 Monte-Carlo samples, for $\Gamma = 3$. The narrower and transparent boxplots correspond to the value under ignorability ($\Gamma = 1$) for each method.}
    \label{fig:simul_surv_fun_hd_results}
\end{figure}

\begin{figure}[h!]
    \centering
    \subfigure[Among the treated.]{\includegraphics[width=0.49\textwidth]{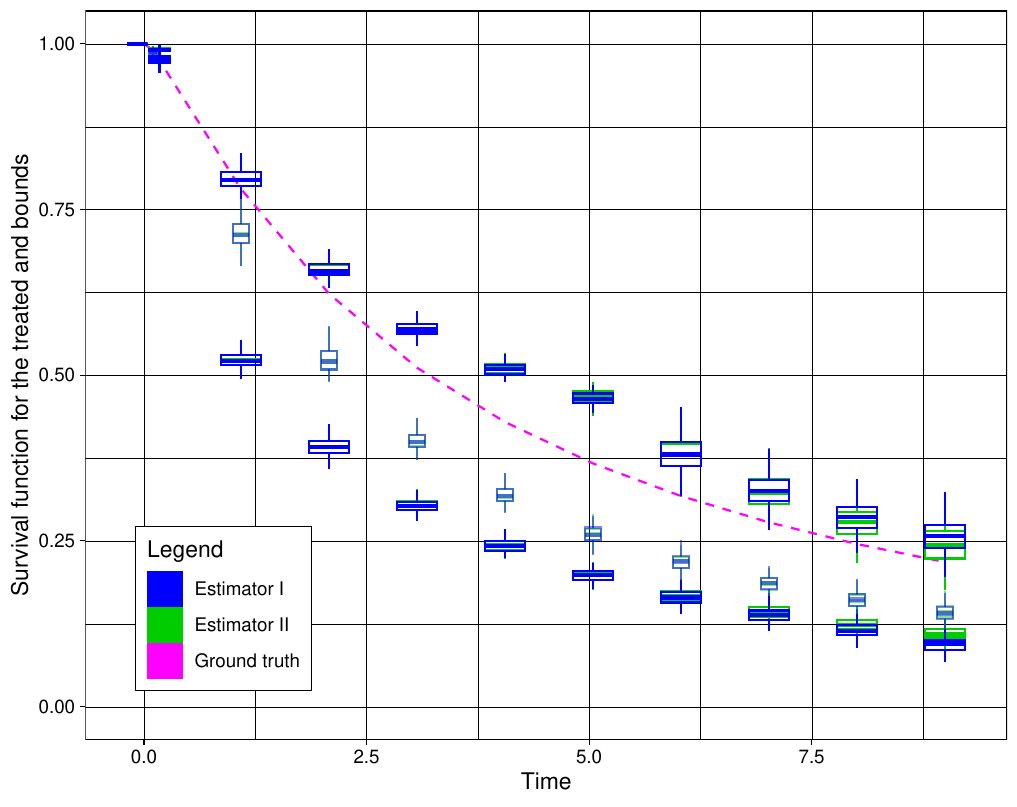}\label{fig:treated_surv_fun_hd_2}}
    \hfill
    \subfigure[Among the control.]{\includegraphics[width=0.49\textwidth]{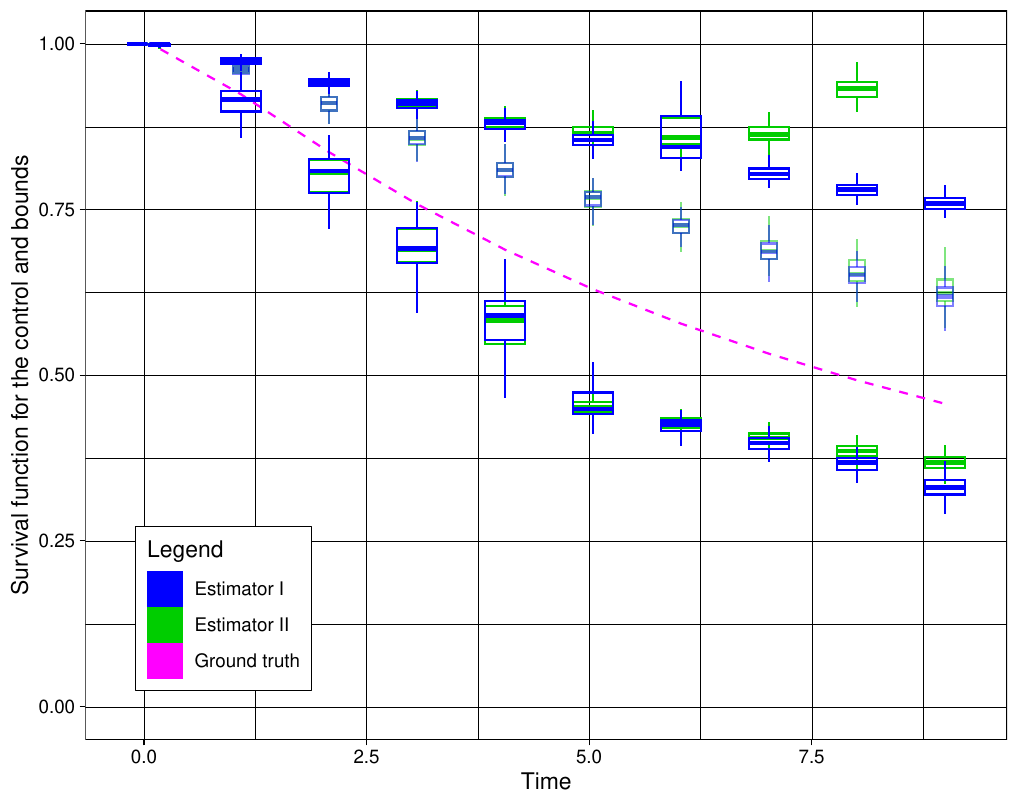}\label{fig:control_surv_fun_hd_2}}
    \hfill
    \subfigure[Difference of survival functions.]{\includegraphics[width=0.49\textwidth]{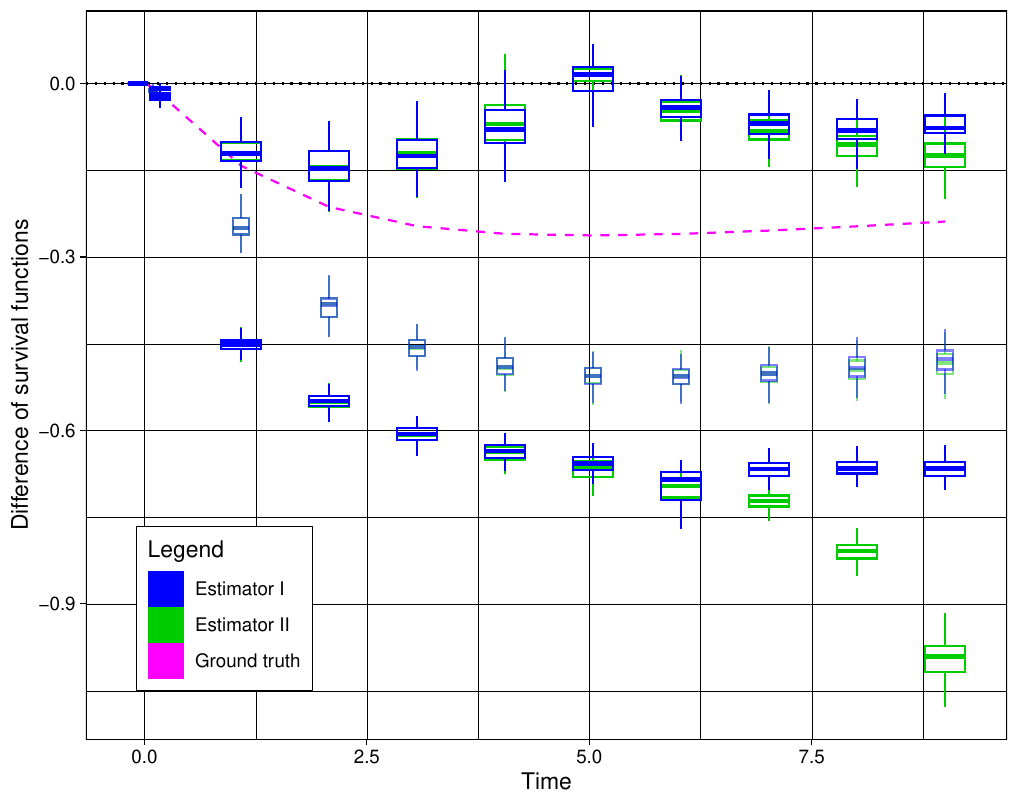}\label{fig:diff_surv_fun_hd_2}}
    \caption{Sensitivity analysis for the survival function among the treated (\ref{fig:treated_surv_fun_hd_2}), the control (\ref{fig:control_surv_fun_hd_2}) and the difference of survival functions (\ref{fig:diff_surv_fun_hd_2}) \textbf{on high dimensional simulated data ($n = 3000$, $p_\X = 40$, and $p_\U = 20$)}. The dotted lines in magenta are the true survival functions and difference of survival functions. The large boxplots correspond to the estimated upper and lower sensitivity bounds (PEI) with Form~I (in blue) and Form~II (in green) on 100 Monte-Carlo samples, for $\Gamma = 3$. The narrower and transparent boxplots correspond to the value under ignorability ($\Gamma = 1$) for each method.}
    \label{fig:simul_surv_fun_hd_2_results}
\end{figure}

We also added two experiments where the true sensitivity parameter $\Gamma^\star$ is equal to 1.01 (near unconfoundedness) and 5 (higher confounding strength), and we kept the sensitivity parameter $\Gamma$ given to our method equal to 3. As expected, under near unconfoundedness ($\Gamma^\star = 1.01$), our methodology and the one from \citet{lee2024sensitivity} are less biased: in Figure~\ref{fig:simul_surv_fun_gamma_1_results}, the boxplots of the values under unconfoundedness are closer to the true estimand and the bounds are centered around these true values. When the true sensitivity parameter is higher ($\Gamma^\star = 5$) but $\Gamma = 3$ (Figure~\ref{fig:simul_surv_fun_gamma_5_results}), the boxplots of the values under unconfoundedness move away from the true value. If $\Gamma$ was too small (for instance, 1.5), many intervals would not contain the true estimand anymore. Therefore, $\Gamma$ should be large enough for the bounds to contain the true treatment effect.

\begin{figure}[h!]
    \centering
    \subfigure[Among the treated.]{\includegraphics[width=0.49\textwidth]{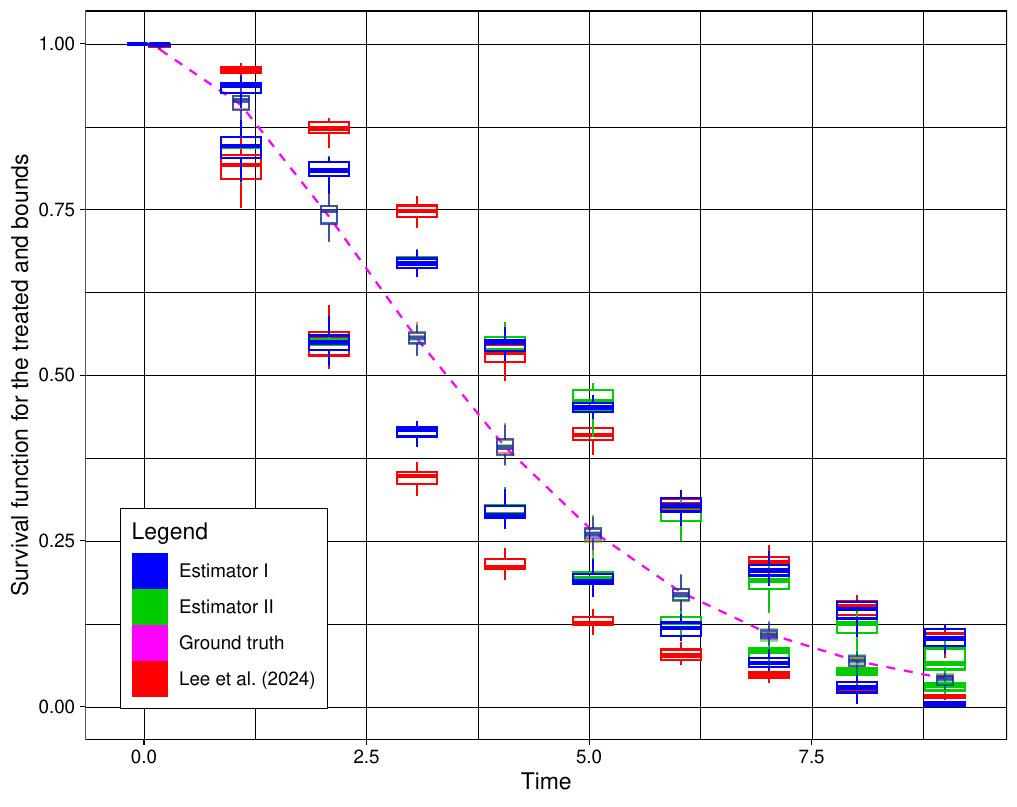}\label{fig:treated_surv_fun_gamma_1}}
    \hfill
    \subfigure[Among the control.]{\includegraphics[width=0.49\textwidth]{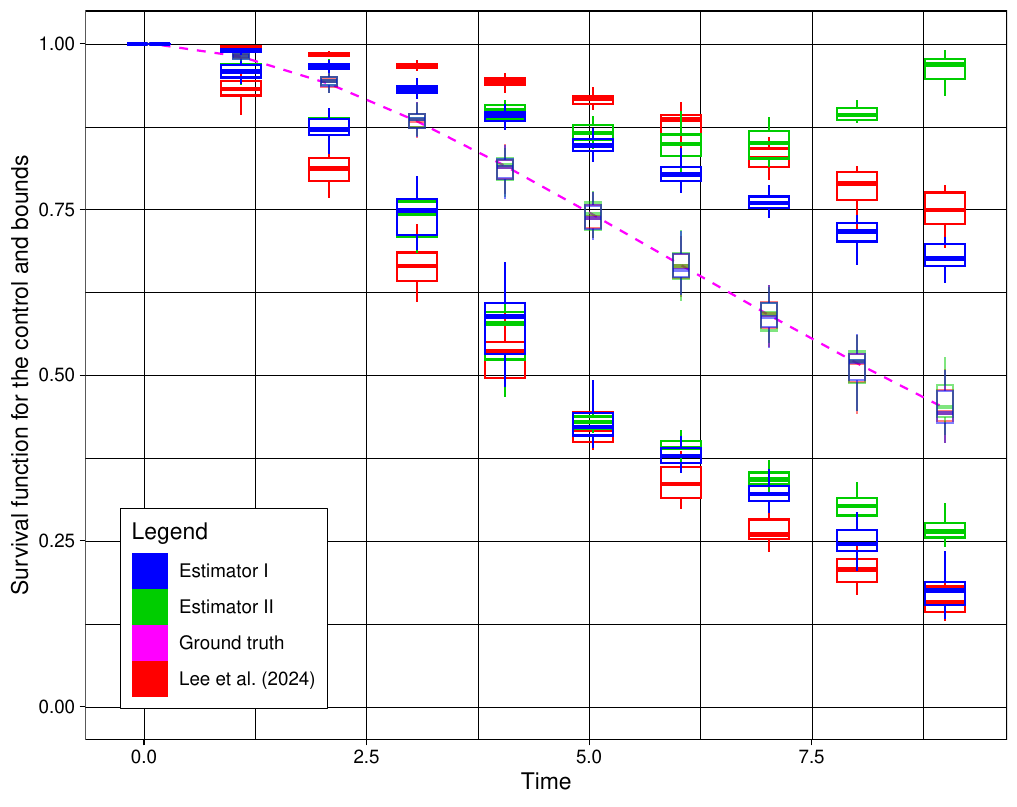}\label{fig:control_surv_fun_gamma_1}}
    \hfill
    \subfigure[Difference of survival functions.]{\includegraphics[width=0.49\textwidth]{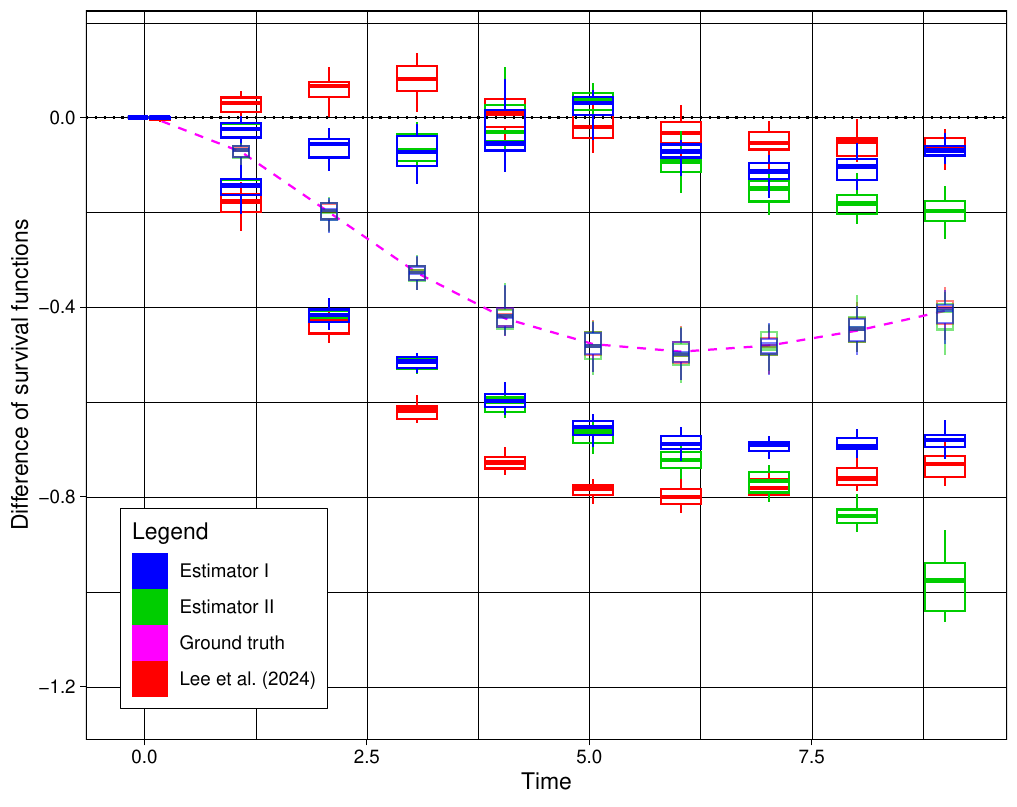}\label{fig:diff_surv_fun_gamma_1}}
    \caption{Sensitivity analysis for the survival function among the treated (\ref{fig:treated_surv_fun_gamma_1}), the control (\ref{fig:control_surv_fun_gamma_1}) and the difference of survival functions (\ref{fig:diff_surv_fun_gamma_1}) \textbf{on the simulated data when the true sensitivity parameter is $\Gamma^\star = 1.01$}. The dotted lines in magenta are the true survival functions and difference of survival functions. The large boxplots correspond to the estimated upper and lower sensitivity bounds (PEI) with Form~I (in blue), Form~II (in green), and the method from \citet{lee2024sensitivity} (in red) on 20 Monte-Carlo samples, for $\Gamma = 3$. The narrower and transparent boxplots correspond to the value under ignorability ($\Gamma = 1$) for each method.}
    \label{fig:simul_surv_fun_gamma_1_results}
\end{figure}

\begin{figure}[h!]
    \centering
    \subfigure[Among the treated.]{\includegraphics[width=0.49\textwidth]{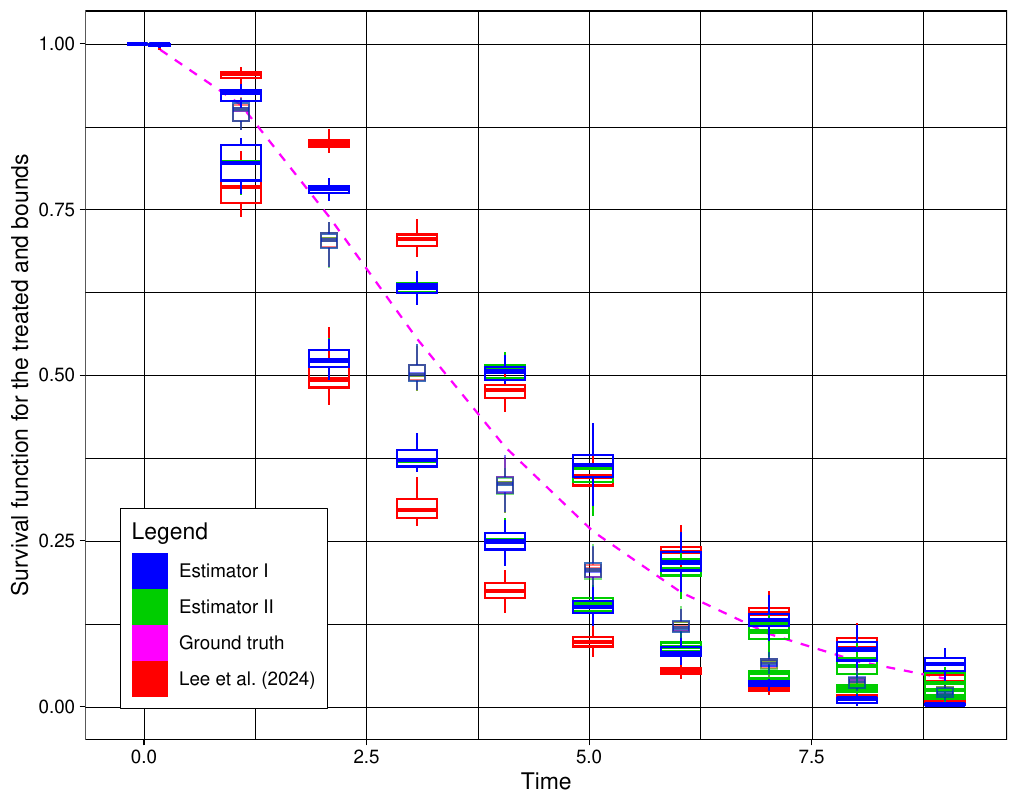}\label{fig:treated_surv_fun_gamma_5}}
    \hfill
    \subfigure[Among the control.]{\includegraphics[width=0.49\textwidth]{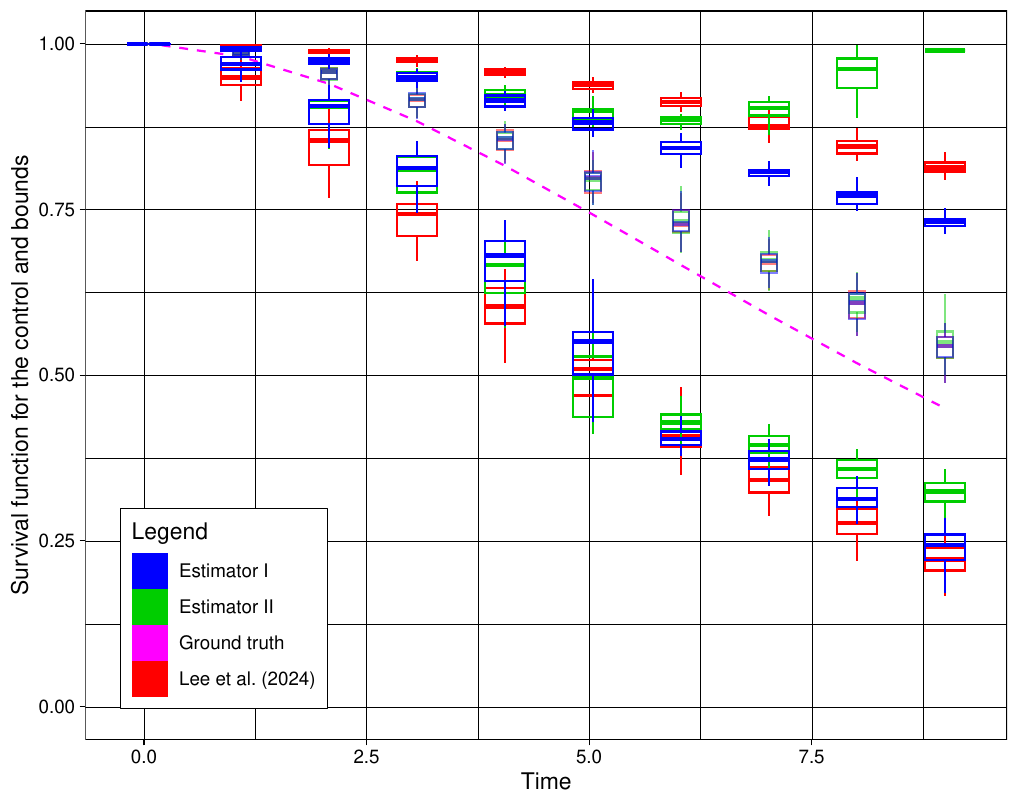}\label{fig:control_surv_fun_gamma_5}}
    \hfill
    \subfigure[Difference of survival functions.]{\includegraphics[width=0.49\textwidth]{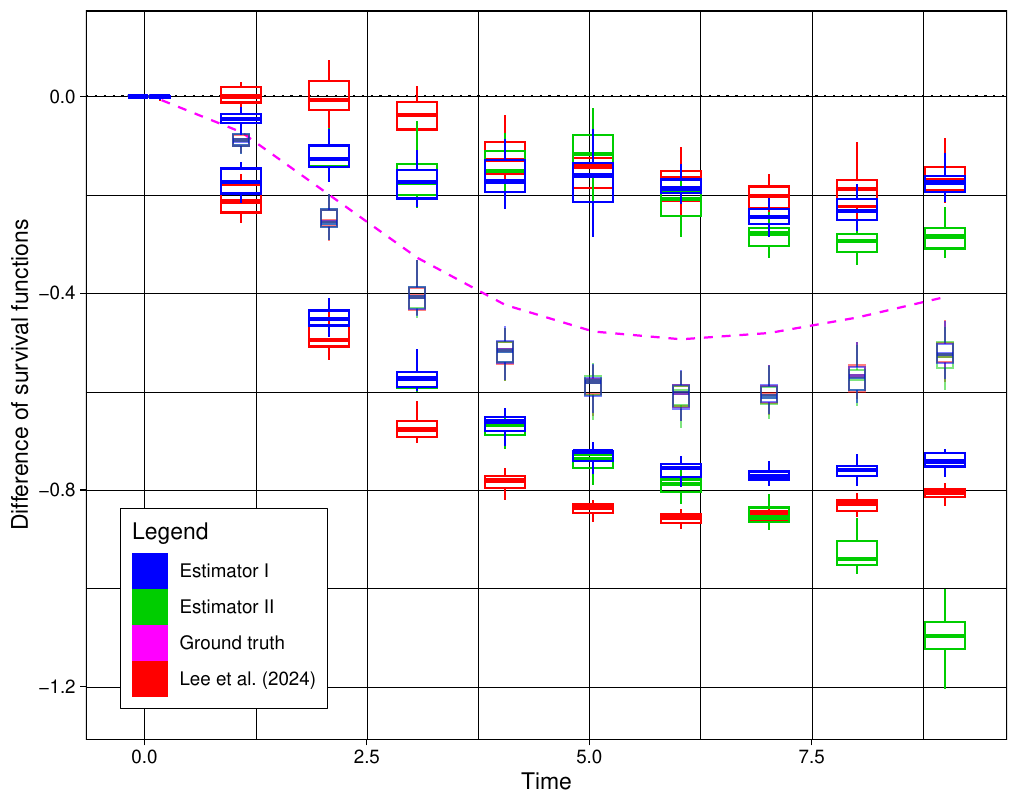}\label{fig:diff_surv_fun_gamma_5}}
    \caption{Sensitivity analysis for the survival function among the treated (\ref{fig:treated_surv_fun_gamma_5}), the control (\ref{fig:control_surv_fun_gamma_5}) and the difference of survival functions (\ref{fig:diff_surv_fun_gamma_5}) \textbf{on the simulated data when the true sensitivity parameter is $\Gamma^\star = 5$}. The dotted lines in magenta are the true survival functions and difference of survival functions. The large boxplots correspond to the estimated upper and lower sensitivity bounds (PEI) with Form~I (in blue), Form~II (in green), and the method from \citet{lee2024sensitivity} (in red) on 20 Monte-Carlo samples, for $\Gamma = 3$. The narrower and transparent boxplots correspond to the value under ignorability ($\Gamma = 1$) for each method.}
    \label{fig:simul_surv_fun_gamma_5_results}
\end{figure}

Finally, we compared the methods in the case of informative censoring, where the censoring time $C$ depends on $\X$ and $A$ (and we could also add $\U$). We do not present the bounds obtained with Form~II because they diverged. $C$ was generated the same way as the event time $T$, with
\begin{equation*}
    C^{(a)} = 7 \cdot \biggl( \frac{- \log (U^{(a)})}{0.95 \cdot \exp \bigl( a\log(5) + \beta_{\X, C}^T \X \bigr)} \biggr)^{1/6},
\end{equation*}
for $a \in \{0, 1\}$, where $U^{(a)} \sim \mathcal{U}(0, 1)$ and $\beta_{\X, C} = \mathcal{U}(5.0, 5.5)^{p_\mathbf{X}}$. Figure~\ref{fig:simul_surv_fun_inform_cens_results} shows that both methods stay valid when $\Gamma^\star = \Gamma = 3$, but the boxplots of the two methods do not overlap under ignorability, as expected. Note that our method is less biased when the estimand is the difference of survival functions (Figure~\ref{fig:diff_surv_fun_inform_cens}), but at the cost of potentially larger intervals in practice. Indeed, we can see that, for high values of time, such as around 7.5 time units, our method yields larger intervals than the method of \citet{lee2024sensitivity}. This can be due to the fact that our bounds do not include a stabilization with respect to the censoring function $\bar{G}$, while the bounds of \citet{lee2024sensitivity} do not include $\bar{G}$ at all.

\begin{figure}[h!]
    \centering
    \subfigure[Among the treated.]{\includegraphics[width=0.49\textwidth]{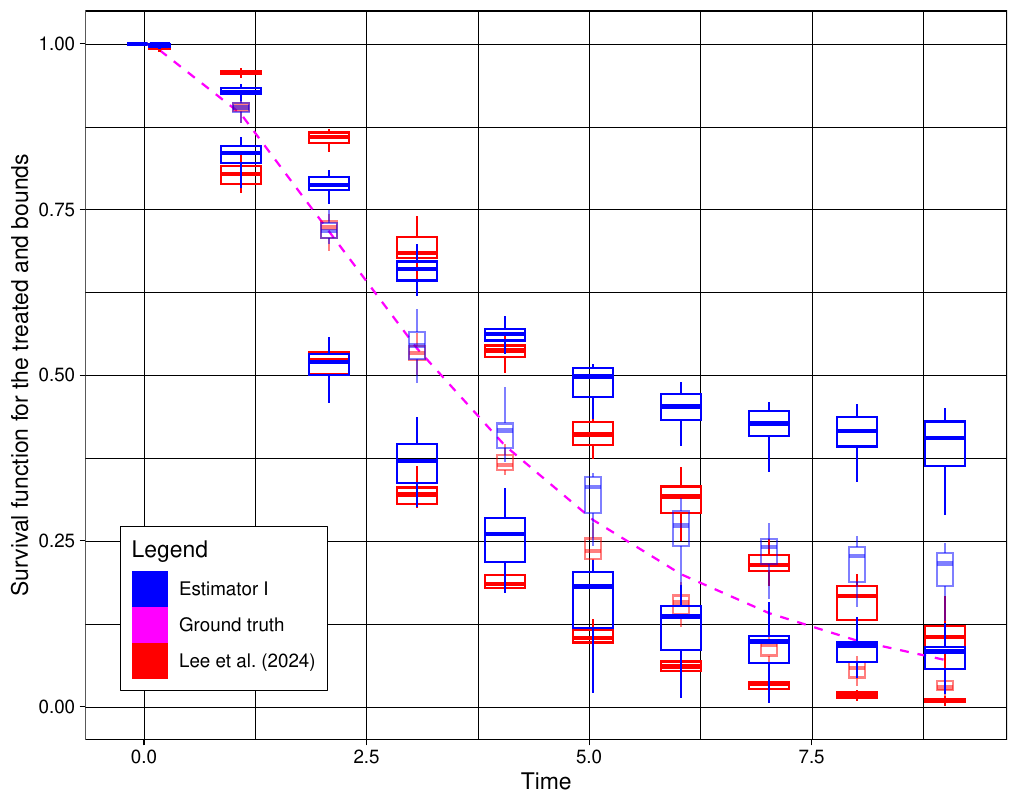}\label{fig:treated_surv_fun_inform_cens}}
    \hfill
    \subfigure[Among the control.]{\includegraphics[width=0.49\textwidth]{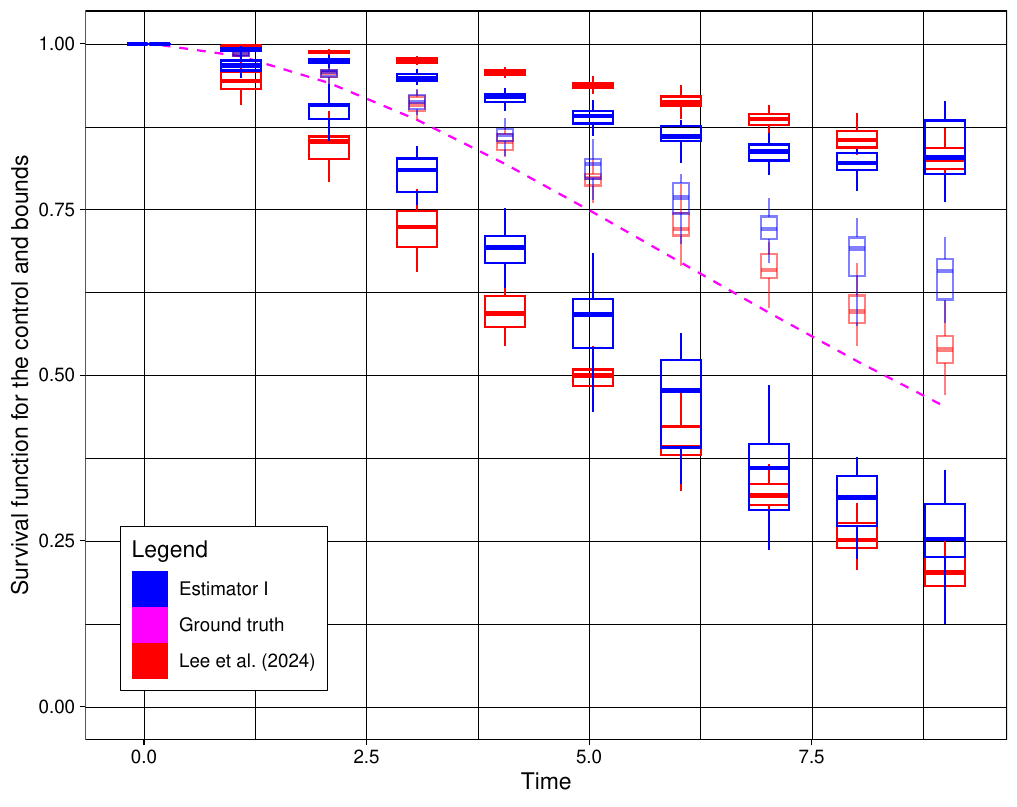}\label{fig:control_surv_fun_inform_cens}}
    \hfill
    \subfigure[Difference of survival functions.]{\includegraphics[width=0.49\textwidth]{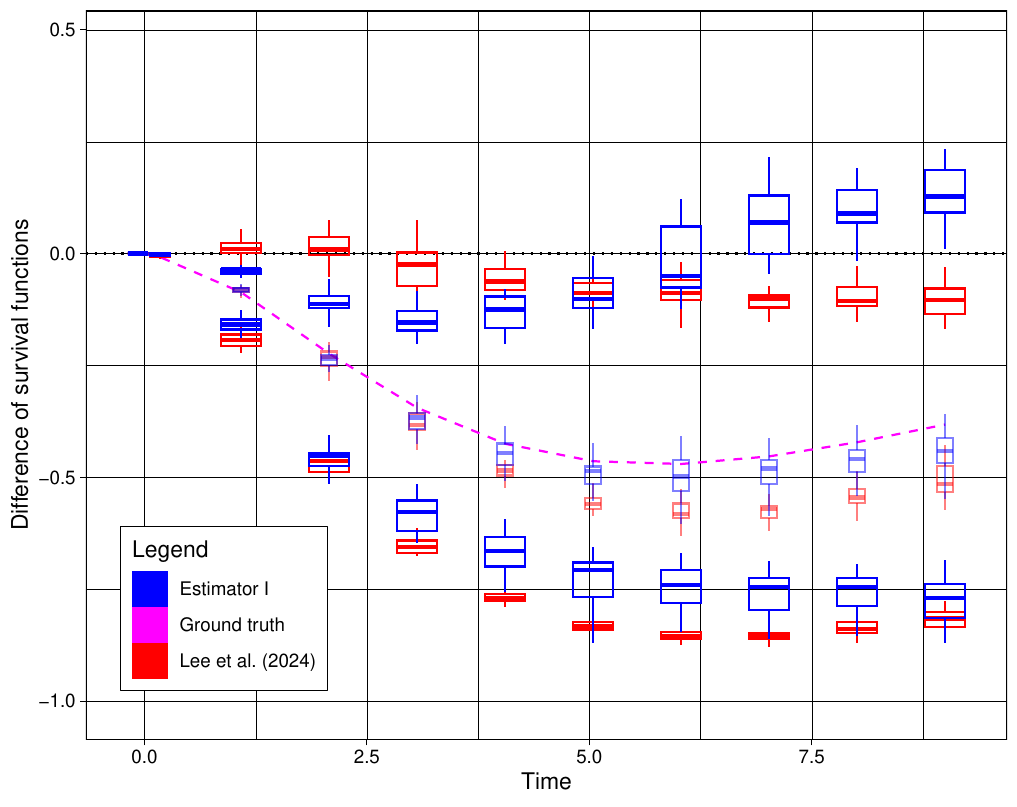}\label{fig:diff_surv_fun_inform_cens}}
    \caption{Sensitivity analysis for the survival function among the treated (\ref{fig:treated_surv_fun_inform_cens}), the control (\ref{fig:control_surv_fun_inform_cens}) and the difference of survival functions (\ref{fig:diff_surv_fun_inform_cens}) \textbf{on the simulated data in the presence of informative censoring}. The dotted lines in magenta are the true survival functions and difference of survival functions. The large boxplots correspond to the estimated upper and lower sensitivity bounds (PEI) with Form~I (in blue) and the method from \citet{lee2024sensitivity} (in red) on 20 Monte-Carlo samples, for $\Gamma = 3$. The narrower and transparent boxplots correspond to the value under ignorability ($\Gamma = 1$) for each method.}
    \label{fig:simul_surv_fun_inform_cens_results}
\end{figure}

\paragraph{RMST and Difference in RMST.}

We parallelized computations between different values of time horizon $\tau$ for our DVDS bounds.

For the RMST, \citet{lee2024sensitivity} originally computed the area under the survival curve via the left point rectangle method whereas we implemented the middle point rectangle method because it converges slightly faster towards the true integral. To compute the area under the curve, we used all observed time points between 0 and the time horizon $\tau$.

Results of sensitivity analysis for the RMST and the difference in RMST are given in Figure~\ref{fig:simul_rmst_results} for our DVDS bounds (Form~I in blue and Form~II in green) and the bounds from \citet{lee2024sensitivity} (in red). Under ignorability (narrow and transparent boxplots), the estimations deviate slightly from the ground truth (dotted lines in magenta) because of unobserved confounders, but are still close to the true value. Note that the three methods give approximately the same estimation.

When $\Gamma = 3$, that is, the same value $\Gamma^\star$ that generated the data, the intervals almost always contain the true value of the RMST or difference in RMST.

The sharpness property is less apparent with the difference in RMST (Figure~\ref{fig:diff_rmst}) than with the difference in survival functions (Figure~\ref{fig:diff_surv_fun}), especially with high values of time $t$, where an offset between our bounds and the bounds of the comparator appears. This could be explained by two facts: our DVDS bounds were written for the survival functions and then integrated with a rectangle method, whereas \citet{lee2024sensitivity} directly derived bounds for the RMST; moreover, to simplify, the simulation setup was based on the assumption of non-informative censoring, which was supposed by \citet{lee2024sensitivity}, whereas we assumed informative censoring and introduced $\bar{G}$.

\begin{figure}[h!]
     \centering
     \subfigure[Among the treated.]{\includegraphics[width=0.49\textwidth]{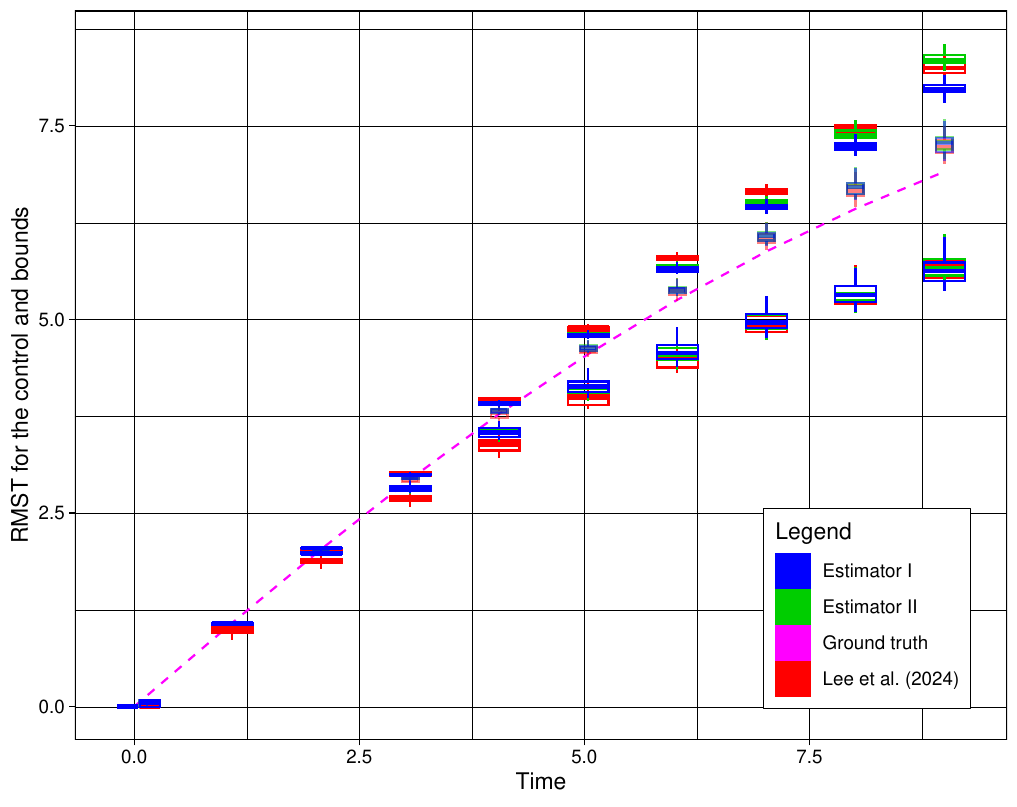}\label{fig:treated_rmst}}
     \hfill
     \subfigure[Among the control.]{\includegraphics[width=0.49\textwidth]{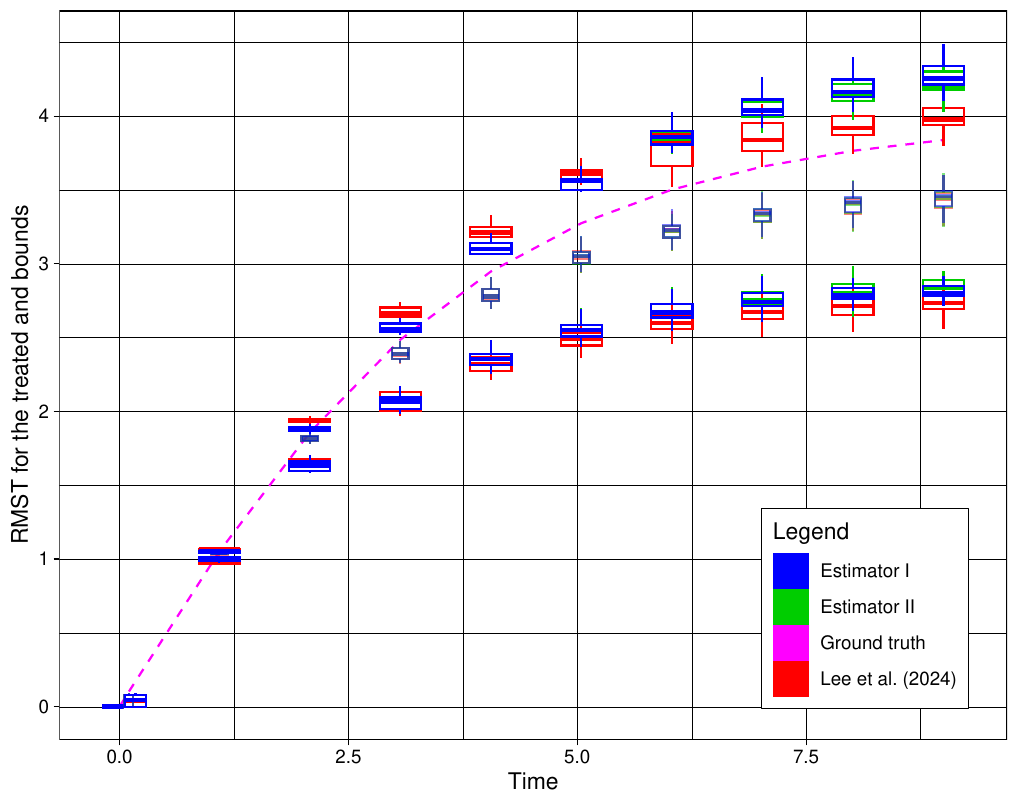}\label{fig:control_rmst}}
     \hfill
     \subfigure[Difference of RMST.]{\includegraphics[width=0.49\textwidth]{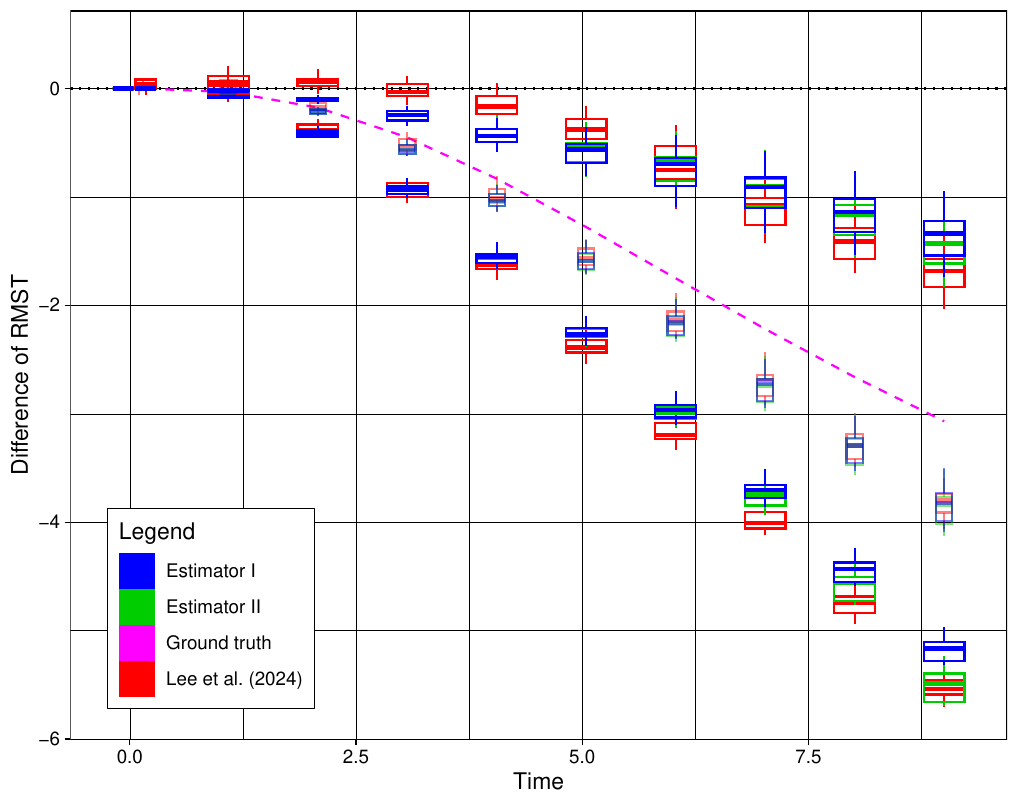}\label{fig:diff_rmst}}
    \caption{Sensitivity analysis for the RMST among the treated (\ref{fig:treated_rmst}), the control (\ref{fig:control_rmst}) and the difference in RMST (\ref{fig:diff_rmst}) on the simulated data. The dotted lines in magenta are the true RMST and difference in RMST. The large boxplots correspond to the estimated upper and lower sensitivity bounds (PEI) \textbf{via integration by the rectangle method} with Form~I (in blue) and Form~II (in green), and with the method from \citet{lee2024sensitivity} (in red) on 20 Monte-Carlo samples, for $\Gamma = 3$. The narrower and transparent boxplots correspond to the value under ignorability ($\Gamma = 1$) for each method.}
    \label{fig:simul_rmst_results}
\end{figure}

Our DVDS bounds, and especially the ones for Form~I, show good coverage of the real (difference in) RMST, and sharpness with respect to the bounds from \citet{lee2024sensitivity}. Moreover, our bounds also enjoy shorter computation times, but, as compared with bounds for the survival function, the computation time gap is less important because of the integration by the rectangle method and the need to compute the PEIs \textit{for each observed time point} between 0 and the time horizon $\tau$ with our method. The method by \citet{lee2024sensitivity} took 9 hours and 54 minutes to complete, whereas our DVDS bounds (for Forms~I and II simultaneously) required 8 hours and 48 minutes. Our method is still more rapid but the gain factor is only 1.1. Note that \citet{lee2024sensitivity} developed an approximate method which is faster than direct optimization, but at the cost of simplifying assumptions.

As seen in Appendix~\ref{app:fast_rmst_bounds}, it is possible to implement faster DVDS bounds for the RMST (a few minutes with our bounds versus several hours with the direct optimization method from \citet{lee2024sensitivity}). However, in practice, they can perform poorly because they are more conservative than the bounds obtained via integration of the survival function. Indeed, in Figure~\ref{fig:simul_fast_rmst_results}, we can see that our bounds are tighter than the bounds from \citet{lee2024sensitivity} until a certain time horizon, from which several observations start to be censored. We tried to stabilize the bounds, with no improvement.

\begin{figure}[h!]
     \centering
     \subfigure[Among the treated.]{\includegraphics[width=0.49\textwidth]{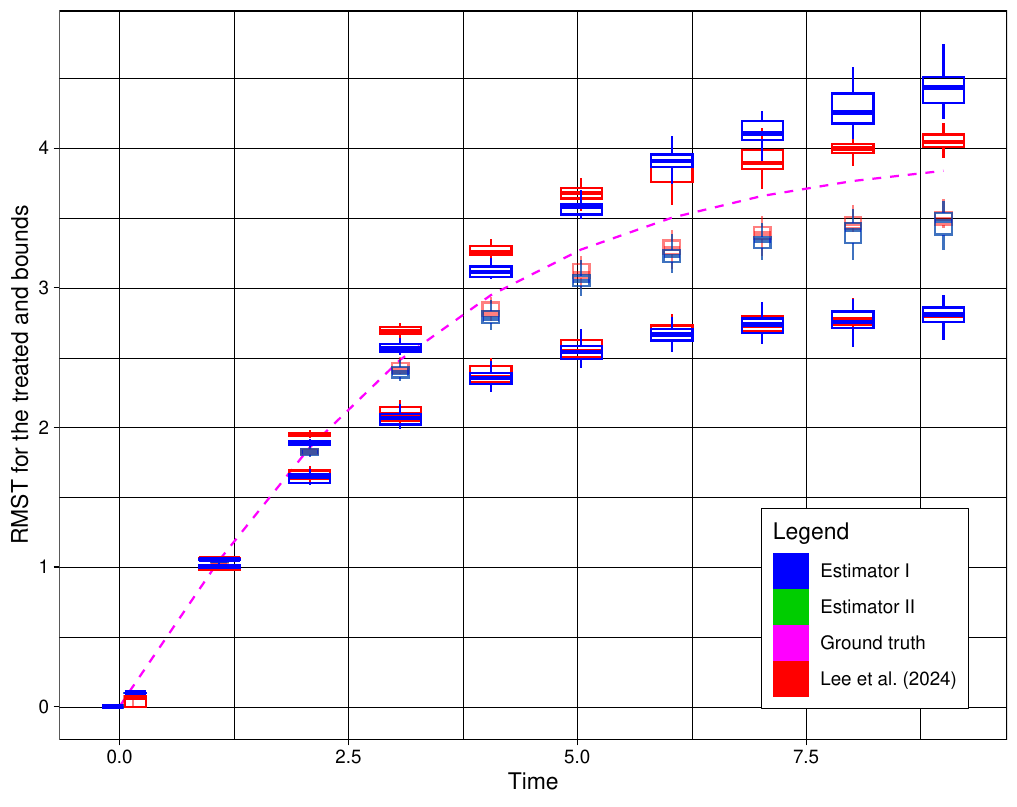}\label{fig:treated_fast_rmst}}
     \hfill
     \subfigure[Among the control.]{\includegraphics[width=0.49\textwidth]{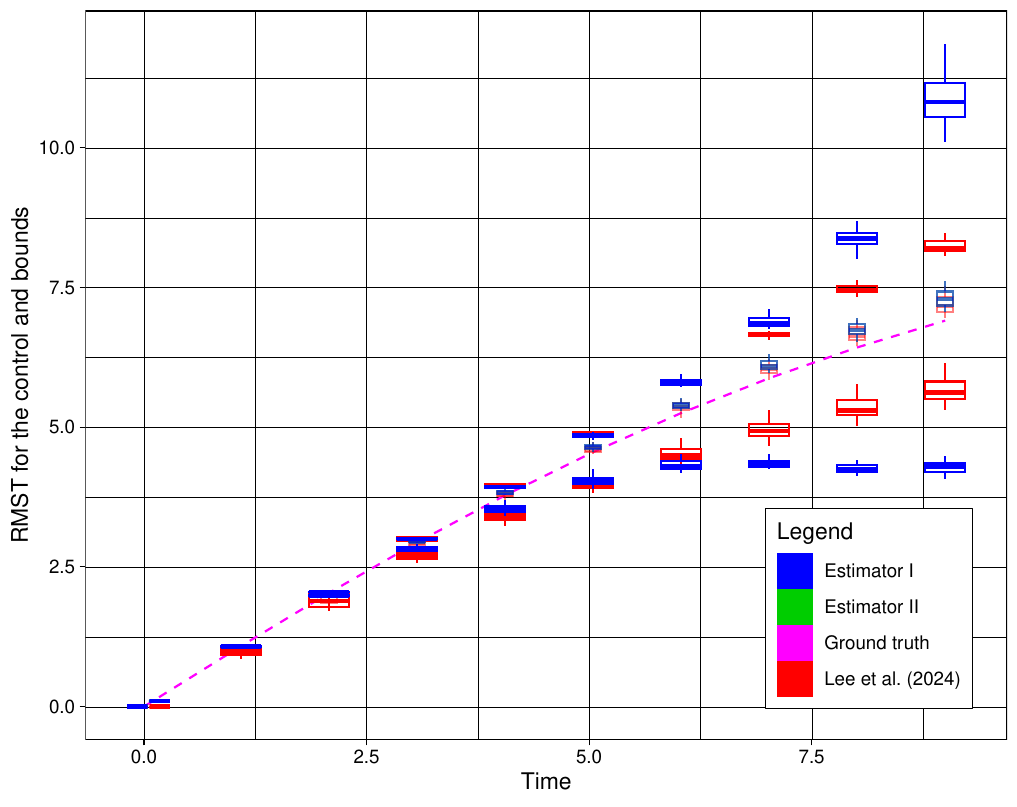}\label{fig:control_fast_rmst}}
     \hfill
     \subfigure[Difference in RMST.]{\includegraphics[width=0.49\textwidth]{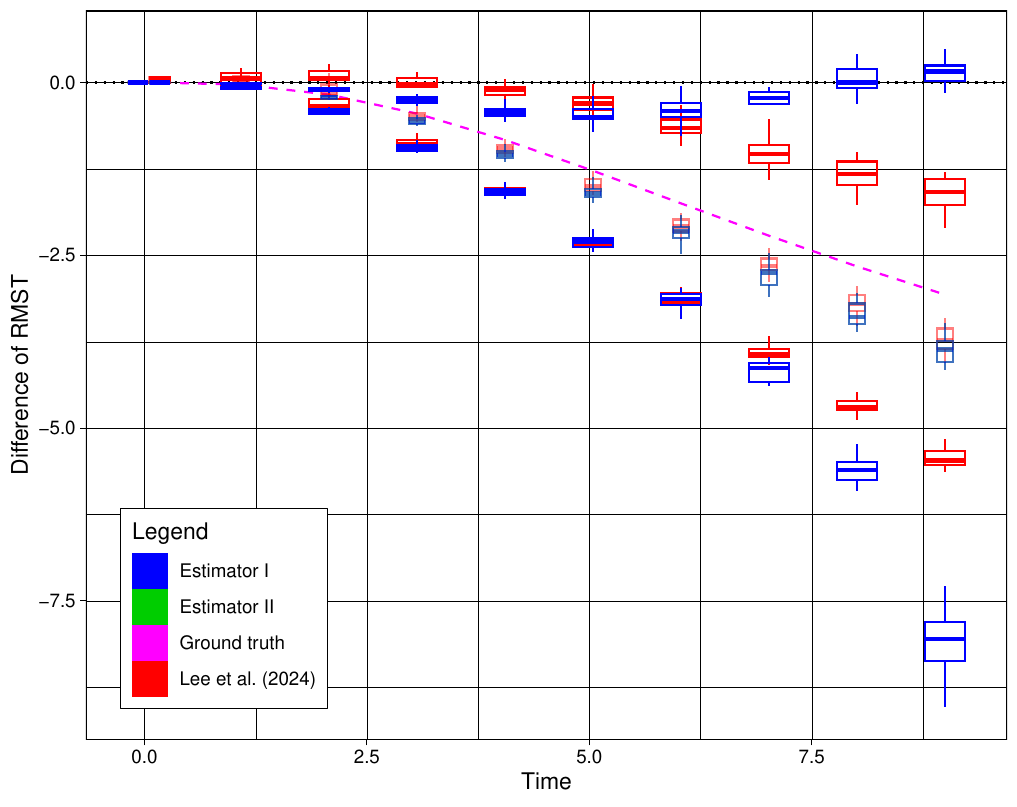}\label{fig:diff_fast_rmst}}
    \caption{Sensitivity analysis for the RMST among the treated (\ref{fig:treated_fast_rmst}), the control (\ref{fig:control_fast_rmst}) and the difference in RMST (\ref{fig:diff_fast_rmst}) on the simulated data. The dotted lines in magenta are the true RMST and difference in RMST. The large boxplots correspond to the estimated upper and lower sensitivity bounds (PEI) \textbf{with our ``fast RMST" method} on Form~I (in blue) and Form~II (in green), and with the method from \citet{lee2024sensitivity} (in red) on 20 Monte-Carlo samples, for $\Gamma = 3$. The narrower and transparent boxplots correspond to the value under ignorability ($\Gamma = 1$) for each method.}
    \label{fig:simul_fast_rmst_results}
\end{figure}

\subsubsection{Results on Real Data}

\paragraph{Calibration of $\Gamma$ by informal benchmarking.}

We calibrated the sensitivity parameter $\Gamma$ using an informal benchmarking approach where we took the maximum between the maximum odds ratio of the MSM and the inverse of the minimum odds ratio of the MSM with and
without each observed covariate. The results are summarized in Tables~\ref{tab:ib_gbcsg} and \ref{tab:ib_rhc} for, respectively, the RHC and the GBCSG data.

\begin{table}[ht]
    \centering
    \begin{tabular}{|lr|lr|}
        \hline
        Covariate & $\hat \Gamma$ & Covariate & $\hat \Gamma$ \\ 
        \hline
        immunhx & 1.07 & seps\_Yes & 1.63 \\ 
        sex\_Male & 1.08 & gastr\_Yes & 1.65 \\ 
        income\_\$25-\$50k & 1.11 & psychhx & 1.71 \\ 
        income\_\$11-\$25k & 1.11 & hema\_Yes & 1.76 \\ 
        amihx & 1.12 & sod1 & 1.79 \\ 
        race\_other & 1.13 & resp\_Yes & 1.81 \\ 
        income\_Under \$11k & 1.14 & liverhx & 1.99 \\ 
        race\_white & 1.15 & dnr1\_Yes & 1.99 \\ 
        age & 1.19 & gibledhx & 2.04 \\ 
        chrpulhx & 1.24 & ca\_No & 2.20 \\ 
        ca\_Yes & 1.24 & bili1 & 2.24 \\ 
        scoma1 & 1.24 & neuro\_Yes & 2.25 \\ 
        temp1 & 1.35 & hema1 & 2.31 \\ 
        meta\_Yes & 1.37 & hrt1 & 2.37 \\ 
        ninsclas\_Medicare & 1.38 & wtkilo1 & 2.47 \\ 
        malighx & 1.38 & crea1 & 2.74 \\ 
        cardiohx & 1.38 & ph1 & 2.98 \\ 
        chfhx & 1.39 & pot1 & 3.10 \\ 
        ninsclas\_Medicare \& Medicaid & 1.48 & card\_Yes & 3.59 \\ 
        edu & 1.50 & aps1 & 3.80 \\ 
        transhx & 1.56 & trauma\_Yes & 4.18 \\ 
        wblc1 & 1.57 & resp1 & 4.78 \\ 
        ninsclas\_No insurance & 1.57 & meanbp1 & 4.86 \\ 
        renal\_Yes & 1.58 & ortho\_Yes & 8.38 \\ 
        ninsclas\_Private & 1.59 & paco21 & 8.49 \\ 
        renalhx & 1.60 & alb1 & 19.51 \\ 
        ninsclas\_Private \& Medicare & 1.61 & pafi1 & 35.58 \\ 
        dementhx & 1.63 & & \\
        \hline
    \end{tabular}
    \caption{Informal benchmarking on the RHC data. Variables with an underscore in their name correspond to binary encoded variables.}
    \label{tab:ib_rhc}
\end{table}

\begin{table}[ht]
    \centering
    \begin{tabular}{|lr|}
        \hline
        Covariate & $\hat \Gamma$ \\ 
        \hline
        grade\_2 & 1.26 \\ 
        grade\_3 & 1.37 \\ 
        pgr & 1.50 \\ 
        age & 2.00 \\ 
        size & 2.05 \\ 
        meno & 2.18 \\ 
        nodes & 2.50 \\ 
        er & 2.75 \\ 
        \hline
    \end{tabular}
    \caption{Informal benchmarking on the GBCSG data}
    \label{tab:ib_gbcsg}
\end{table}

In the RHC data, there were 47 covariates that were included, among which some were binary encoded. The variable ``immunhx" (immunosuppression) corresponds to the minimum estimated value of $\Gamma$, $\hat{\Gamma}_\mathrm{min} = 1.07$, whereas ``pafi1" (PaO2/FIO2 ratio) corresponds to the maximum estimated value of $\Gamma$, $\hat{\Gamma}_\mathrm{max} = 35.6$. The median value is 1.63. We already identified a critical value $\Gamma_{c, \mathrm{DVDS}} = 1.2$ at day 30 with our method and a critical value $\Gamma_{c, \mathrm{Lee}} = 1.15$ with the method of \citet{lee2024sensitivity}. 8 covariates have a corresponding $\hat \Gamma$ lower than $\Gamma_{c, \mathrm{Lee}}$ (covariates on income, race, immunosuppression, sex, or myocardial infarction) and 9 covariates have a corresponding $\hat \Gamma$ lower than $\Gamma_{c, \mathrm{DVDS}}$ (the same as before and age). As the majority of observed covariates has a corresponding $\hat \Gamma$ that is greater than the critical values, the results are not robust to unobserved confounders and we cannot conclude on a negative effect of RHC on survival at day 30 because small confounders would be enough to explain away a causal association.

In the GBCSG data, there were 7 covariates that were included, among which some were binary encoded. The variable ‘‘grade\_2" (tumor of grade 2) has the lowest associated $\hat \Gamma$, $\hat \Gamma_\mathrm{min} = 1.26$, and the variable ‘‘er" (estrogen receptor, in fmol/L) has the highest associated $\hat \Gamma$, $\hat \Gamma_\mathrm{max} = 2.75$. The median value is 2.03. The estimated $\hat \Gamma$ associated with tumor size (‘‘size" variable) reported by \citet{lee2024sensitivity} is equal to 1.17 whereas we estimated a value of 2.05 because we used 5-fold cross-fitting. With 95\%-level confidence intervals, we identified a critical value $\Gamma_{c, \mathrm{I}} = 1.09$ at 5 years with our Form~I estimator and a critical value $\Gamma_{c, \mathrm{Lee}} = 1.08$ with the method of \citet{lee2024sensitivity}. These two critical values are lower than 2.05, which means that we cannot conclude about a positive effect of tamoxifen and chemotherapy on recurrence-free survival at 5 years when compared to chemotherapy only.

\paragraph{Survival Function.}

Figures~\ref{fig:rhc_surv_all_results} and \ref{fig:gbc_surv_all_results} were obtained with the same setup as described in Section~\ref{sec:results_real_data}, respectively, on the RHC and GBCSG data.

\begin{figure}[h!]
    \centering
    \subfigure[$\Gamma = 1.15$]{\includegraphics[width=0.49\textwidth]{rhc_diff_surv_plot_lee_dvds_estim_I_and_II_v3_v3.pdf}\label{fig:surv_all_gamma_1_15}}
    \hfill
    \subfigure[$\Gamma = 1.2$]{\includegraphics[width=0.49\textwidth]{rhc_diff_surv_plot_lee_dvds_estim_I_and_II_v4_v4.pdf}\label{fig:surv_all_gamma_1_2}}
    \hfill
    \subfigure[$\Gamma = 1.5$]{\includegraphics[width=0.49\textwidth]{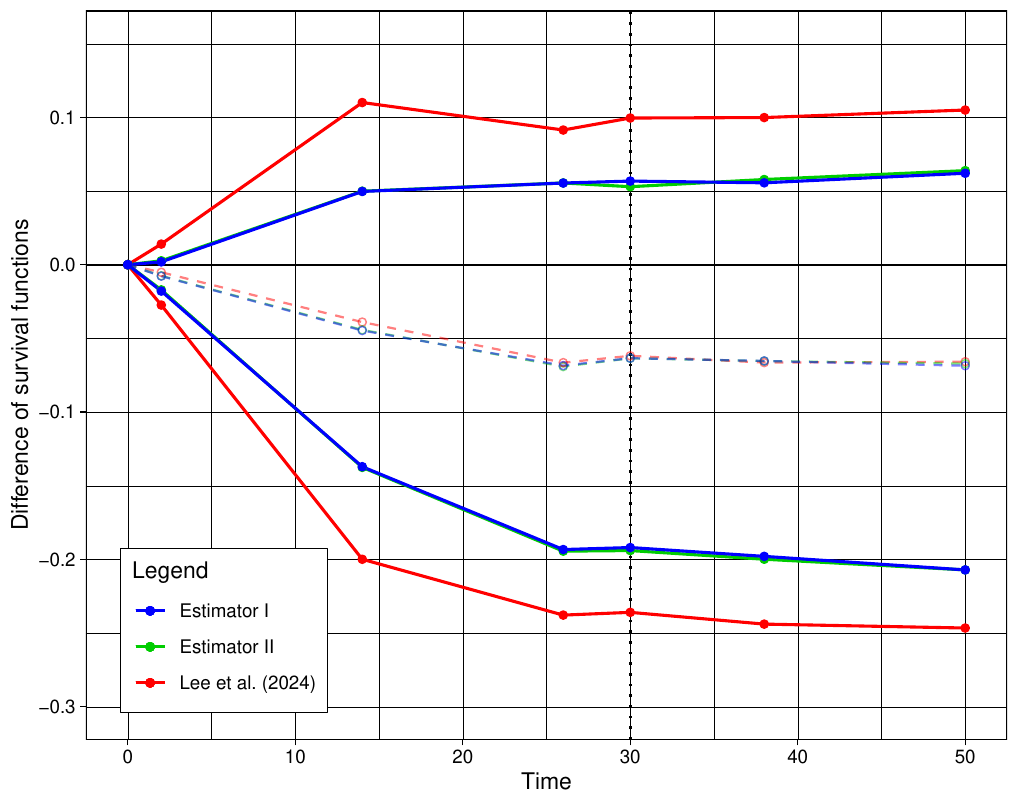}\label{fig:surv_all_gamma_1_5}}
    \caption{Sensitivity analysis for the difference of survival functions for $\Gamma = 1.15$ (\ref{fig:surv_all_gamma_1_15}), $1.2$ (\ref{fig:surv_all_gamma_1_2}), and $1.5$ (\ref{fig:surv_all_gamma_1_5}) on the RHC data. The curves correspond to the lower and upper sensitivity bounds (PEI) obtained with Form~I (in blue), Form~II (in green), and the method from \citet{lee2024sensitivity} (in red). The dotted and transparent lines correspond to the value under ignorability ($\Gamma = 1$) for each method. A vertical dotted black line indicates a time of 30 days. Note that the bounds with Form~II are almost coincident with the bounds with Form~I.}
    \label{fig:rhc_surv_all_results}
\end{figure}

\begin{figure}[h!]
    \centering
    \subfigure[$\Gamma = 1.08$]{\includegraphics[width=0.49\textwidth]{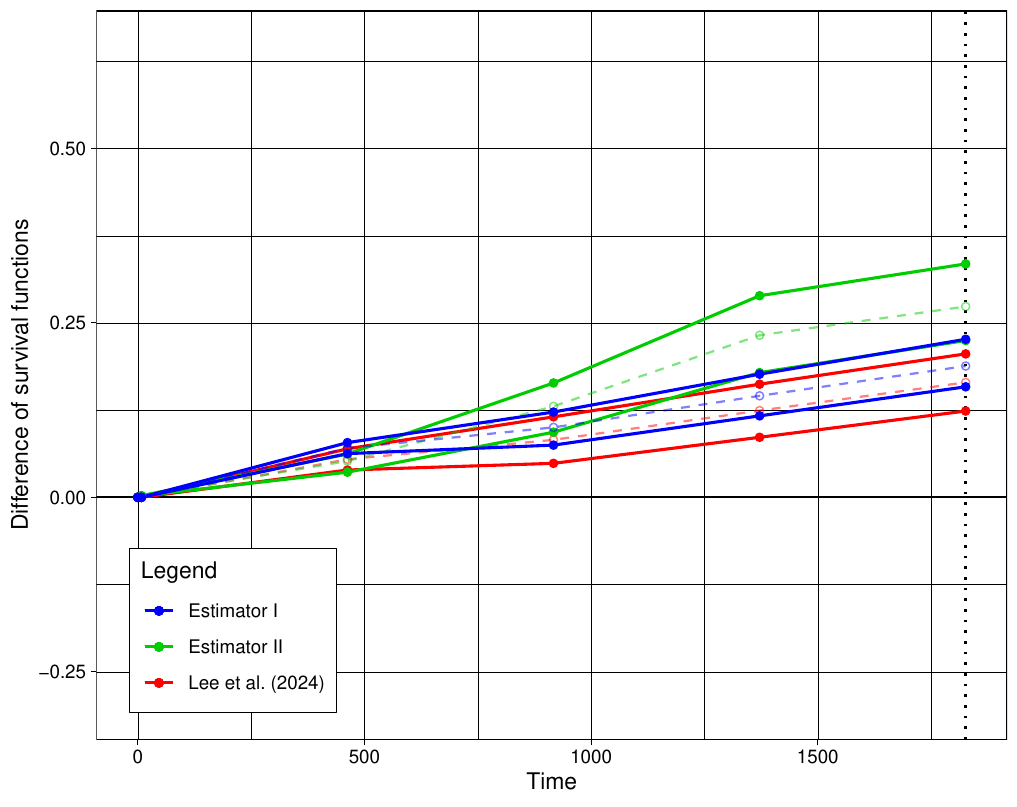}\label{fig:gbc_surv_all_gamma_1_08}}
    \hfill
    \subfigure[$\Gamma = 1.35$]{\includegraphics[width=0.49\textwidth]{gbsg_diff_surv_plot_lee_dvds_estim_I_and_II_v32_v32.pdf}\label{fig:gbc_surv_all_gamma_1_35}}
    \hfill
    \subfigure[$\Gamma = 1.47$]{\includegraphics[width=0.49\textwidth]{gbsg_diff_surv_plot_lee_dvds_estim_I_and_II_v34_v34.pdf}\label{fig:gbc_surv_all_gamma_1_47}}
    \hfill
    \subfigure[$\Gamma = 1.58$]{\includegraphics[width=0.49\textwidth]{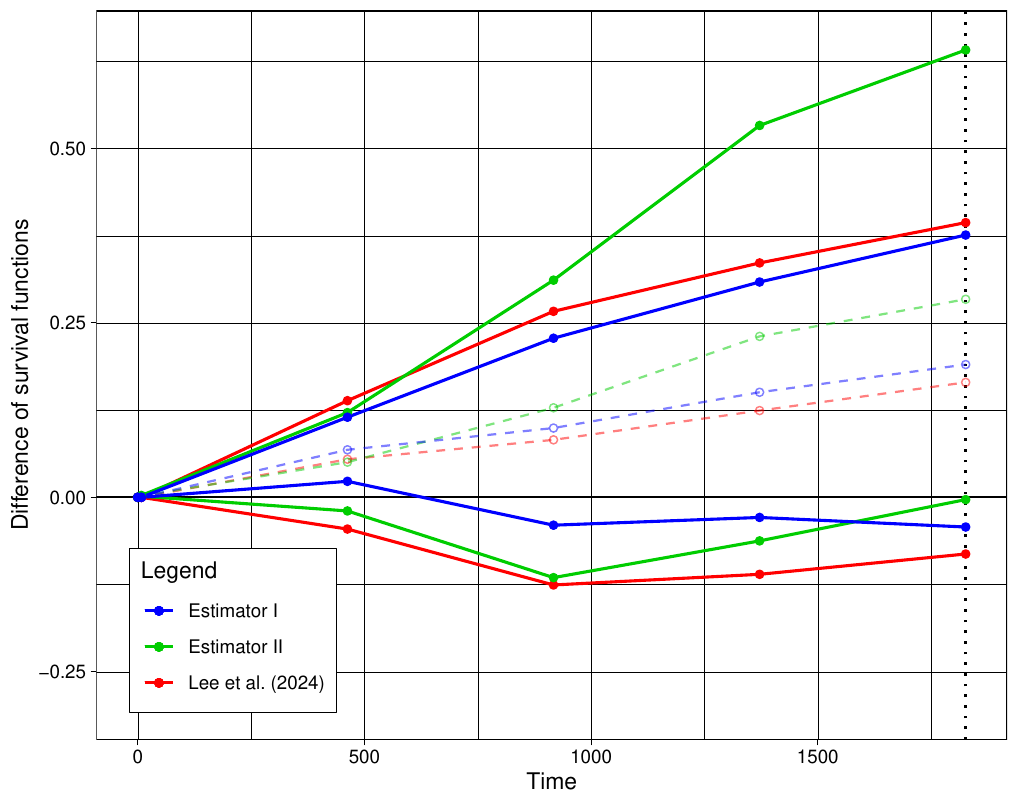}\label{fig:gbc_surv_all_gamma_1_58}}
    \caption{Sensitivity analysis for the difference of survival functions for $\Gamma = 1.08$ (\ref{fig:gbc_surv_all_gamma_1_08}), $\Gamma = 1.35$ (\ref{fig:gbc_surv_all_gamma_1_35}), $1.47$ (\ref{fig:gbc_surv_all_gamma_1_47}), and $1.58$ (\ref{fig:gbc_surv_all_gamma_1_58}) on the GBCSG data. The curves correspond to the lower and upper sensitivity bounds (PEI) obtained with Form~I (in blue), Form~II (in green), and the method from \citet{lee2024sensitivity} (in red). The dotted and transparent lines correspond to the value under ignorability ($\Gamma = 1$) for each method. A vertical black dotted line indicates a time of 1826.25 days (5 years).}
    \label{fig:gbc_surv_all_results}
\end{figure}

\begin{figure}[h!]
    \centering
    \includegraphics[width=0.49\textwidth]{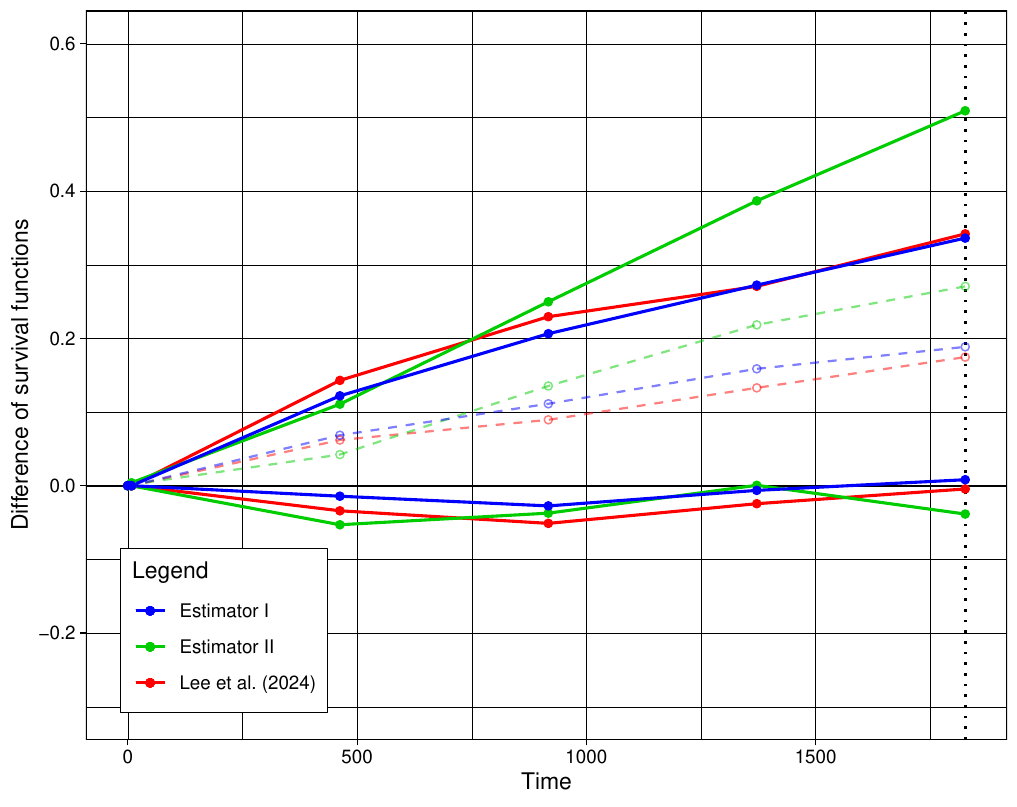}
    \caption{Sensitivity analysis for the difference of survival functions for $\Gamma = 1.08$ on the GBCSG data. The curves correspond to the 95\%-CIs obtained with Form~I (in blue), Form~II (in green), and the method from \citet{lee2024sensitivity} (in red). The dotted and transparent lines correspond to the value under ignorability ($\Gamma = 1$) for each method. A vertical black dotted line indicates a time of 1826.25 days (5 years).}
    \label{fig:gbc_surv_ci_all_results}
\end{figure}

\paragraph{RMST and Difference in RMST.}

\begin{figure}[h!]
    \centering
    \subfigure[$\Gamma = 1.15$]{\includegraphics[width=0.49\textwidth]{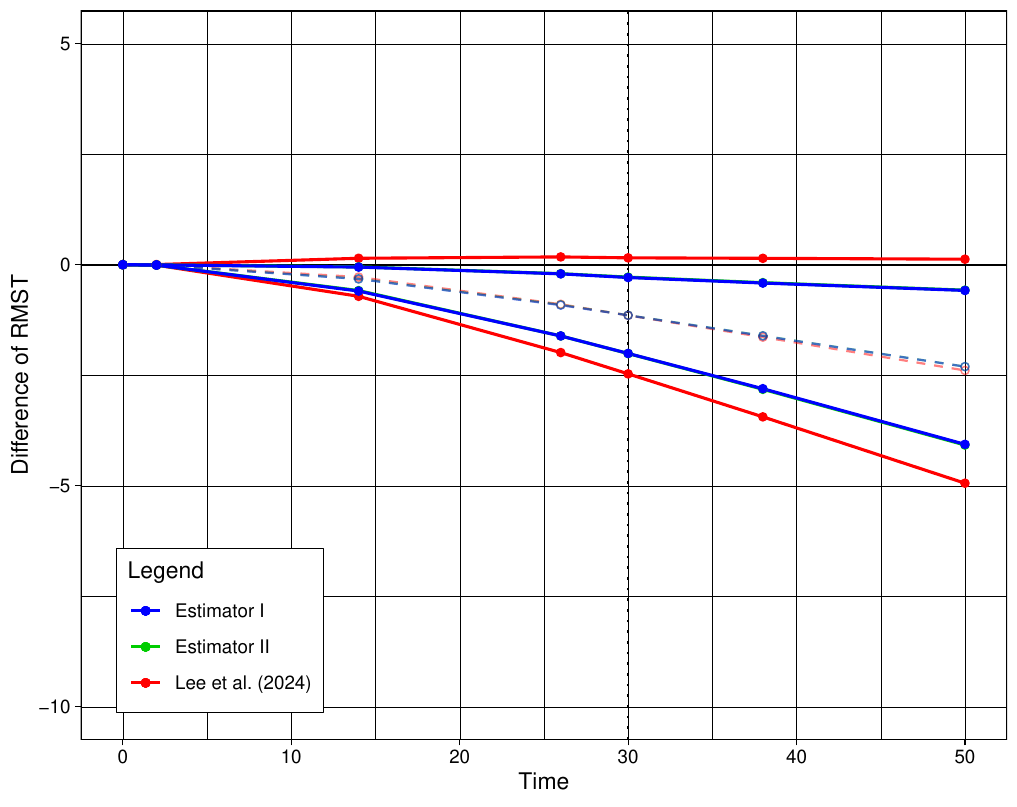}\label{fig:rmst_gamma_1_15}}
    \hfill
    \subfigure[$\Gamma = 1.2$]{\includegraphics[width=0.49\textwidth]{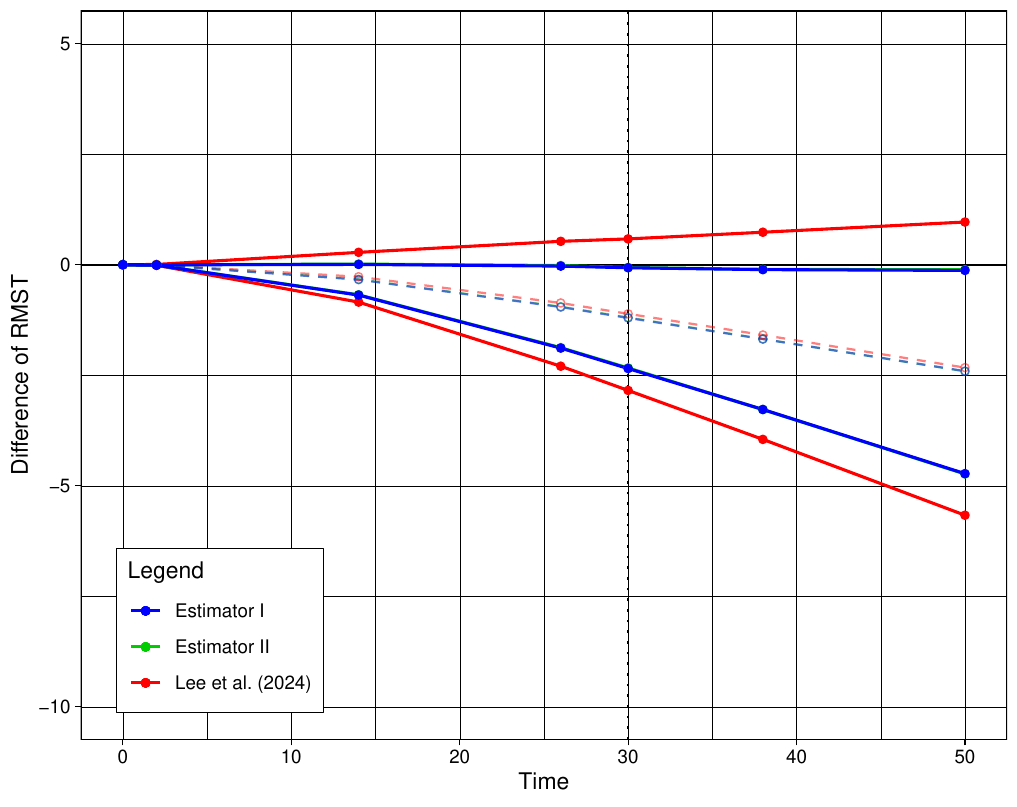}\label{fig:rmst_gamma_1_2}}
    \hfill
    \subfigure[$\Gamma = 1.5$]{\includegraphics[width=0.49\textwidth]{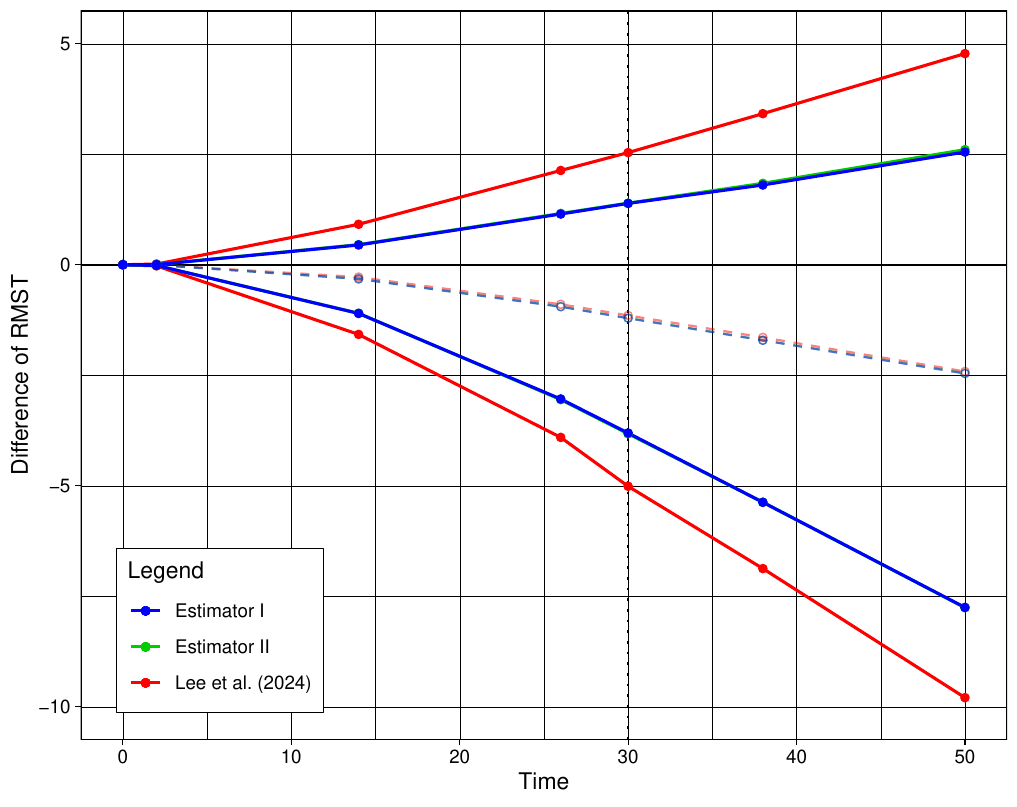}\label{fig:rmst_gamma_1_5}}
    \caption{Sensitivity analysis for the difference in RMST for $\Gamma = 1.15$ (\ref{fig:rmst_gamma_1_15}), $1.2$ (\ref{fig:rmst_gamma_1_2}), and $1.5$ (\ref{fig:rmst_gamma_1_5}) on the RHC data. The curves correspond to the lower and upper sensitivity bounds (PEI) obtained with Form~I (in blue), Form~II (in green), and the method from \citet{lee2024sensitivity} (in red). The dotted and transparent lines correspond to the value under ignorability ($\Gamma = 1$) for each method. A vertical black dotted line indicates a time horizon of 30 days. Note that the bounds with Form~II are almost coincident with the bounds with Form~I.}
    \label{fig:rhc_rmst_results}
\end{figure}

The results of the sensitivity analyses on the difference in RMST in the RHC and GBCSG datasets are given, respectively, in Figures~\ref{fig:rhc_rmst_results} and \ref{fig:gbc_rmst_results}, for a sensitivity parameter $\Gamma$ equal to 1.15, 1.2, and 1.5 (and an additional value of 1.08 for the GBCSG dataset). For the RHC dataset, the PEIs obtained via the DVDS bounds for Form~I (in blue), II (in green) and via the method of \citet{lee2024sensitivity} (in red) were computed for 5 equally spaced values of time between 2 and 50 days, to which we also added 30 days. For the GBCSG dataset, we used 5 equally spaced values between 8 and 1826.25 days. The estimates under ignorability ($\Gamma = 1$) were represented with dotted lines.

For the RHC data, observe that, under ignorability, the difference in RMST is always negative, regardless of the time point between 0 and 50 days. This is consistent with findings from \citet{connors1996effectiveness} on the negative effect of RHC on survival. In particular, if no confounders were hidden, Figure~\ref{fig:rhc_rmst_results} would suggest that RHC reduces survival by approximately 1 day on a time horizon of 30 days. For the GBCSG data, under ignorability, the difference in RMST is almost always positive (except at time 463 days for \citet{lee2024sensitivity}).

For the RHC data, at day 30, we reach similar conclusions to the ones obtained with survival functions. For the GBCSG data, at 5 years, the critical value is 1.2 with Form~II, 1.15 with Form~I, and even lower with the method from \citet{lee2024sensitivity}. Bounds with Form~I are tighter than our comparator, and bounds with Form~II are shifted toward higher values.

\begin{figure}[h!]
    \centering
    \subfigure[$\Gamma = 1.15$]{\includegraphics[width=0.49\textwidth]{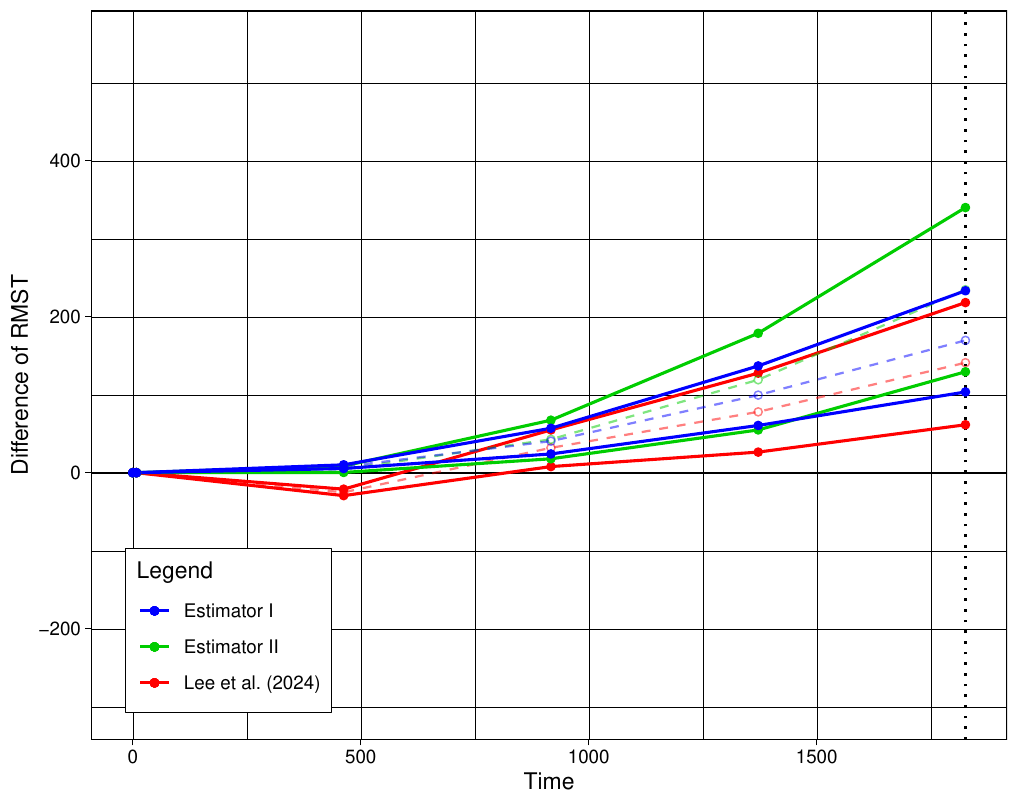}\label{fig:gbc_rmst_gamma_1_15}}
    \hfill
    \subfigure[$\Gamma = 1.2$]{\includegraphics[width=0.49\textwidth]{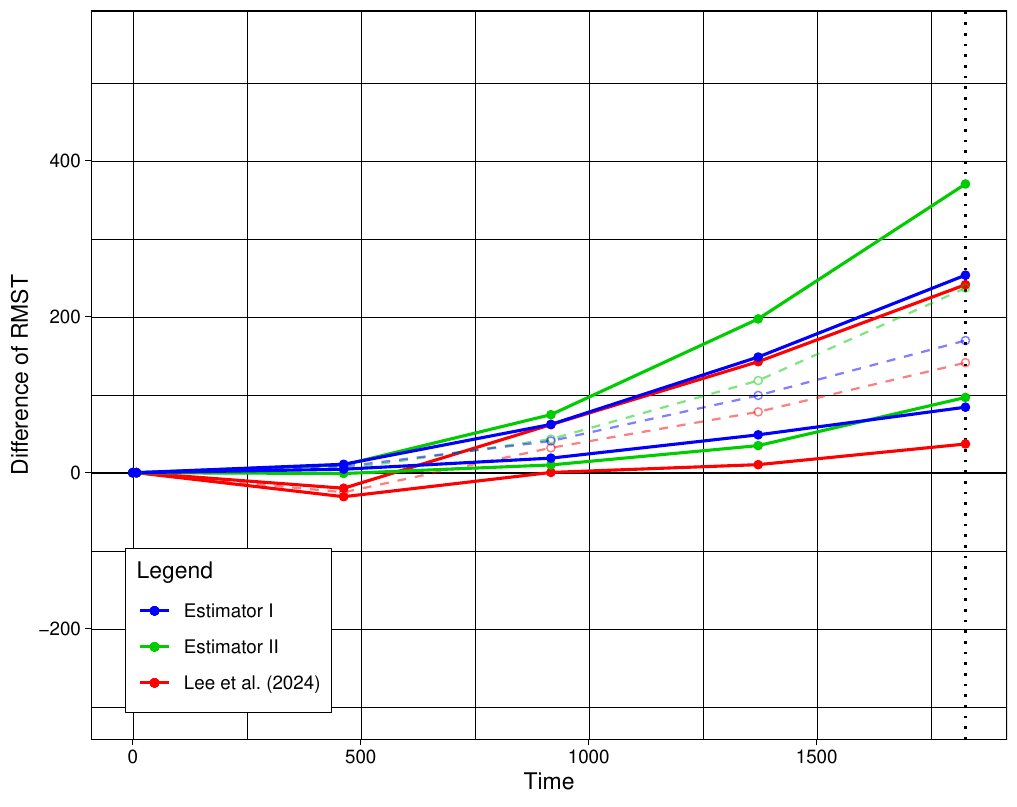}\label{fig:gbc_rmst_gamma_1_2}}
    \hfill
    \subfigure[$\Gamma = 1.5$]{\includegraphics[width=0.49\textwidth]{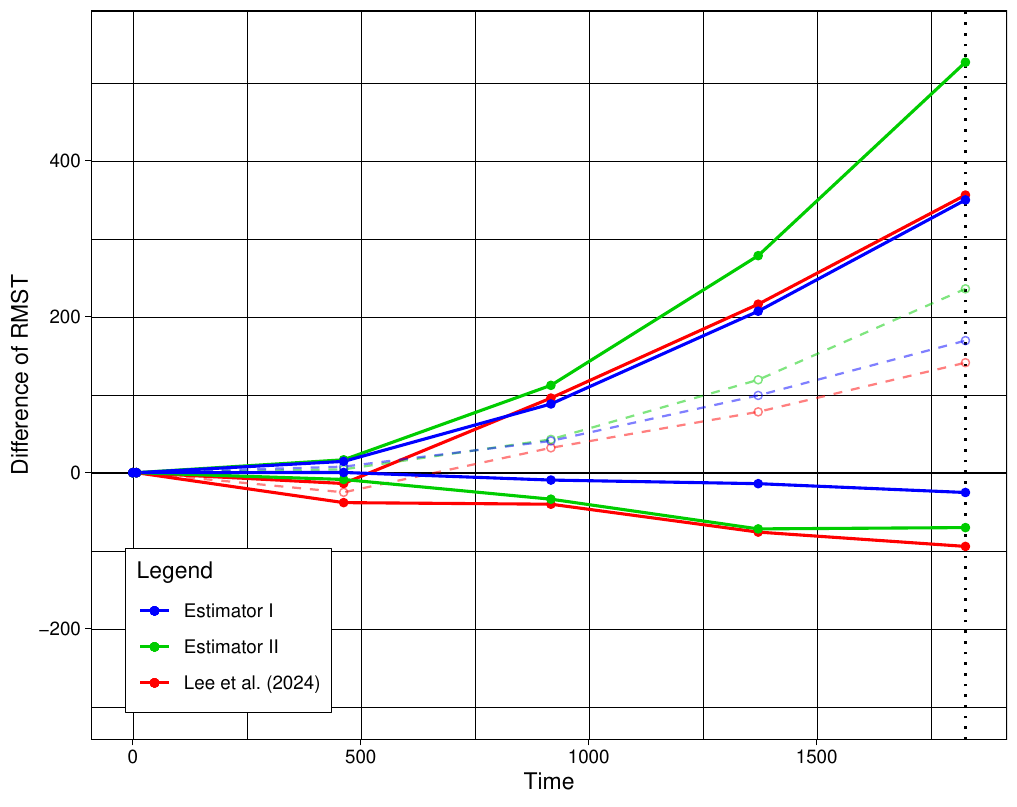}\label{fig:gbc_rmst_gamma_1_5}}
    \caption{Sensitivity analysis for the difference in RMST for $\Gamma = 1.15$ (\ref{fig:gbc_rmst_gamma_1_15}), $1.2$ (\ref{fig:gbc_rmst_gamma_1_2}), and $1.5$ (\ref{fig:gbc_rmst_gamma_1_5}) on the GBCSG data. The curves correspond to the lower and upper sensitivity bounds (PEI) obtained with Form~I (in blue), Form~II (in green), and the method from \citet{lee2024sensitivity} (in red). The dotted and transparent lines correspond to the value under ignorability ($\Gamma = 1$) for each method. A vertical black dotted line indicates a time horizon of 1826.25 days (5 years).}
    \label{fig:gbc_rmst_results}
\end{figure}

\paragraph{Computation time.}

Tables~\ref{tab:exec_times_rhc} and \ref{tab:exec_times_gbcsg} provide the execution times for each experiment on real-world data (RHC and GBCSG data).

With the RHC data ($n = 5735$ and $p_\mathbf{X} = 47$), our method is faster than the method of \citet{lee2024sensitivity} by at least a factor of 1.19 for the difference of survival function, and by at least a factor of 1.87 for the difference in RMST. Moreover, the execution time stays consistent with our method when $\Gamma$ changes whereas it varies more with the method of \citet{lee2024sensitivity} because of the optimization step.

With the GBCSG data ($n = 686$ and $p_\mathbf{X} = 7$), our method is faster than the method of \citet{lee2024sensitivity} for the difference of survival function by at least a factor of 6.90 but is longer for the difference in RMST by at least a factor of 20.0. This longer execution time can be explained by the integration step in our method, because we computed Forms~I and II simultaneously, and because the optimization step of \citet{lee2024sensitivity} is well suited to the RMST when the sample size is not too large (a few hundreds in the GBCSG data vs.\ a few thousands in the RHC data).

\begin{table}[h]
    \centering
    \begin{tabular}{|c|c|c|c|}
        \hline
        Estimand & $\Gamma$ & \makecell{Exec. time for\\ Estim.~I and II} & \makecell{Exec. time for\\ \citet{lee2024sensitivity}} \\
        \hline
        Survival function & 1.15 & 1.26 & 1.50 \\
        Survival function & 1.2 & 1.21 & 1.60 \\
        Survival function & 1.5 & 1.20 & 1.85 \\
        RMST & 1.15 & 0.53 & 0.99 \\
        RMST & 1.2 & 0.51 & 1.15 \\
        RMST & 1.5 & 0.51 & 1.33 \\
        \hline
    \end{tabular}
    \caption{Comparison of execution times (in hours) on the RHC data for the difference in survival functions and the difference in RMST.}
    \label{tab:exec_times_rhc}
\end{table}

\begin{table}[h]
    \centering
    \begin{tabular}{|c|c|c|c|}
        \hline
        Estimand & $\Gamma$ & \makecell{Exec. time for\\ Estim.~I and II} & \makecell{Exec. time for\\ \citet{lee2024sensitivity}} \\
        \hline
        Survival function & 1.15 & 25.9 s & 2.98 min \\
        Survival function & 1.2 & 22.2 s & 3.24 min \\
        Survival function & 1.5 & 20.9 s & 3.68 min \\
        RMST & 1.15 & 11.2 min & 27.4 s \\
        RMST & 1.2 & 11.1 min & 27.3 s \\
        RMST & 1.5 & 11.0 min & 33.2 s \\
        \hline
    \end{tabular}
    \caption{Comparison of execution times on the GBCSG data for the difference in survival functions and the difference in RMST.}
    \label{tab:exec_times_gbcsg}
\end{table}

\end{document}